%% file: state_preparation.tex
\documentclass[a4paper,11pt]{article}

\usepackage{geometry}
\newgeometry{left=1in, right=1in, top=1in, bottom=1in}

\usepackage[T1]{fontenc}
\usepackage[utf8]{inputenc}
\usepackage{authblk}
\usepackage{commath}
\usepackage{braket}
\usepackage{multicol}
\usepackage{enumerate}
\usepackage{mathrsfs}
\usepackage{amsmath,amsthm,amssymb,xspace}
\usepackage{cite}
\usepackage{txfonts}
\usepackage{graphics}
\usepackage{float}
\usepackage{epstopdf}
\usepackage{epsfig}
\usepackage{caption}
\usepackage[labelformat=simple]{subcaption}

\usepackage{url}
\usepackage{pgfplots}
\usepgfplotslibrary{polar}
\usepgflibrary{shapes.geometric}
\usetikzlibrary{calc}

\usepackage{booktabs}
\usepackage{commath}
\usepackage{tikz}
\usepackage{tikz-3dplot}
\usepackage[qm]{qcircuit}
\usepackage{makecell}

\usepackage[ruled,vlined,linesnumbered]{algorithm2e}
\usepackage{multirow}

\newcommand{\B}{\mbox{$\{0,1\}$}}
\newcommand{\Bn}{\mbox{$\{0,1\}^n$}}

\newcommand\defeq{\stackrel{\mathrm{\scriptsize def}}{=}}

\newtheorem{theorem}{Theorem}
\newtheorem{corollary}[theorem]{Corollary}
\newtheorem{lemma}[theorem]{Lemma}

\newtheorem{definition}[theorem]{Definition}

\definecolor{lgray}{gray}{0.9}
\definecolor{llgray}{gray}{0.97}

\title{Asymptotically Optimal Circuit Depth for Quantum State Preparation  and General Unitary Synthesis}

\author[*1,2]{Xiaoming Sun}
\author[*1,2]{Guojing Tian}
\author[1,2]{Shuai Yang\footnote{Email:\{sunxiaoming, tianguojing, yangshuai21b\}@ict.ac.cn}}
\author[$\dagger$3]{Pei Yuan}
\author[3]{Shengyu Zhang \footnote{Email: \{peiyuan, shengyzhang\}@tencent.com}}
\affil[1]{State Key Lab of Processors, Institute of Computing Technology, Chinese Academy of Sciences, Beijing 100190, China }
\affil[2]{School of Computer Science and Technology, University of Chinese Academy of Sciences, Beijing 100049, China}
\affil[3]{Tencent Quantum Laboratory, Tencent, Shenzhen, Guangdong 518057, China}

\date{} 

\pgfplotsset{compat=1.17}
\begin{document}

\maketitle
\begin{abstract}
    The Quantum State Preparation problem aims to prepare an $n$-qubit quantum state $|\psi_v\rangle =\sum_{k=0}^{2^n-1}v_k|k\rangle$ from {the} initial state $|0\rangle^{\otimes n}$, for a given unit vector $v=(v_0,v_1,v_2,\ldots,v_{2^n-1})^T\in \mathbb{C}^{2^n}$ with $\|v\|_2 = 1$. The problem is of fundamental importance in quantum algorithm design, Hamiltonian simulation and quantum machine learning, yet its circuit depth complexity remains open when ancillary qubits are available. In this paper, we study quantum circuits when there are $m$ ancillary qubits available. We construct, for any $m$, circuits that can prepare $|\psi_v\rangle$ in depth $\tilde O\big(\frac{2^n}{m+n}+n\big)$ and size $O(2^n)$, achieving the optimal value for both measures simultaneously. 
    These results also imply a depth complexity of $\Theta\big(\frac{4^n}{m+n}\big)$ for quantum circuits implementing a general $n$-qubit unitary for any $m \le O(2^n/n)$ number of ancillary qubits.
    This resolves the depth complexity for circuits without ancillary qubits. And for circuits with exponentially many  ancillary qubits, our result quadratically improves the currently best upper bound of {$O(4^n)$} to $\tilde \Theta(2^n)$. 
    
    Our circuits are deterministic, prepare the state and carry out the unitary precisely, utilize the ancillary qubits tightly and the depths are optimal in a wide parameter regime. The results can be viewed as (optimal) time-space trade-off bounds, which are not only theoretically interesting, but also practically relevant in the current trend that the number of qubits starts to take off, by showing a way to use a large number of qubits to compensate the short qubit lifetime.
\end{abstract}

\section{Introduction}
\label{sec:introduction}
\input{QSP_Introduction}

\section{Preliminaries}
\label{sec:preliminaries}
\input{QSP_preliminaries}

\section{Quantum state preparation with $O(2^n/n^2)$ ancillary qubits}
\label{sec:diag_matrix}
\input{QSP_diagonal_matrices}

\section{Diagonal unitary implementation with ancillary qubits}
\label{sec:QSP_withancilla}
\input{QSP_withancilla}

\section{Diagonal unitary implementation without ancillary qubits}
\label{sec:QSP_withoutancilla}
\input{QSP_withoutancilla}

\section{Quantum state preparation with $\Omega(2^n/n^2)$ ancillary qubits}
\label{sec:QSP_withmoreancilla}
\input{QSP_withmoreancilla}

\section{Extensions and implications}
\label{sec:extensions}
\input{QSP_extensions}

\section{Conclusion}
\label{sec:conclusions}
In this paper, we have shown that an arbitrary $n$-qubit quantum state can be prepared by a quantum circuit consisting of single-qubit gates and CNOT gates with $m=O(2^n)$ ancillary qubits, of depth $O\big(n\log n+\frac{2^n}{n+m}\big)$ and size $O(2^n)$. 
The bound is improved to $O(n)$ if we have more ancillary qubits, and all these bounds are tight (up to a logarithmic factor in a small range of $m$). 
These results can be applied to reduce the depth of the circuit of general unitary to $O\big(n2^n+\frac{4^n}{m+n}\big)$ with $m$ ancillary qubits, which is optimal when $m = O(2^n/n)$. The results can be extended to approximate state preparation by circuit using the Clifford+T gate set.

Many questions are left open for future studies. An immediate one is to close the gap for unitary synthesis for large $m$ in Corollary \ref{coro:unitary}. One can also put more practical restrictions into consideration. For instance, we assume that two-qubit gates can be applied on any two qubits. Though this all-to-all connection is indeed the case for certain quantum computer implementations (such qubits made of trapped ions), some others (such as superconducting qubits) can only support nearest neighbor interactions, and it is interesting to study QSP for that case. Another direction is to take various noises into account, and see how much that affects the complexity. We call for more studies of state preparation and circuit synthesis, and hope that methods and techniques developed in this paper can be used to design efficient circuits in those extended models. 

\bibliographystyle{unsrt}
\bibliography{state_preparation}
\appendix
\input{QSP_append}

\end{document}

%% file: QSP_Introduction.tex
Quantum computers provide a great potential of solving certain important information processing tasks that are  believed to be intractable  for classical computers. In recent years, quantum machine learning \cite{biamonte2017quantum} and Hamiltonian simulation \cite{berry2015simulating,low2017optimal,low2019hamiltonian,berry2015hamiltonian} have also been extensively investigated, including quantum principal component analysis (QPCA) \cite{lloyd2014quantum}, quantum recommendation systems \cite{kerenidis2017quantum}, quantum singular value decomposition \cite{rebentrost2018quantum}, quantum linear system algorithm \cite{harrow2009quantum,wossnig2018quantum}, quantum clustering \cite{kerenidis2018q,kerenidis2020quantum} and quantum support vector machine (QSVM) \cite{rebentrost2014quantum}.
One of the challenges to fully exploit quantum algorithms for these tasks, however, is to efficiently prepare a starting state\footnote{These starting states (for example, those in \cite{harrow2009quantum,wossnig2018quantum}) are very generic. Indeed, the lower bound argument in our later Theorem \ref{thm:lowerbound_QSP} applies to the generation of these states as well.}, which is usually the first step of those algorithms. 
This raises the fundamental question about the complexity of the quantum state preparation (QSP) problem.

The QSP problem can be formulated as follows. Suppose we have a vector $v=(v_0,v_1,v_2,\ldots,v_{2^n-1})^T\in \mathbb{C}^{2^n}$ with unit $\ell_2$-norm, i.e. $\sqrt{\sum_{k=0}^{2^n-1}|v_k|^2}=1$. The task is to generate a corresponding $n$-qubit quantum state 
\[|\psi_v\rangle=\sum_{k=0}^{2^n-1}v_k|k\rangle,\]
by a quantum circuit from the initial state $|0\rangle^{\otimes n}$,  where $\{|k\rangle: k=0, 1, \ldots, 2^n-1\}$ is the computational basis of the quantum system. 


Different cost measures can be studied for quantum circuits: Size, depth, and number of qubits are among the most prominent ones. For a quantum circuit, the depth  corresponds to the time for executing the quantum circuit, and the number of qubits used to its space cost. 
Apart from minimizing each cost measure individually, it is of particular interest to study a time-space trade-off for quantum circuits. The reason is that in the past decade, we have witnessed a rapid development in qubit number and in qubit lifetime\footnote{Take superconducting qubits, for example, the qubit number jumped from 5 in 2014 to 127 in 2021 \cite{barends2014superconducting,kelly2015state,song201710,arute2019quantum,IBM65qubit,chow2021ibm} 
.}, but it seems hard to significantly improve \textit{both} on the same chip. Looking into the near future, big players such as IBM and Google announced their roadmaps of designing and manufacturing quantum chips with about 1,000,000 superconducting qubits by 2026 and 2029, respectively, rocketing from 50-100 today \cite{IBMroadmap, Googleroadmap}.
This raises a natural question for quantum algorithm design: How to utilize the fast-growing number of qubits to overcome the relatively limited decoherence time? This seems especially relevant in the near future when we have $10^4- 10^5$ qubits, which are expected to run certain quantum simulation algorithms for chemistry problems, but are not sufficient for the full quantum error correction to fight the decoherence.  
Or put in a computational complexity language, how to efficiently trade space for time in a quantum circuit? In this paper, we will address this question in the fundamental tasks of quantum state preparation and general unitary circuit synthesis.

Let us first fix a proper circuit model. If we aim to generate the target state $|\psi_v\rangle$ or perform the target unitary precisely, then a finite universal gate set is not enough. A natural choice is the set of circuits that consist of arbitrary single-qubit gates and CNOT gates
, which is expressive enough to generate arbitrary states $|\psi_v\rangle$ precisely  with certainty. 
We will study the optimal depth for this class of circuits\footnote{Since two-qubit gates are usually harder to implement, one may also like to consider CNOT depth, the number of layers with at least one CNOT gate. But note that between two CNOT layers, consecutive single-qubit gates on the same qubit can be compressed to one single-qubit gate, and single-qubit gates on different qubits can be  paralleled to within one layer, we can always assume that the circuit has alternative single-qubit gate layers and CNOT gate layers. Therefore the circuit depth is at most twice of the CNOT depth, making the two measures the same up to a factor of 2.}.

The study of QSP dates back to 2002, when Grover and Rudolph gave an algorithm for QSP for the special case of efficiently integrable probability density functions \cite{grover2002creating}. {Their circuit has $n$ stages, and each stage $j$ has $2^{j-1}$ layers, with each layer being a rotation on last qubit conditioned on the first $j-1$ qubits being certain computational basis state. This type of multiple-controlled $(2\times 2)$-unitary can be implemented in depth $O(n)$ without ancillary qubit\footnote{The standard method \cite{nielsen2002quantum} gives a depth upper bound of $O(n^2)$ without ancillary qubit and $O(n)$ with sufficiently many ancillary qubits. The first bound can be improved to $O(n)$ by the method in \cite{multi-controlled-gate}.}, yielding a depth upper bound of $O(n2^n)$ for the QSP problem.} 
In \cite{bergholm2005quantum}, Bergholm \textit{et al.} gave an upper bound of $2^{n+1}-2n-2$ for the \emph{number} of CNOT gates, with depth also of order {$O(2^n)$}.  The number of CNOT gates is improved to $\frac{23}{24} 2^n - 2^{\frac{n}{2}+1} + \frac{5}{3}$ for even $n$, and $\frac{115}{96} 2^n $ for odd $n$ by Plesch and Brukner  \cite{plesch2011quantum}, based on a universal gate decomposition technique in \cite{mottonen2005decompositions}. The same paper \cite{bergholm2005quantum} also gives a depth upper bound of $\frac{23}{48}2^n$ for even $n$ and $\frac{115}{192}2^n$ for odd $n$. All these results are about the exact quantum state preparation without ancillary qubits.

With ancillary qubits, Zhang \textit{et al.} \cite{zhang2021low} proposed circuits which involve measurements and can generate the target state in $O(n^2)$ depth but only with certain success probability, which is at least $\Omega(1 / (\max_i |v_i|^2 2^n))$, but in the worst case can be an exponentially small order of $O(1/2^n)$.  In addition, they need $O(4^n)$ ancillary qubits to achieve this depth. 
In a different paper \cite{zhang2021parallel}, the authors showed that for $\epsilon\le 2^{-\Omega(n)}$, an $n$-qubit quantum state $|\psi'_v\rangle$ can be implemented by an $O(n^3)$-depth quantum circuit with sufficiently many ancillary qubits\footnote{No explicit bound on the ancillary qubits is given.}, where $\||\psi'_v\rangle-|\psi_v\rangle\|\le \epsilon$. 
Though QSP is only used as a tool for their main topic of parallel quantum walk, their concluding section did call for studies on the trade-off between the circuit depth and the number of ancillary qubits for better parallel quantum algorithms.
{Another related study is \cite{johri2021nearest}, which considers to prepare a state not in the binary encoding $\sum_{k=0}^{2^n-1} v_k |k\rangle$, but in the \textit{unary} encoding $\sum_{k=0}^{2^n-1} v_k |e_k\rangle$, where $e_i\in \{0,1\}^{2^n}$ is the vector with the $k$-th bit being 1 and all other bits being 0. The paper shows that the unary encoding QSP can be carried out by a quantum circuit of depth $O(n)$ and size $O(2^n)$. Note that the unary encoding itself takes $2^n$ qubits, as opposed to $n$ qubits in the binary encoding. The binary encoding is the most efficient one in terms of the number of qubits needed for the resulting state, and indeed in most quantum machine learning tasks the quantum speedup depends crucially on this encoding efficiency at the first place \cite{harrow2009quantum,kerenidis2017quantum,kerenidis2020classification,dervovic2018quantum,zhao2021smooth,larose2020robust}. 
In \cite{johri2021nearest} the authors also extended this by using a $d$-dimensional tensor $(k_1, k_2, \ldots, k_d)$ to encode $k$, which needs $d 2^{n/d}$ qubits to encode and a circuit of depth $O(\frac{n}{d} 2^{n-n/d})$ to prepare. When $d = n$ the encoding coincides with the binary encoding, but their depth bound is $O(2^n)$, which is not optimal. 
} 

In this paper, we tightly characterize the depth and size complexities of the quantum state preparation problem by constructing optimal quantum circuits. Our circuits generate the target state precisely, with certainty, and use an optimal number of ancillary qubits. 
We present our results on QSP first, where a general number $m$ of ancillary qubits are available.

\begin{theorem} \label{thm:QSP_anci}
For any $m \ge 2n$, any $n$-qubit quantum state $\ket{\psi_v}$ can be generated by a circuit with $m$ ancillary qubits, using single-qubit gates and CNOT gates, of {size $O(2^n)$ and} depth 
\[\left\{\begin{array}{ll}
   O\big(\frac{2^n}{m+n}\big),  & \text{if~} m\in [2n, O(\frac{2^n}{n\log n})],  \\
   O\left(n\log n\right),  & \text{if~} m\in [\omega(\frac{2^n}{n\log n}),o(2^n)],\\
   O\left(n\right),  & \text{if~} m = \Omega(2^n).
\end{array}\right.\]

\end{theorem}

These depth bounds improve the depth of $O(2^n)$ in \cite{bergholm2005quantum,plesch2011quantum} by a factor of $m$ for any $m\in [2n, O(\frac{2^n}{n\log n})]$, and the result shows that more ancillary qubits can indeed provide more help in shortening the depth for QSP. Compared with the result in \cite{zhang2021low} which needs $O(4^n)$ ancillary qubits to achieve depth $O(n^2)$, ours needs only $m=O(2^n/n^2)$ qubits to reach the same depth. In addition, our circuit is deterministic and generates the state with certainty, and the only two-qubit gates used are the CNOT gates.  

The above construction needs at least $2n$ ancillary qubits. Next we show an optimal depth construction of circuits without ancillary qubits. 
\begin{theorem}
Any $n$-qubit quantum state $\ket{\psi_v}$ can be generated by a quantum circuit, using single-qubit gates and CNOT gates, of depth $O(2^n/n)$ {and size $O(2^n)$,} without using ancillary qubits.
 \label{thm:QSP_noanci}
\end{theorem}

These two theorems combined give asymptotically optimal bounds for depth and size complexity. Indeed, a lower bound of $\Omega(2^n)$ for size is known \cite{plesch2011quantum}, and the same paper also presents a depth lower bound of $\Omega(2^n/n)$ for quantum circuits without ancillary qubits. This can be extended to a lower bound of $\Omega\big(\frac{2^n}{n+m}\big)$ for circuits with $m$ ancillary qubits. This bound deteriorates to 0 as $m$ grows to infinity. In \cite{aharonov2018quantum}, the authors gave a depth lower bound of $\Omega(\log n)$ for circuit with arbitrarily many ancillary qubits. We note that it can be improved to $\Omega(n)$ for any $m$, as stated in the next theorem as well as independently discovered in \cite{zhang2021low}.



\begin{theorem}
\label{thm:lowerbound_QSP}
Given $m$ ancillary qubits, there exist $n$-qubit quantum states which can only be prepared by quantum circuits of depth at least $\Omega\big(\max\big\{n,\frac{2^n}{m+n}\big\}\big)$, for circuits using arbitrary single-qubit and 2-qubit gates.
\end{theorem}
{The proof of Theorem \ref{thm:lowerbound_QSP} is shown in Appendix \ref{sec:QSP_lowerbound}.}

Putting the above results together, we can tightly characterize the size and depth complexity of QSP, except for a logarithmic factor gap over a small parameter regime for $m$. It is interesting to note that our circuits achieve the optimal depth and size simultaneously. Our results are summarized in the next Corollary \ref{corol:tight_depth} and illustrated in Figure \ref{fig:result}.

\begin{corollary}\label{corol:tight_depth}
For a circuit preparing an $n$-qubit quantum state with $m$ ancillary qubits, the minimum size is $\Theta(2^n)$, and the minimum depth $D_{\textsc{QSP}}(n,m)$ for different ranges of $m$ are characterized as follows.
\[
\begin{cases}
   \Theta\big(\frac{2^n}{m+n}\big), & \text{if } m=O\big(\frac{2^n}{n\log n}\big), \\
    \left[ \Omega(n),  O(n\log n)\right], & \text{if }m \in [\omega\big(\frac{2^n}{n\log n}\big),o\left(2^n\right)],\\
    \Theta(n), & \text{if }m = \Omega\left(2^n\right).\\
\end{cases}
\]
\end{corollary}

\begin{figure}[!hbt]
\centering
\begin{tikzpicture}
\tikzstyle{every node}=[font=\small,scale=0.8]
\draw[thick,dotted,white,fill=blue!10] (0,0) -- (0,3) -- (3.3,3) -- (3.3,0) -- cycle;
\draw[thick,dotted,white,fill=green!10] (3.3,0) -- (3.3,3) -- (4.7,3) -- (4.7,0) -- cycle;
\draw[thick,dotted,white,fill=blue!10] (4.7,0) -- (4.7,3) -- (7.5,3) -- (7.5,0) -- cycle;

\draw[->] (0,0) -- (8,0);
\draw (6.8,0)  node[anchor=north] { $\sharp$ ancillary qubits $\slash~ O(\cdot)$};
\draw	(0,0) node[anchor=north] {$0$}
		(3.3,0) node[anchor=north] {$\frac{2^n}{n\log n}$}
		(4,0) node[anchor=north] {$\frac{2^n}{n}$}
		(4.7,0) node[anchor=north] { $2^n$}
		(0,0.9) node[anchor=east] {$n$}
		(0,3) node[anchor=east] {$\frac{2^n}{n}$}
		(0,1.2) node[anchor=east] {$n \log n$}
		(0,4) node[anchor=east] {$2^n$}
		(0,5) node[anchor=east] {$n2^n$};
		
\draw	(1.8,2.7) node{{optimal depth}}
        (1.8,2.2) node{{$\Theta(\frac{2^n}{n+m})$}}
		(4,2.7) node{{depth}}
		(4,2.3) node{{$O(n\log n)$}}
		(4,1.9) node{{$\Omega(n)$}}
		(6.3,2.7) node{{optimal depth}}
		(6.3,2.2) node{{$\Theta(n)$}};

\draw[thick,magenta] plot[smooth, domain = 0:4] (\x, {6/(\x + 2)});
\draw[thick,dashed] plot[smooth, domain = 0:3.3] (\x, {6/(\x + 2)});


\draw[->] (0,0) -- (0,6);
\draw (0,6.2) node {depth $\slash~ O(\cdot)$};
\draw[thick,magenta] (4,1) -- (7.5,1);

\draw[thick,dashed] (3.3,1.132)--(4.35,1.132);
\draw[thick,dashed] plot[smooth, domain = 4.35:4.7] (\x, {3.002/(\x-1.698)});
\draw[thick,dashed] (4.7,1)--(7.5,1);
\draw[thick,dashed] (1.6,4)--(1.9,4);
\draw[thick,magenta] (1.6,3.5)--(1.9,3.5);

\draw(0,4) node[fill,white,draw=black,circle,scale = 0.6]{}
		(0,5) node[fill,black,circle,scale = 0.6]{}
		(4.7,1) node[fill,gray,circle,scale = 0.6]{}
		(1.8,5.5) node[fill,white,draw=black,circle,scale = 0.6]{}
		(1.8,5) node[fill,black,circle,scale = 0.6]{}
		(1.8,4.5) node[fill,gray,circle,scale = 0.6]{};

\draw[thick,dotted,gray] (0,1.132) -- (3.3,1.132);
\draw[thick,dotted,gray] (0,1) -- (4,1) -- (4,0);

\draw	(2,5.5) node[anchor=west]{  Upper bound in \cite{plesch2011quantum,bergholm2005quantum,mottonen2005decompositions}}
		(2,5) node[anchor=west]{ Upper bound in \cite{grover2002creating}}
		(2,4) node[anchor=west]{  Our upper bound}
		(2,3.5) node[anchor=west]{  Our lower bound}
		(2,4.5) node[anchor=west]{ Upper bound in \cite{johri2021nearest}};
\draw[thick] (1.4,3.1)--(1.4,5.9)--(7.5,5.9)--(7.5,3.1)--(1.4,3.1);
\end{tikzpicture}
\caption{Circuit depth upper and lower bound for $n$-qubit quantum state preparation. $m$ denote the number of ancillary qubits. If $m=O(\frac{2^n}{n\log n})$ and $\Omega(2^n)$, our circuit depths are $\Theta\big(\frac{2^n}{n+m}\big)$ and $\Theta(n)$, which are asymptotically optimal. When $m\in[\omega(\frac{2^n}{n\log n}),o(2^n)]$, the gap between our depth upper and lower bound is at most logarithmic.
}
\label{fig:result}
\end{figure}
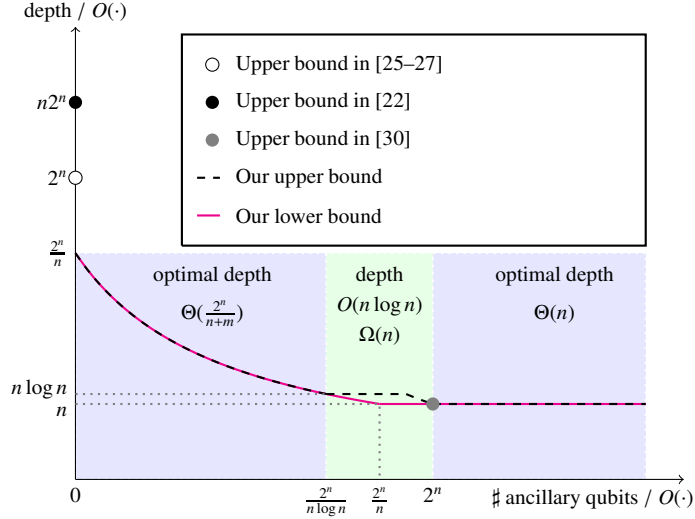
\medskip
Now we give two applications of the result, the first of which is general unitary synthesis. Given a unitary matrix, a fundamental question is to find a circuit implementing it in optimal depth or size. Previous studies on this problem focus on circuits without ancillary qubits. Barenco \textit{et al.} \cite{barenco1995elementary} gave an upper bound $O(n^34^n)$ for the number of CNOT gates for arbitrary $n$-qubit unitary matrix. Knill \cite{knill1995approximation} improved the upper bound to $O(n4^n)$.  Vartiainen \textit{et al.} \cite{vartiainen2004efficient} constructed a quantum circuit for an $n$-qubit unitary matrix with $O(4^n)$ CNOT gates. Mottonen and Vartiainen \cite{mottonen2005decompositions} designed a quantum circuit of depth $O(4^n)$ using $\frac{23}{48}4^n$ CNOT gates. The best known lower bound for \textit{number} of CNOT gates is $\left\lceil\frac{1}{4}(4^n-3n-1)\right\rceil$ \cite{shende2004minimal}, which also implies a depth lower bound of $\Omega(4^n/n)$. In a nutshell, the previous work put the optimal depth to within the range of $[{\Omega}(4^n/n), O(4^n)]$ for general $n$-qubit circuit compression without ancillary qubits.

Our results on QSP can be applied to close this gap, by showing a circuit of depth $O(4^n/n)$. And this is actually a special case of the next theorem which handles a general number $m$ of ancillary qubits.
\begin{theorem}
Any unitary matrix $U\in\mathbb{C}^{2^n\times 2^n} $ can be implemented by a quantum circuit of {size $O(4^n)$ and} depth $O\big(n2^n+\frac{4^n}{m+n}\big)$ with $m \le 2^n$ ancillary qubits.
 \label{thm:unitary_anci}
\end{theorem}

The second application of our QSP result is approximate QSP, for which one can obtain the following bound for circuit with a finite set of gates such as $\{CNOT, H, S, T\}$ using a variant of the Solovay–Kitaev theorem.
\begin{corollary}\label{coro:approx_QSP}
    For any $n$-qubit target state $\ket{\psi_v}$, one can prepare a state $\ket{\psi'_v}$ which is $\epsilon$-close to $\ket{\psi_v}$ in $\ell_2$-distance, by a circuit consisting of $\{CNOT,H,S,T\}$ gates of depth 
        \[\left\{\begin{array}{ll}
        O\big(\frac{2^n\log(2^n/\epsilon)}{m+n}\big), &  \text{if~}m=O\big(\frac{2^n}{n\log n}\big),\\
         O(n\log n\log(2^n/\epsilon)),& \text{if~}m\in [\omega\big(\frac{2^n}{n\log n}\big),o(2^n)],\\
         O(n \log(2^n/\epsilon)),& \text{if~}m=\Omega(2^n),\\
    \end{array}\right.\]
    using $m$ ancillary qubits.
\end{corollary}


%


\paragraph{Proof techniques} We give a brief account of the proof techniques used in our circuit constructions. We first reduce the problem to implementing diagonal unitary matrices. Making a phase shift for each computational basis state costs at least {$\Omega(n2^n)$}-size, which is unnecessarily high. We make the shift in Fourier basis, and carefully use ancillary qubits to parallelize the process. With ancillary qubits, we can first make some copies of the computational basis variables $x_i$, then partition $\{0,1\}^n$ into some parts of equal size, and use the ancillary qubits to handle different parts in parallel. We define the partition via a Gray code to minimize the update cost. Gray codes were also used in \cite{bullock2004asymptotically} to minimize the circuit size. They only need to minimize the difference between adjacent two words, so the defining property of Gray Code is enough. In our construction, however, we also need to make sure that the changed bits in different parts of the Gray code are evenly distributed. 

When no ancillary qubits are available, designing efficient circuit needs more ideas. Since there is no ancillary qubit available, all phase shifts need be made inside the input register. 
We divide the input register into two parts, control register and target register, and make phase shifts in the latter. 
As we only have a small space, we cannot use it to enumerate all $2^{r_t}-1$ suffixes as in the previous case, where $r_t\approx n/2$ is the length of suffixes. But we can enumerate them in many stages, by which we pay the price of time to compensate for the shortage of space. We need to make a transition between two consecutive stages. It turns out that the transition can be realized by a low-depth circuit if the suffixes enumerated in each stage are linearly independent as vectors over $\mathbb{F}_2$. Thus we need carefully divide the set of suffixes into sets of linearly independent vectors to facilitate the efficient update. Some other parts need special treatment as well. One is that we need to reset the suffix to the original input variables after going along a Gray code path. Another one is that the all-zero suffix cannot be handled in the same way for some singularity reason, for which we will use a recursion to solve the issue. It turns out that the overall depth and size obtained this way are asymptotically optimal.

The above constructions work well when $m$ is relatively small, but do not give a tight bound when $m = \Omega(2^n/n^2)$, for which we use another method. As we mentioned earlier, \cite{johri2021nearest} shows that unary-encoded QSP can be made in $O(n)$ depth and $O(2^n)$ size. Though the resulting state uses an exponentially long unary encoding, we can transform it to a binary encoding. A direct parallelization for this transform takes $O(n2^n)$ ancillary qubits, which can be improved to the $O(2^n)$ by first transforming it to a $2^{n/2+1}$-long matrix encoding $\ket{e_i} \to \ket{e_s}\ket{e_t}$, and then to the binary encoding. This gives the optimal depth and size for the regime $m \ge 2^n$. For $m\in[\omega(2^n/n^2),o(2^n)]$, the ancillary qubits only suffice for conducting the above for the first $\log_2 m$ qubits of the target state. For the rest $\le 2\log_2 n$ qubits, we invoke our first construction to complete the generation. This gives the optimal depth if $m\in[\omega(2^n/n^2),O(2^n/(n\log n))]$, the overall depth is asymptotically optimal, leaving a gap $[\Omega(n), O(n\log n)]$ only when $m$ is in a small range $[\omega(2^n/(n\log n), o(2^n)]$.


\paragraph{Other related work}  

Besides the standard QSP, researchers have also studied some relaxed versions.  
Araujo \textit{et al.} \cite{araujo2020divide} have given a depth upper bound of $O(n^2)$  to prepare a state $\sum_{k=0}^{2^n-1}v_k |k\rangle |\text{garbage}_k\rangle $, where $|\text{garbage}_k\rangle$ is $O(2^n)$-qubit state entangled with the target state register. Note that there is no generic way to remove the entangled garbage, this cannot be directly used to solve the standard QSP problem. 

One may also consider to approximately prepare quantum states by quantum circuits made of $\{H,S,T,CNOT\}$ gates to generate $|\psi'_v\rangle $ satisfying $ \left\| |\psi'_v\rangle -|\psi_v\rangle \right\|\leq \epsilon$ for various distance measures $\|\cdot \|$. Previous attention was paid to minimizing the number and depth of $T$ gates  \cite{low2018trading,babbush2018encoding}, which is non-Clifford and usually thought to be hard to realize experimentally \cite{low2018trading}. They have applied ancillary qubits to implement a circuit such that the number of $T$ gates can be optimized to $\frac{2^n}{\lambda}+\lambda \log ^2 \frac{2^n\lambda}{\epsilon}$ \cite{low2018trading}, where $\lambda \in [1,O(\sqrt{2^n})]$. We shall show that our construction can be adapted to this gate set and the circuit depth increases only by $O(n+\log(1/\epsilon))$. 

\paragraph{Subsequent work} After this work appeared on arXiv \cite{sun2021asymptotically}, Rosenthal \cite{rosenthal2021query} constructed a QSP circuit of depth $O(n)$, using $O(n2^n)$ ancillary qubits, as opposed to ours that only uses $O(2^n)$ ancillary qubits. Rosenthal also presented a circuit for general unitary synthesis of depth $\tilde{O}(2^{n/2})$ using $\tilde{O}(4^n)$ ancillary qubits.
This year, Zhang \emph{et al.} \cite{zhang2022quantum} presented yet another QSP circuit of depth $O(n)$ using $\Theta(2^n)$ ancillary qubits, which is a special cases of our results.

\paragraph{Organization} The rest of this paper is organized as follows. In  Section \ref{sec:preliminaries}, we will review notations and a framework of quantum state preparation. Then we will present how to decompose the uniformly controlled gate to diagonal unitary matrices and show the depth of quantum state preparation when the number of ancillary qubits $m=O(2^n/n^2)$ in Section \ref{sec:diag_matrix}. Next we will show two quantum circuit for diagonal unitary matrices used in previous section, with and without ancillary qubits in Section \ref{sec:QSP_withancilla} and Section \ref{sec:QSP_withoutancilla}, respectively. Furthermore, we present a new circuit framework for quantum state preparation when $m = \Omega\left(2^n/n^2\right)$ in Section \ref{sec:QSP_withmoreancilla}. 
In Section \ref{sec:extensions}, we will show some extensions and implications of the above bounds. Finally we conclude in Section \ref{sec:conclusions}. 

%% file: QSP_preliminaries.tex
In this section, we will introduce some basic concepts and notation. 

\paragraph{Notation} Let $[n]$ denote the set $\{1,2,\cdots,n\}$. All logarithms $\log(\cdot)$ are base 2 in this paper. Let $\mathbb{I}_n\in\mathbb{R}^{2^n\times 2^n}$ be the $n$-qubit identity operator. Denote by $\mathbb{F}_2$ the field with 2 elements, with multiplication $\cdot$ and addition $\oplus$, which can be overloaded to vectors: $x\oplus y=(x_1\oplus y_1,x_2\oplus y_2\cdots ,x_n\oplus y_n)^T$ for any $x,y\in \mathbb{F}_2^n$. 
The inner product of two vectors $s,x\in \mathbb{F}_2^n$ is $\langle s,x\rangle:=\oplus_{i=1}^{n}s_i\cdot x_i$ in which the addition and multiplication are over $\mathbb{F}_2$. 
%
We use $0^n$ and $1^n$ for the all-zero and all-one vectors of length $n$, respectively. Vector $e_i$ is the vector where the $i$-th element is $1$ and all other elements are $0$. The multiplication $\cdot$ is sometimes dropped if no confusion is caused. 
For $t,k\ge 1$ and $U_1,\ldots, U_k\in \mathbb{C}^{t\times t}$, $diag(U_1,U_2,\ldots,U_k)$ is defined as
\[diag(U_1,\ldots,U_k)\defeq\left[\begin{array}{ccc}
    U_1 &  &\\
     & \ddots &\\
     & & U_k
\end{array}\right]\in\mathbb{C}^{kt \times kt}.\]

\paragraph{Elementary gates} 
We will use the following $R_y(\theta)$, $R_z(\theta)$ and $R(\theta)$ to denote 1-qubit rotation (about Y-axis, Z-axis) gates and  phase-shift gate, i.e.,
\[
    R_y(\theta)=\left[\begin{array}{cc}
  \cos(\theta/2)   & -\sin(\theta/2) \\
   \sin(\theta/2)  &  \cos(\theta/2)
\end{array}\right], \quad
R_z(\theta)=\left[\begin{array}{cc}
  e^{-i(\theta/2)}   &  \\
    &  e^{i(\theta/2)} 
\end{array}\right],\quad
R(\theta)=\left[\begin{array}{cc}
  1   &  \\
    &  e^{i\theta} 
\end{array}\right],
\]
where $\theta\in\mathbb{R}$ is a parameter. All blank elements denote zero throughout this paper. Three important and special cases are the $\pi/8$ gate $T$, the phase gate $S$ and the Hadamard gate $H$,
\[{T=\left[\begin{array}{cc}
  1   &  \\
    &  e^{i\pi/4}
\end{array}\right],~
S=\left[\begin{array}{cc}
  1   &  \\
    &  i
\end{array}\right],~
H=\frac{1}{\sqrt{2}}\left[\begin{array}{cc}
  1   & 1 \\
   1  &  -1
\end{array}\right].}
\] 
The 2-qubit controlled-NOT gate is 
\[\text{CNOT}=\left[\begin{array}{cccc}
1 & & &\\
& 1 & &\\
& & & 1\\
& & 1 &\\
\end{array}\right].\]
The gate flips the \textit{target qubit} conditioned that the \textit{control qubit} is $\ket{1}$.

\paragraph{Single-qubit gate decomposition} Any single-qubit operator $U\in\mathbb{C}^{2\times 2}$ can be decomposed as
\[U=e^{i\alpha}R_z(\beta)R_y(\gamma)R_z(\delta)\]
for some $\alpha,\beta,\gamma,\delta\in\mathbb{R}$ \cite{nielsen2002quantum}. 
It is not hard to verify that the Y-axis rotation $R_y(\gamma)\in \mathbb{R}^{2\times 2}$ can be decomposed as $R_y(\gamma)=SHR_z(\gamma)HS^{\dagger},$
for any $\gamma\in\mathbb{R}$.
Putting these two facts together, we know that for a single-qubit operation $U$, there exist $\alpha,\beta,\gamma,\delta\in\mathbb{R}$ such that
\begin{equation}\label{eq:single_qubit_gate}
    U=e^{i\alpha}R_z(\beta)SHR_z(\gamma)HS^{\dagger}R_z(\delta).
\end{equation}

\paragraph{Gray code} 
A Gray code path is an ordering of all $n$-bit strings $\Bn$ in which any two adjacent strings differ by exactly one bit \cite{frank1953pulse,savage1997survey,gilbert1958gray}, and the first and the last string differ by one bit. That is, a Gray code path/cycle is a Hamiltonian path/cycle on the Boolean hypercube graph. Gray code paths/cycles are not unique, and a common one, called reflected binary code (RBC) or Lucal code, is as follows. Denote the ordering of $n$-bit strings by $x^1, x^2, \ldots, x^{2^n}$ and we will construct them one by one. Take $x^1 = 0^n$. For each $i = 1, 2, \ldots, 2^n-1$, the next string $x^{i+1}$ is obtained from $x^i$ by flipping the $\zeta(i)$-th bit, where the Ruler function $\zeta(i)$ is defined as $\zeta(i)=\max\{k:2^{k-1}|i\}$. In other words, $\zeta(i)$ is 1 plus the exponent of 2 in the prime factorization of $i$. The following fact is easily verified.
\begin{lemma}
\label{lem:GrayCode}
    The reflected binary code defined above is a Gray code cycle.
\end{lemma}
Note that in the above construction, if we list all the bits changed between circularly adjacent strings, we will get a list of length $2^n$. For instance, when $n=4$, the list is: 1,2,1,3,1,2,1,4,1,2,1,3,1,2,1,4.
In general, bit 1 appears $2^{n-1}$ times, bit 2 appears $2^{n-2}$ times, ..., bit $n-1$ appears twice, and bit $n$ appears twice as well. If we regard the code as a path, i.e. ignore the change of bit from the last string to the first string, then bit $n$ appears once. 

By circularly shifting the bits, we can also construct Gray code cycle such that bit 2 appears $2^{n-1}$ times, ..., bit $n$ and bit $1$ appear twice. In general, for any $k\in [n]$, we can make each bit $k,k+1,\ldots, n, 1, 2, \ldots, k-1$ to appear $2^n$, $2^{n-1},2^{n-2}, \ldots,2^2, 2, 2$ times, respectively. Let us call this construction \emph{$(k,n)$-Gray code path/cycle}, or simply the $k$-Gray code path/cycle if $n$ is clear from context.

%% file: QSP_diagonal_matrices.tex
In this section, we will review a natural framework of algorithm for quantum state preparation, first appeared in \cite{grover2002creating}. Our results presented in Section \ref{sec:QSP_withancilla} and \ref{sec:QSP_withoutancilla}, which achieve the optimal circuit depth, also fall into this framework. 
The framework to prepare an $n$-qubit quantum state is depicted in Figure \ref{fig:QSP_circuit}(a), where each qubit $j$ is handled by the circuit $V_j$.  The task for $V_j$ is to apply a single-qubit unitary on the last qubit conditioned on the basis state of the first $j-1$ qubits. In a matrix form, $V_j$ is a block-diagonal operator
\begin{equation}\label{matrix:Vn}
    V_j=diag(U_1,U_2,\ldots,U_{2^{j-1}})\in\mathbb{C}^{2^j\times 2^j},
\end{equation}
where each $U_i$ is a $2\times 2$ unitary matrix. There are different ways to implement $V_j$, and the most  natural one, which is also the one suggested in \cite{grover2002creating}, is in Figure \ref{fig:QSP_circuit}(b): it includes $2^{j-1}$ layers, and each layer is a controlled gate, which conditions on every possible computational basis state of the previous $j-1$ qubits and operates on the current qubit $j$. 
This is why sometimes $V_j$ is called \emph{uniformly controlled gate} (UCG). {We give a specific example for illustration in Appendix \ref{sec:app_BST}.}

\begin{figure}[ht]
\centering
\subfloat[]{{
\Qcircuit @C=0.6em @R=0.6em {
\lstick{\scriptstyle\ket{0}} & \gate{\scriptstyle V_1} & \gate{} \qwx[1] & \gate{} \qwx[1] & \qw & {\scriptstyle\cdots} & &  \qw & \gate{} \qwx[1] & \qw\\
\lstick{\scriptstyle\ket{0}} & \qw & \gate{\scriptstyle V_2} & \gate{} \qwx[1] & \qw & {\scriptstyle \cdots} & & \qw & \gate{} \qwx[1] & \qw\\
\lstick{\scriptstyle\ket{0}} & \qw & \qw & \gate{\scriptstyle V_3} & \qw & {\scriptstyle \cdots} &  & \qw & \gate{} \qwx[1] & \qw\\
{\scriptstyle\vdots~~~~~} & \qw & \qw  & \qw  & \qw & {\scriptstyle \cdots}  &  & \qw & \gate{} \qwx[1]& \qw\\
\lstick{\scriptstyle\ket{0}} & \qw & \qw & \qw & \qw & {\scriptstyle \cdots} &  & \qw & \gate{\scriptstyle V_n} & \qw 
\\
}}%
}\label{fig:circuit_n}
\subfloat[]{{
\Qcircuit @C=0.64em @R=1em {
& \gate{} \qwx[1] & \qw & & & \ctrlo{1} & \ctrl{1} & \ctrlo{1} & \qw &  ... & & \ctrl{1} & \qw\\
& \gate{} \qwx[1] & \qw & && \ctrlo{1} & \ctrlo{1} & \ctrl{1} &\qw & ... & & \ctrl{1}& \qw\\
& \gate{} \qwx[1] & \qw & & \boldsymbol{\scriptstyle =}~~ & \ctrlo{1}& \ctrlo{1} & \ctrlo{1} & \qw  &...  & & \ctrl{1}& \qw\\
& \gate{} \qwx[1] & \qw & & & \ctrlo{1} &\ctrlo{1} & \ctrlo{1} &  \qw  & ... & &  \ctrl{1} & \qw\\
& \gate{\scriptstyle V_j} & \qw &  & & \gate{\scriptstyle U_1} & \gate{ \scriptstyle U_2} & \gate{\scriptstyle U_3}  & \qw & ... & & \gate{\scriptstyle U_{2^{j-1}}} & \qw 
\\
}}%
}\label{fig:V_n}
\caption{(a) A quantum circuit to prepare an $n$-qubit quantum state. Every $V_j$ ($j\in[n]$) is a $j$-qubit uniformly controlled gate, where the first $j-1$ qubits are controlled qubits and the last one qubit is the target qubit. (b) A $j$-qubit uniformly controlled gate.}
\label{fig:QSP_circuit}
\end{figure}
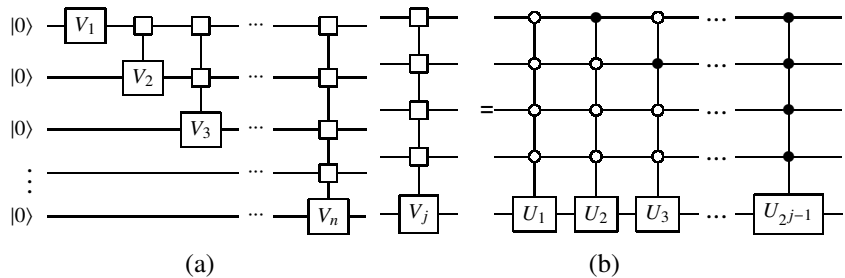
Thus the depth of the circuit {for quantum state preparation} in the above framework crucially depends on the circuit depth of the implementation of $V_j$'s.
\begin{lemma}\label{lem:QSP-by-UCG}
    If each $V_j$ can be implemented by a quantum circuit of depth $d_j$, then the quantum state can be prepared by a circuit of depth $\sum_{j=1}^n d_j$. 
\end{lemma}

As mentioned in Section \ref{sec:introduction}, if we implement each $V_j$ directly as in \cite{grover2002creating}, then the whole QSP circuit has a depth of $\Theta(n2^n)$, which is sub-optimal compared to our bound of $\Theta(2^n/n)$ in Theorem \ref{thm:QSP_noanci}. More importantly, the method in \cite{grover2002creating} cannot well utilize ancillary qubits to reduce the circuit depth. In this section, we will give a framework of efficient implementation of UCGs with the help of $m=O(2^n/n^2)$ ancillary qubits. The case when more ancillary qubits are available, i.e.  $m=\Omega(2^n/n^2)$, is handled by a different framework in Section \ref{sec:QSP_withmoreancilla}.

To overcome these drawbacks, we will first reduce the implementation of UCG to that of diagonal operators of the following form:
\begin{equation}\label{matrix:lambda_n}
    \Lambda_n=diag(1,e^{i\theta_1},e^{i\theta_2},\ldots,e^{i\theta_{2^n-1}} )\in\mathbb{C}^{2^n\times 2^n}.
\end{equation}

\begin{lemma}\label{lem:lamda2circuit}
    If one can implement $\Lambda_n$ in Eq. \eqref{matrix:lambda_n} by a circuit of depth $D(n)$ and size $S(n)$ using $m\ge 0$ ancillary qubits, then any $n$-qubit quantum state can be prepared by a circuit of depth $3 \sum_{k=1}^n D(k) + 2n + 1$  and size $3\sum_{k=1}^n S(k)+2n+1$. 
\end{lemma}
\begin{proof}
According to Eq. \eqref{eq:single_qubit_gate}, each unitary matrix $U_k\in \mathbb{C}^{2\times 2}$ can be decomposed as
\[
U_k=e^{i\alpha_k}R_z(\beta_k)SHR_z(\gamma_k)HS^{\dagger}R_z(\delta_k).
\]
Then the UCG $V_n$ can thus be decomposed to
\begin{multline} \label{eq:UCG}
    {V_n}=\underbrace{
    {diag(e^{i\alpha_1},\cdots,e^{i\alpha_{2^{n-1}}})\otimes \mathbb{I}_1}
    }_{A_1}  \cdot \underbrace{
   {diag(R_z(\beta_1),\cdots,R_z(\beta_{2^{n-1}}))
    }}_{A_2}\\\cdot
    \underbrace{  {\mathbb{I}_{n-1}\otimes (SH)}}_{A_3}
     \cdot\underbrace{
    {diag(R_z(\gamma_1),\cdots,R_z(\gamma_{2^{n-1}}))}
    }_{A_4}\cdot
     \underbrace{ { \mathbb{I}_{n-1}\otimes (HS^\dagger)}}_{A_5}
     \cdot\underbrace{
    {diag(R_z(\delta_1),\cdots,R_z(\delta_{2^{n-1}}))}
    }_{A_6}.
\end{multline}

Note that the unitary matrix $A_3$ can be implemented by a Hadamard gate $H$ and a phase gate $S$ operating on the last qubit, and similarly for $A_5$. The rest matrices, $A_1$, $A_2$, $A_4$, and $A_6$ are all $n$-qubit diagonal unitary matrices. 
Since a global phase can be easily implemented by a rotation on any one qubit, we can focus on implementing diagonal matrices of the form as in  Eq. \eqref{matrix:lambda_n}. If $\Lambda_n$ can be implemented by a circuit of depth $D(n)$ and size $S(n)$, so will be QSP by a circuit of depth { and size}  $\sum_{k=1}^n (3 D(k) +2) + 1 = 3\sum_{k=1}^n D(k) +2n + 1$ and $\sum_{k=1}^n (3 S(k) +2) + 1 = 3\sum_{k=1}^n S(k) +2n + 1,$
where the terms ``$3D(k)$'' {and ``$3S(k)$'' }
are for diagonal matrices $A_1$, $A_2$, $A_4$ and $A_6$, the term ``2'' is for $A_3$ and $A_5$, and the term ``1'' is for the global phase.
\end{proof}

Thus we only need to consider how to implement diagonal operators as in Eq. \eqref{matrix:lambda_n}. We will prove the following lemmas in Section \ref{sec:QSP_withancilla} and Section \ref{sec:QSP_withoutancilla}. 
\begin{lemma} \label{lem:DU_with_ancillary}
For any $m\in [2n,2^n/n]$, any diagonal unitary matrix $\Lambda_n\in\mathbb{C}^{2^n\times 2^n}$ as in Eq. \eqref{matrix:lambda_n} can be implemented by a quantum circuit of depth $O\left(\log m + \frac{2^n}{m}\right)$ and size $O(2^n)$, with $m$ ancillary qubits. 
\end{lemma}

\begin{lemma} \label{lem:DU_without_ancillary}
Any diagonal unitary matrix $\Lambda_n\in\mathbb{C}^{2^n\times 2^n}$ as in Eq. \eqref{matrix:lambda_n} can be implemented by a quantum circuit of depth $O\big(\frac{2^n}{n}\big)$  and size $O\big(2^n\big)$ without ancillary qubits. 
\end{lemma}
Lemmas \ref{lem:DU_with_ancillary}  and \ref{lem:DU_without_ancillary} imply Lemma \ref{lem:UCG_depth}.
\begin{lemma}\label{lem:UCG_depth}
    For $m\ge 0$, any uniformly controlled gate $V_n\in\mathbb{C}^{2^n\times 2^n}$ as in Eq. \eqref{matrix:Vn} can be implemented by a quantum circuit of depth $O\big(n+\frac{2^n}{n+m}\big)$  and size $O(2^n)$ with $m$ ancillary qubits. 
\end{lemma}
\begin{proof}
According to Eq. \eqref{eq:UCG}, every $V_n$ can be decomposed into 3 $n$-qubit diagonal unitary matrices and 4 single-qubit gates. Combining with Lemma \ref{lem:DU_with_ancillary} and \ref{lem:DU_without_ancillary}, $V_n$ can be realized by a quantum circuit of depth $O\big(n+\frac{2^n}{n+m}\big)$  and size $O(2^n)$ with $m$ ancillary qubits.
\end{proof}

Once we prove these two lemmas, we will be able to prove Theorems \ref{thm:QSP_anci} and \ref{thm:QSP_noanci}. Indeed, we can apply the next Lemma \ref{lem:partial_result} to prove Theorem \ref{thm:QSP_noanci} $(m=0)$ and the $m= O(2^n/n^2)$ part of Theorem \ref{thm:QSP_anci}. The other part $m= \Omega(2^n/n^2)$ of Theorem \ref{thm:QSP_anci} is the same as Corollary \ref{coro:QSP_moreancilla} and will be treated in Section \ref{sec:QSP_withmoreancilla}.

\begin{lemma}\label{lem:partial_result}
    For any $m\ge 0$, any $n$-qubit quantum state $\ket{\psi_v}$ can be generated by a quantum circuit with $m$ ancillary qubits, using single-qubit gates and CNOT gates, of {size $O(2^n)$ and } depth $O\big(n^2+\frac{2^n}{m+n}\big)$.
\end{lemma}
 \begin{proof}
    We prove the case $m=0$ first.
    Plugging Lemma \ref{lem:DU_without_ancillary} into Lemma \ref{lem:lamda2circuit}, we get a circuit solving QSP in size $\sum_{j=1}^nO(2^{j})+2n+1=O(2^n)$ and depth 
    $O\big(\sum_{j=1}^n \frac{2^j}{j}+ n\big) = O\big(\sum_{j=1}^{n-\lceil\log n\rceil}\frac{2^j}{j}+ \sum_{j=n-\lceil\log n\rceil+1}^n\frac{2^j}{j}\big)=O\big(\sum_{j=1}^{n-\lceil\log n\rceil}2^j+ \sum_{j=n-\lceil\log n\rceil+1}^n\frac{2^j}{n-\lceil\log n\rceil+1}\big)=O\big( \frac{2^n}{n}\big),$
    as desired.
    
    Now we prove the case $m>0$. 
    If $1\le m < 2n$, we will not use the ancillary qubits---we just invoke Theorem \ref{thm:QSP_noanci} to obtain a circuit of depth $O(2^n/n)$.
    If $2n \le m\le 2^n/n^2(\le 2^n/n)$, we can combine Lemma \ref{lem:lamda2circuit} and Lemma \ref{lem:DU_with_ancillary} to give a circuit of size $3\sum_{j=1}^nO(2^j)+2n+1=O(2^n)$ and depth 
    $O\big(\sum_{j=1}^n \big( \log m + \frac{2^j}{m}\big)+ n\big) = O\big(n^2 + \frac{2^n}{m}\big).
    $
    If $m > 2^n/n^2$, we only use the first $2^n/n^2$ ancillary qubits, then the above equality gives a circuit of depth $O(n^2)$. 
    Putting these three cases together, we obtain the claimed size upper bound of $O(2^n)$ and depth upper bound of $O\big(n^2+\frac{2^n}{m+n}\big)$.
\end{proof}

Next let us consider how to efficiently implement $\Lambda_n$, which essentially makes a phase shift on each computational basis state. Again, if we do this on each basis state, it takes at least $\Omega(2^n)$ rounds, with each round implementing an $n$-qubit controlled phase shift. One way of avoiding sequential applications of $(n-1)$-qubit controlled unitaries is to make rotations on its Fourier basis. Indeed, there are several pieces of work to synthesis a diagonal unitary matrix, and a common approach is generating all the linear functions of variables and adding corresponding rotation $R(\theta)$ gate when a new combination generated \cite{welch2014efficient,welch2014efficient2,bullock2004asymptotically}. In \cite{bullock2004asymptotically} the authors use Gray code to adjust the order of combinations so the size and depth of the circuit are $O(2^n)$. 
With ancillary qubits, we can actually achieve this with much smaller depth by carefully parallelizing the operations (Section \ref{sec:QSP_withancilla}). Interestingly, this approach turns out to inspire our construction for circuits \textit{without} ancillary qubits (Section \ref{sec:QSP_withoutancilla}), to achieve the optimal depth complexity as in Theorem \ref{thm:QSP_noanci}.

We now give more details. Suppose we can accomplish the following two tasks: 
\begin{enumerate}
    \item For every $s\in \{0,1\}^n-\{0^n\}$, make a phase shift of $\alpha_s$ on each basis $\ket{x}$ when $\langle s,x\rangle = 1$ (recall that $\langle \cdot , \cdot\rangle$ is over $\mathbb{F}_2$), i.e. 
    \begin{equation}\label{eq:task1}
     \ket{x} \to  e^{i\alpha_s\langle s,x\rangle } \ket{x}.
    \end{equation}
    
    \item Find $\{\alpha_s:s\in \Bn-\{0^n\}\}$ s.t. \begin{equation}\label{eq:alpha}
        \sum_{s\in \{0,1\}^n-\{0^n\}}\alpha_s\langle x,s\rangle = \theta(x), \quad \forall x\in \{0,1\}^n-\{0^n\}.
    \end{equation}
\end{enumerate}
Then we get
\[\ket{x} \to \prod_{ s\in \{0,1\}^n-\{0^n\}} e^{i\alpha_s\langle s,x\rangle } \ket{x} = e^{i\Sigma_s\alpha_s \langle s,x\rangle} \ket{x} = e^{i\theta(x) } \ket{x},\]  
as required in $\Lambda_n$. For notational convenience, we define $\alpha_{0^n}=0$.

{The implementations of above two tasks in Eq. \eqref{eq:task1} and Eq. \eqref{eq:alpha} are accomplished in Appendix \ref{sec:2tasks}.}




%% file: QSP_withancilla.tex
In this section, we prove Lemma \ref{lem:DU_with_ancillary}. That is, for any $m\in [2n,2^n{/n}]$, any diagonal unitary matrix $\Lambda_n\in\mathbb{C}^{2^n\times 2^n}$ as in Eq. \eqref{matrix:lambda_n} can be implemented by a quantum circuit of depth $O\left( \log m + \frac{2^n}{m}\right)$ and size $O(2^n)$ with $m$ ancillary qubits. 
Let us first give a high-level explanation of the circuit. We divide the ancillary qubits into two registers: One is used to make multiple copies of basis input bits to help on parallelization, and the other is used to generate all $n$-bit strings and  apply the rotation gates. 
State $\ket{\langle s,x\rangle}$ will be generated for all $s\in\Bn-\{0^n\}$. To reduce the depth of the circuit, these strings are split as equally as possible, and we use Gray Code to minimize the cost of generating a new $n$-bit string from an old one.  {A quantum circuit to implement $\Lambda_4$ by using 8 ancillary qubits is shown in Appendix \ref{sec:warm-up}.}

We will show how to implement $\Lambda_n$ with $m$ ancillary qubits. Let us assume $m$ to be an even number to save some floor or ceiling notation without affecting the bound. 
The framework is shown in Figure \ref{fig:Lambda_n_gray_code}. Our framework consists of {three registers and five stages}. The first $n$ qubits labeled as $x_1,x_2,\cdots,x_n$ form the \textit{input register}, the next $\frac{m}{2}$ qubits are the \textit{copy register}, and the last $\frac{m}{2}$ qubits are the \textit{phase register}.  The linear functions $\langle s, x\rangle$ of the input variables $x=x_1\ldots x_n$ are generated in the phase register. We use the copy register to make copies of $x$ for parallelizing the circuit later. Partition $s$ into a prefix $s_1$ and a suffix $s_2$. We then generate a specific function $\langle s_{1}0\cdots 0,x\rangle$ on each qubit in the phase register, and iterate other non-zero suffixes $s_2$ in the order of a  Gray code and generate $\langle s_1s_2, x\rangle$. All qubits in the copy and phase registers are initialized to $|0\rangle$.

\begin{figure}[ht]
\centerline{
\Qcircuit @C=0.6em @R=0.4em {
& \lstick{} &\lstick{\scriptstyle\ket{x_1}} & \multigate{6}{\text{Prefix~Copy}\atop \text{Stage}} & \multigate{10}{\text{Gray Initial}\atop \text{Stage}} &\multigate{6}{\text{Suffix~Copy}\atop \text{Stage}} & \multigate{10}{\text{Gray Path} \atop \text{Stage}} & \multigate{10}{\text{Inverse}\atop \text{Stage}} &\qw\\
& \lstick{}  &\lstick{\scriptstyle\vdots~~}  & \ghost{\text{Prefix~Copy}\atop \text{Stage}} & \ghost{\text{Gray Initial}\atop \text{Stage}}  &\ghost{\text{Suffix~Copy}\atop \text{Stage}} &\ghost{\text{Gray Path} \atop \text{Stage}} & \ghost{\text{Inverse}\atop \text{Stage}}  &\qw\\
& \lstick{}  &\lstick{\scriptstyle\ket{x_n}} & \ghost{\text{Prefix~Copy}\atop \text{Stage}} & \ghost{\text{Gray Initial}\atop \text{Stage}}  &\ghost{\text{Suffix~Copy}\atop \text{Stage}} & \ghost{\text{Gray Path} \atop \text{Stage}}& \ghost{\text{Inverse}\atop \text{Stage}}  &\qw \\
& \lstick{}  & \lstick{}    & & & & & &\\
& \lstick{}  &\lstick{\scriptstyle|0\rangle} & \ghost{\text{Prefix~Copy}\atop \text{Stage}} & \ghost{\text{Gray Initial}\atop \text{Stage}}  &\ghost{\text{Suffix~Copy}\atop \text{Stage}} & \ghost{\text{Gray Path} \atop \text{Stage}}& \ghost{\text{Inverse}\atop \text{Stage}} &\qw\\
&\lstick{}  &\lstick{\scriptstyle\vdots~~}  & \ghost{\text{Prefix~Copy}\atop \text{Stage}} & \ghost{\text{Gray Initial}\atop \text{Stage}} &\ghost{\text{Suffix~Copy}\atop \text{Stage}} & \ghost{\text{Gray Path} \atop \text{Stage}} & \ghost{\text{Inverse}\atop \text{Stage}} &\qw\\
& \lstick{}  &\lstick{\scriptstyle|0\rangle} &  \ghost{\text{Prefix~Copy}\atop \text{Stage}}& \ghost{\text{Gray Initial}\atop \text{Stage}}  &\ghost{\text{Suffix~Copy}\atop \text{Stage}} & \ghost{\text{Gray Path} \atop \text{Stage}} & \ghost{\text{Inverse}\atop \text{Stage}}  &\qw\\
&\lstick{}  &              & & & & & &\\
& \lstick{}  &\lstick{\scriptstyle|0\rangle} & \qw & \ghost{\text{Gray Initial}\atop \text{Stage}} &\qw &  \ghost{\text{Gray Path} \atop \text{Stage}} & \ghost{\text{Inverse}\atop \text{Stage}} &\qw\\
& \lstick{}  &\lstick{\scriptstyle\vdots~~}  & \qw & \ghost{\text{Gray Initial}\atop \text{Stage}}  & \qw & \ghost{\text{Gray Path} \atop \text{Stage}} & \ghost{\text{Inverse}\atop \text{Stage}} &\qw\\
& \lstick{}  &\lstick{\scriptstyle|0\rangle} & \qw & \ghost{\text{Gray Initial}\atop \text{Stage}}  & \qw & \ghost{\text{Gray Path} \atop \text{Stage}}& \ghost{\text{Inverse}\atop \text{Stage}} &\qw\\ 
}
}
\caption{Framework for the circuit of $\Lambda_n$ with $m$ ancillary qubits. The first $n$ qubits $\ket{x_1\cdots x_n}$ form the input register, the next $\frac{m}{2}$ qubits the copy register and the last $\frac{m}{2}$ qubits the phase register. The framework consists of five stages: Prefix Copy, Gray Initial, Suffix Copy, Gray Path and Inverse. The depth of the five stages are $O(\log m)$, $O(\log m)$, $O(\log m)$, $O\left(\frac{2^n}{m}\right)$ and $O\left(\log m+\frac{2^n}{m}\right)$,  respectively. 
}\label{fig:Lambda_n_gray_code}
\end{figure}
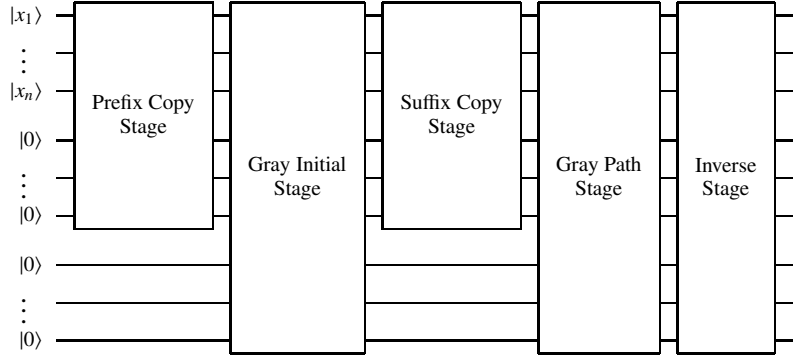


\paragraph{\underline{Stage 1: Prefix Copy}}
In this stage, we make $\left\lfloor \frac{m}{2t} \right\rfloor$ copies of each qubit  $x_1,x_2,\cdots,x_t$ in the input register, where $t = {\left\lfloor \log \frac{m}{2} \right\rfloor} < n$. More formally, the circuit implements the unitary $U_{copy,1}$ which operates on the input and copy registers only. Its effect is 
\begin{align}\label{eq:Uc1-effect}
\ket{x} \ket{0^{m/2}}  \xrightarrow{U_{copy,1}} 
\ket{x}\ket{x_{pre}} 
\end{align}
where the two parts in the ket notation are for the input and copy register, respectively, and 
\begin{align*}
    & \ket{x} = \ket{x_1x_2\cdots x_n},\\
    &\ket{x_{pre}} = \overbrace{|\underbrace{x_1\cdots x_1}_{\left\lfloor\frac{m}{2t}\right\rfloor~\text{qubits}}\underbrace{x_2\cdots x_2}_{\left\lfloor\frac{m}{2t}\right\rfloor~\text{qubits}}\cdots\underbrace{x_t\cdots x_t}_{\left\lfloor\frac{m}{2t}\right\rfloor~\text{qubits}}0\cdots 0}^{{{m/2}}~\text{qubits}}\rangle.
\end{align*}
The next lemma says that this operation can be carried out by a circuit of small depth. 
\begin{lemma}\label{lem:copy1}
    We can make $\left\lfloor \frac{m}{2t} \right\rfloor$ copies of each qubit  $x_1,x_2,\cdots,x_t$ in the input register and the copy register, by an  ($m/2$)-size circuit $U_{copy,1}$ of CNOT gates only, in depth at most $\log m$. 
\end{lemma}
\begin{proof}
First, we make 1 copy of each $x_i$ in the input register to a qubit in the copy register by applying a CNOT gate. Note that the CNOT gates for different $x_i$'s are applied on different pairs of qubits, thus they can be implemented in parallel in depth 1. Next, we utilize the $x_i$ in the input register and the $x_i$ in the copy register (that we just obtained) to make two more copies of $x_i$ in the copy register, and again all these $2t$ CNOT gates can be implemented in depth 1. We continue this until we get $\lfloor m/(2t)\rfloor$ copies of each qubit  $x_1,x_2,\cdots,x_t$ in the copy register. The depth of this Copy stage is $\big \lceil \log \big \lfloor m/2t \big\rfloor \big \rceil \le \log m$. And the size of this stage is $m/2$, since each qubit in copy register is used as the target qubit of CNOT gate only once.
\end{proof}

\paragraph{\underline{Stage 2: Gray Initial}} 
In this stage, the circuit includes two steps. The first step $U_{1}$ implements $m/2$ linear functions $f_{j1}(x) = \langle s(j,1), x\rangle$ for some $n$-bit strings $s(j,1)$,
one for each qubit $j$ in the phase register. The second step implements some rotations in the phase registers. 
To elaborate on which {strings} are implemented in the first {step}, we need the following lemma and notation. Recall that we are in the parameter regime $m \in [2n,2^n/n]$. 
\begin{lemma}\label{lem:2D-array}
    Let $t = {\lfloor \log \frac{m}{2} \rfloor}$ and $\ell = 2^t$. The set $\Bn$ can be partitioned into a 2-dimensional array $\{s(j,k): j\in [\ell], k\in [2^n/\ell]\}$ of $n$-bit strings, satisfying that 
    \begin{enumerate}
        \item Strings in the first column $\{s(j,1): j\in [\ell]\}$ have the last $(n-t)$ bits being all 0, and strings in each row $\{s(j,k): k\in [2^n/\ell]\}$ share the same first $t$ bits.
        \item $\forall j\in [\ell], \forall k\in [2^n/\ell-1]$, $s(j,k)$ and $s(j,k+1)$ differ by 1 bit.
        \item For any fixed $k\in [2^n/\ell-1]$, and any $t' \in \{t+1,...,n\}$, there are at most $\big(\frac{m}{2(n-t)} + 1\big)$ many $j\in [\ell]$ s.t. $s(j,k)$ and $s(j,k+1)$ differ by the $t'$-th bit.
    \end{enumerate}
\end{lemma}
The proof of Lemma \ref{lem:2D-array} is shown in Appendix \ref{sec:2D-array}.

Let us denote by $t_{jk}$ the index of the bit that $s(j,k)$ and $s(j,k+1)$ differ by. We can now describe this stage in more details. 
\begin{enumerate}
    \item The first step $U_{1}$ aims to let each qubit $j$ in the phase register have the state $\ket{f_{j1}(x)}$ at the end of this step, where $f_{j1}(x) = \langle s(j,1), x\rangle$. 
    \item The second step applies the rotation $R_{j,1} \defeq R(\alpha_{s(j,1)})$ on each qubit $j$ in the phase register. That is, the state is rotated by an phase angle of $\alpha_{s(j,1)}$ if $\langle x, s(j,1)\rangle = 1$, and left untouched otherwise. Put $R_1 = \otimes_{j\in [\ell]} R_{j,1}$.
\end{enumerate}

The next lemma gives the cost and effect of this stage. 

\begin{lemma}\label{lem:GrayInit}
    The Gray Initial Stage 
    can be implemented in depth at most $2\log m$ and in size at most $\frac{(n+1)m}{2}$ such that 
    its unitary $U_{GrayInit}$ satisfies
    \begin{equation}\label{eq:UGI-effect}
        \ket{x} \ket{x_{pre}} \ket{0^{m/2}}
        \xrightarrow{U_{GrayInit}}  e^{i \sum_{j\in [\ell]} f_{j,1}(x) \alpha_{s(j,1)} } \ket{x} \ket{x_{pre}} \ket{f_{[\ell],1}},
    \end{equation}
    where $\ket{f_{[\ell],1}} = \otimes_{j\in [\ell]} \ket{f_{j,1}(x)}$.
\end{lemma}
\begin{proof}
We will show how to implement the first step $U_{1}$ 
such that all $\ell=2^t = 2^{\left\lfloor\log \frac{m}{2}\right\rfloor}$ linear functions of the prefix variables $x_1, \ldots, x_t$ are implemented, namely after $U_{1}$, the states of the $2^t$ qubits in the phase register are exactly $\{a_1x_1 \oplus \cdots \oplus a_t x_t: a_1, \ldots, a_t \in \{0,1\}\}$. The implementation makes each qubit $j$ in the phase register have state $\ket{f_{j,1}(x) }$. Then in the second step, each qubit $j$ adds a phase of $f_{j,1}(x) \cdot \alpha_{s(j,1)}$ to $\ket{x} \ket{x_{pre}} \ket{0^{m/2}}$. We thus have 
\begin{align}\label{eq:U1}
    \ket{x} \ket{x_{pre}} \ket{0^{m/2}} &\xrightarrow{U_1} \ket{x} \ket{x_{pre}} \ket{f_{[\ell],1}}, \\
    & \xrightarrow{R_1} e^{i \sum_{j\in [\ell]} f_{j,1}(x) \alpha_{s(j,1)} } \ket{x} \ket{x_{pre}} \ket{f_{[\ell],1}}.
\end{align}

Now let us construct a shallow circuit for the first step $U_{1}$. Recall that we have $\ell = 2^t$ qubits $j$ each with a corresponding linear function in variables $x_1, \ldots, x_t$. Since $\ell \le m/2$, the phase register has enough qubits to hold these linear functions. For a qubit $j$ in the phase register with corresponding linear function $x_{i_1}\oplus \cdots \oplus x_{i_{t'}}$ ($t'\le t$), we will use CNOT gates to copy the qubits $x_{i_1}, ..., x_{i_{t'}}$ from the work and the copy registers to qubit $j$. We just need to allocate these CNOT gates evenly to make the overall depth small. This step can be divided into $\big\lceil \frac{2^t}{t\lfloor m/(2t)\rfloor}\big\rceil $ mini-steps, each mini-step handling $t\left\lfloor\frac{m}{2t}\right\rfloor$ qubits $j$ by assigning the state $\ket{\ \langle s(j,1),x\rangle\ }$ to it. Since we have $\ell = 2^t$ qubits to handle, it needs $\big\lceil \frac{2^t}{t\lfloor m/(2t)\rfloor}\big\rceil$ mini-steps. 

For all positions $i\in [t]$ with $s(j,1)_i = 1$, we use CNOT to copy $x_i$ to qubit $j$. We have $t$ variables $x_1, \ldots, x_t$, each with $\lfloor m/(2t) \rfloor$ copies. To utilize these copies for parallelization, we break the $t\lfloor m/(2t) \rfloor$ target qubits into $t$ blocks of size $\lfloor m/(2t) \rfloor$ each. 
Each mini-step gives all needed variables for $t(\lfloor\frac{m}{2t}\rfloor+1)$ qubits $j$, in depth $t$. In the first layer, we use the $\lfloor\frac{m}{2t}\rfloor$ copies of $x_1$ as control qubits in CNOT to copy $x_1$ to the first block of target qubits $j$, use the $\lfloor\frac{m}{2t}\rfloor$ copies of $x_2$ for the second block of target qubits, and so on, to $x_t$ for the $t$-th block. Then in the second layer, we repeat the above process with a circular shift: Copy $x_1$ to block 2, $x_2$ to block 3, ..., $x_{t-1}$ to block $t$, and $x_t$ to block 1. Repeat this and we can complete this mini-step in depth $t$, such that $t \lfloor\frac{m}{2t}\rfloor$ many qubits $j$ get their needed variables. 

Since there are $\big\lceil \frac{2^t}{t\lfloor m/(2t)\rfloor}\big\rceil $ mini-steps, each of depth $t$, the total depth for $U_{1}$ is 
$\big\lceil \frac{2^t}{t\lfloor m/(2t)\rfloor}\big\rceil  \cdot t \le \frac{m/2}{m/(2t)} + t = 2t = 2 \lfloor \log(m/2) \rfloor \le 2\log m - 2.$

The rotations in the second step are on different qubits and thus can be put into one layer, thus the overall depth for Gray Initial Stage is at most $2\log m$.

The size of this stage is at most {$(n+1)m/2$}, because each qubit in the phase register has at most $n$ CNOT gates {and one $R_z$ gate} on it. 
%
\end{proof}

\paragraph{\underline{Stage 3: Suffix Copy}} 
In this stage, we first undo $U_{copy,1}$, and then make $\big\lfloor \frac{m}{2(n-t)}\big\rfloor$ copies for each of the suffix variables, namely $x_{t+1},...,x_n$. The next lemma is similar to Lemma \ref{lem:copy1} and we omit the proof.
\begin{lemma}\label{lem:copy2}
    We can make $\left\lfloor \frac{m}{2(n-t)} \right\rfloor$ copies of each qubit  $x_{t+1},x_{t+2},\cdots,x_n$ in the input register and the copy register, by applying on $\ket{x} \ket{0^{m/2}}$ an $m$-size circuit $U_{copy,2}$ of CNOT gates only, in depth at most $\log m$.
\end{lemma}
Define 
\[
\ket{x_{suf}}\defeq 
|\overbrace{
\underbrace{x_{t+1}\cdots x_{t+1}}_{\left\lfloor\frac{m}{2(n-t)}\right\rfloor~\text{qubits}} \cdots\underbrace{x_n\cdots x_n}_{\left\lfloor\frac{m}{2(n-t)}\right\rfloor~\text{qubits}}0\cdots 0
}^{{m/2}~\text{qubits}}\rangle,
\]
then the effect of $U_{copy,2}$ is 
\[
    \ket{x} \ket{0^{m/2}} \xrightarrow{U_{copy,2}} \ket{x}\ket{x_{suf}}.
\]
The operator of this stage is $U_{copy,2} U_{copy,1}^\dagger$, and the depth is at most $2\log m$ and the size is at most $m$. The effect of this stage $U_{copy,2}  U_{copy,1}^\dagger$ is 
\begin{equation}\label{eq:suf-copy}
    \ket{x} \ket{x_{pre}} 
    \xrightarrow{U_{copy,1}^\dagger} \ket{x} \ket{0^{m/2}} \xrightarrow{U_{copy,2}} \ket{x}\ket{x_{suf}}.
\end{equation}

\paragraph{\underline{Stage 4: Gray Path}} 
This stage contains $2^n/\ell- 1$ phases, indexed by $k = 2, 3, \ldots, 2^n/\ell$. The previous Gray Initial Stage can be also viewed as the phase $k=1$. We single it out as a stage because it implements linear functions from scratch, while each phase in the Gray Path Stage implements linear functions only by a small update from the previous phase. 

In each phase $k$ in this stage, the circuit has two steps: 
\begin{enumerate}
    \item Step $k.1$ is a unitary circuit $U_{k}$ that applies a CNOT gate on each qubit $j\in [\ell]$ in the phase register, controlled by $x_{t_{j,(k-1)}}$, the bit where $s(j,k-1)$ and $s(j,k)$ differ. 
    \item Step $k.2$ applies the rotation gate $R(\alpha_{s(j,k)})$ on qubit $j$. Put $R_k = \otimes_{j\in [\ell]} R(\alpha_{s(j,k)})$.
\end{enumerate}  
\begin{lemma}\label{lem:GrayPath}
    The phase $k$ of the Gray Path Stage implements 
\begin{equation}\label{eq:GrayPath}
\ket{x}\ket{x_{suf}} \ket{f_{[\ell],k-1}}  \xrightarrow{U_k} \ket{x}\ket{x_{suf}} \ket{f_{[\ell],k}} \xrightarrow{R_k} e^{i \sum_{j\in [\ell]} f_{j,k}(x) \alpha_{s(j,k)} } \ket{x}\ket{x_{suf}} \ket{f_{[\ell],k}}, \end{equation}
    where $f_{j,k}(x)=\langle{s(j,k)},x\rangle$ and $\ket{f_{[\ell],k}}=\otimes_{j\in[\ell]}\ket{f_{j,k}(x)}$. The depth and size of the whole Gray Path Stage are at most ${ 2}\cdot 2^{n}/\ell$ and $2^{n+1}$. 
\end{lemma}
\begin{proof}
The operation can be easily seen in a similar way as that for Lemma \ref{lem:GrayInit}. Next we show the depth bound. 
{The Gray Path stage repeats step $k.1$-$k.2$ for $2^n/\ell-1$ times. Since $s(j,k-1)$ and $s(j,k)$  differ by only 1 bit by Lemma \ref{lem:2D-array}, one CNOT gate suffices to implement the function $\langle x, s(j,k)\rangle$ from $\langle x, s(j,k-1)\rangle$ in the previous phase: The control qubit is $x_{t_{j,(k-1)}}$ and the target qubit is $j$. Moreover, the third property in Lemma \ref{lem:2D-array} shows that each variables $x_i$ is used as a control qubit for at most $\big(\big\lfloor\frac{m}{2(n-t)}\big\rfloor+1\big)$ different $j\in [\ell]$. Since we have { $\big(\big\lfloor \frac{m}{2(n-t)}\big\rfloor+1\big) $} copies in the input register and the copy register, these CNOT gates in step $k.1$ can be implemented in depth $ 1$. 

The step $k.2$ consists of only single qubit gates, which can be all paralleled in depth 1. Thus the total depth of Gray Path stage is at most $2^n/\ell \cdot {(1+1)\le 2}\cdot 2^{n}/\ell$.}

{ The size of this stage is $2^{n+1}$ since each linear combination of input variables is generated once and applied single-qubit phase-shift gates $R_k$. The number of linear combinations of input variables is $2^n$ , so the size is $2^{n+1}$.}
\end{proof}

\paragraph{\underline{Stage 5: Inverse}} 
In this stage, the circuit applies {$U_{copy,1}^\dagger U_1^\dagger U_{copy,1} U_{copy,2}^\dagger U_{2}^{\dagger} \cdots U_{2^n/\ell}^\dagger $}.
\begin{lemma}\label{lem:inverse}
    The depth and size of the Inverse Stage are at most $O(\log m + 2^n /m)$ and $\frac{m}{2}+\frac{nm}{2}+m+2^{n}=2^{n}+\frac{3m+nm}{2}$. The effect of this stage is 
    \begin{equation}\label{eq:inverse}
        \ket{x} \ket{x_{suf}} \ket{f_{[\ell],2^n/\ell}} \xrightarrow{U_{Inverse}}  \ket{x} \ket{0^{m/2}} \ket{0^{m/2}}.
    \end{equation}
\end{lemma}
The proof of Lemma \ref{lem:inverse} is shown in Appendix \ref{sec:inverse}.

\paragraph{Putting things together} After explaining all the five stages, we are ready to put them together to see the overall depth and operation of the circuit. 
\begin{lemma}\label{lem:puttingtogether_ancilla}
    The circuit implements the operation in Eq. \eqref{matrix:lambda_n} in depth $O(\log m + 2^n /m)$ { and in size $3\cdot 2^{n}+nm+{\frac{7}{2}}m$}.
\end{lemma}
{The proof of Lemma \ref{lem:puttingtogether_ancilla} is shown in Appendix \ref{sec:puttingtogether_ancilla}.}
In summary, $\Lambda_n$ can be implemented in $O\big(\log m+\frac{2^n}{m}\big)$ depth and size $3\cdot 2^n +nm+{\frac{7}{2}m}$ with $m\in[2n,2^n/n]$ ancillary qubits, proving Lemma \ref{lem:DU_with_ancillary}.

%% file: QSP_withoutancilla.tex
In this section, we prove Lemma \ref{lem:DU_without_ancillary}. That is, any diagonal unitary $\Lambda_n\in\mathbb{C}^{2^n\times 2^n}$ as in Eq. \eqref{matrix:lambda_n} can be implemented by a quantum circuit of depth $O\left(2^n/n\right)$ and size $O(2^n)$ without ancillary qubits.
In Section \ref{sec:framework_DU_withoutancilla}, we present the framework of our circuit and the functionalities of the operators inside. We then prove the correctness and analyze the depth of our circuit in Section \ref{sec:correctness}. Finally, we give the detailed construction of some operators in Section \ref{sec:G_k}.

\subsection{Framework and functionalities}
\label{sec:framework_DU_withoutancilla}

The framework of our circuit implementing $\Lambda_n$ is a recursive procedure shown in Figure \ref{fig:framework_DU_withoutancilla}.
\begin{figure}[!ht]
\centerline{
\Qcircuit @C=1em @R=0.8em {
\lstick{}& & & \lstick{\scriptstyle\ket{x_1}} & \multigate{7}{{\scriptstyle\Lambda_n}} &\qw & & &  \qw & \multigate{7}{\scriptstyle \mathcal{G}_1} & \multigate{7}{ \scriptstyle\mathcal{G}_2} & \qw & ... &  &\qw & \multigate{7}{\scriptstyle \mathcal{G}_{\ell}} &\qw & \multigate{3}{{\scriptstyle\Lambda_{r_c}}} &\qw\\
\lstick{}& & &\lstick{\scriptstyle\ket{x_2}}&\ghost{\scriptstyle\Lambda_n}&\qw & &  &\qw & \ghost{\scriptstyle\mathcal{G}_1} & \ghost{\scriptstyle\mathcal{G}_2} &\qw & ... & &\qw & \ghost{\scriptstyle\mathcal{G}_{\ell}} &\qw & \ghost{{\scriptstyle\Lambda_{r_c}}} &\qw\\
\lstick{}& & & \lstick{\scriptstyle\vdots} &\ghost{\scriptstyle\Lambda_n}&\qw & & & \qw &  \ghost{\scriptstyle\mathcal{G}_1} & \ghost{\scriptstyle\mathcal{G}_2} &\qw & ... & &\qw & \ghost{\scriptstyle\mathcal{G}_{\ell}} &\qw & \ghost{{\scriptstyle\Lambda_{{r_c}}}} &\qw\\
\lstick{}& &  &\lstick{\scriptstyle\ket{x_{r_c}}}&\ghost{\scriptstyle \Lambda_n}&\qw & & & \qw &  \ghost{\scriptstyle\mathcal{G}_1} & \ghost{\scriptstyle\mathcal{G}_2} &\qw & ... & & \qw &\ghost{\scriptstyle\mathcal{G}_{\ell}} &\qw &\ghost{{\scriptstyle\Lambda_{r_c}}} &\qw
\\
\lstick{} & & &\lstick{\scriptstyle\ket{x_{r_c+1}}}&\ghost{\scriptstyle\Lambda_n}&\qw & ~~~~~~~~~~~\textbf{\Large=}~~~~~~~~~~~& & \qw & \ghost{\scriptstyle\mathcal{G}_1} & \ghost{\scriptstyle\mathcal{G}_2} &\qw & ... & & \qw & \ghost{\scriptstyle\mathcal{G}_{\ell}} &\qw & \multigate{3}{\scriptstyle\mathcal{R}} &\qw\\
\lstick{} & & & \lstick{\scriptstyle\ket{x_{r_c+2}}}&\ghost{\Lambda_n}&\qw & & & \qw  & \ghost{\scriptstyle\mathcal{G}_1} & \ghost{\scriptstyle\mathcal{G}_2} &\qw & ... & &\qw & \ghost{\scriptstyle\mathcal{G}_{\ell}} &\qw & \ghost{\scriptstyle\mathcal{R}} &\qw\\
\lstick{} & & & \lstick{\scriptstyle\vdots}&\ghost{\scriptstyle\Lambda_n}&\qw & & & \qw & \ghost{\scriptstyle\mathcal{G}_1} & \ghost{\scriptstyle\mathcal{G}_2} &\qw & ... & &\qw & \ghost{\scriptstyle\mathcal{G}_{\ell}} &\qw & \ghost{\scriptstyle\mathcal{R}}&\qw\\
\lstick{}& &&\lstick{\scriptstyle\ket{x_n}}&\ghost{\scriptstyle\Lambda_n}&\qw  & & & \qw & \ghost{\scriptstyle\mathcal{G}_1} & \ghost{\scriptstyle\mathcal{G}_2} &\qw & ... & &\qw &\ghost{\scriptstyle\mathcal{G}_{\ell}}&\qw &\ghost{\scriptstyle\mathcal{R}}&\qw
\\
}
}
\caption{A circuit framework to implement an $n$-qubit unitary diagonal matrix $\Lambda_n$, where $r_t=\left\lfloor n/2\right\rfloor$, $r_c=n-r_t=\left\lceil n/2\right\rceil$ and $\ell\le  \frac{2^{r_t+2}}{r_t+1}-1$. The first $r_c$ qubits are control register and the last $r_t$ qubits are target register.
The depth of the operator $\mathcal{G}_k$ is $O(2^{r_c})$ for each $k\in[\ell]$ and the depth of the operator $\mathcal{R}$ is $O(r_t/\log r_t)$}. The $r_c$-qubit diagonal unitary matrix $\Lambda_{r_c}$ is implemented recursively.
    \label{fig:framework_DU_withoutancilla}
\end{figure}
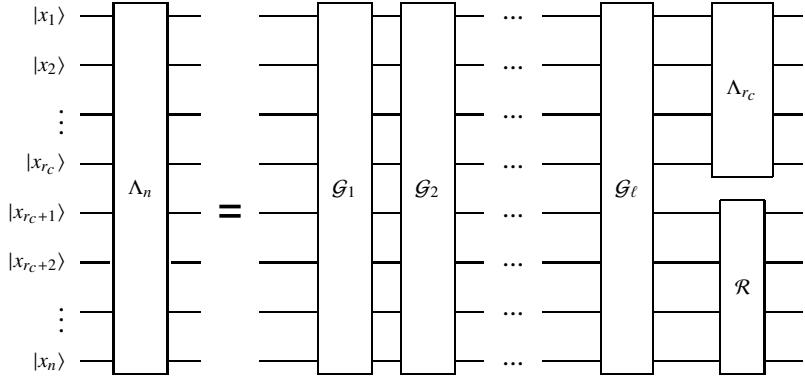


The $n$-qubit work register is divided into two registers: A \emph{control register} consisting of the first $r_c$ qubits, and a \emph{target register} consisting of the last $r_t$ qubits. The circuit has the following components.
\begin{enumerate}
    \item A sequence of $n$-qubit unitary operators $\mathcal{G}_1,\ldots,\mathcal{G}_\ell$, the detailed construction of which will be given in Section \ref{sec:G_k}.
    \item An $r_t$-qubit unitary operator $\mathcal{R}$, which resets the state in the target register to the input value $\ket{x_{r_c+1}, \ldots , x_n}$.
    \item An $r_c$-qubit diagonal unitary operator $\Lambda_{r_c}$, which is  implemented recursively.
\end{enumerate} 
The parameters are set as follows: 
$r_t=\lfloor n/2 \rfloor\approx n/2,\quad  r_c=n-r_t \approx n/2, \quad \text{ and } \quad \ell \le \frac{2^{r_t+2}}{r_t+1}-1 \approx \frac{ 2^{n/2+3}}{n}.$

%

Next we describe the function of each operator in Figure \ref{fig:framework_DU_withoutancilla}, for which it suffices to specify their effects on an arbitrary computational basis state
\[\ket{x} = \ket{x_1x_2\cdots x_{r_c} x_{r_c+1} \cdots x_n} = \ket{\underbrace{x_{control}}_{r_c \text{~qubits}}}\ket{\underbrace{x_{target}}_{r_t \text{~qubits}}},
\]
where $x\in\Bn$. Let us first highlight some key similarities and differences between this circuit and the one presented in the previous section. Recall that in Section \ref{sec:QSP_withancilla}, an $n$-bit string $s\in \{0,1\}^n-\{0^n\}$ is broken into two parts, a ${\left\lfloor\log (\frac{m}{2})\right\rfloor}$-bit prefix and an $\left(n-{\lfloor \log(\frac{m}{2}) \rfloor}\right)$-bit suffix. In the Gray Initial Stage there, we use $2^{\lfloor \log(m/2)\rfloor}$ qubits in the phase register to enumerate all possible ${\lfloor \log(m/2)\rfloor}$-bit prefixes, one prefix on each phase qubit $j$. Then on each such qubit $j$ we enumerate all $(n-{\lfloor \log(m/2)\rfloor})$-bit suffixes in the Gray Path Stage.
%
In this section, we again break $s$ into a prefix and a suffix, and enumerate all prefixes and all suffixes to run over all $n$-bit strings. However, due to the lack of the ancillary qubits, the circuit here differs from the last one in the following two aspects.
\begin{enumerate}
    \item In Section \ref{sec:QSP_withancilla}, $s\in\{0,1\}^n-\{0^n\}$ is generated 
    in the phase register, which is initialized to $|0\rangle$.  In this section, $s=ct$
    , in which $c$ is the $r_c$-bit prefix and $t$ is the $r_t$-bit suffix. The state $\ket{\langle s, x\rangle }$ is generated 
    in target register, whose initial state is $|x_j\rangle$ for some $j\in\{r_c+1, r_c+2, \ldots, r_n\}$. Hence, we enumerate $s$ recursively in this section. That is, we first generate $s=ct$ for $t\neq 0^{r_t}$ and then generate $c0^{r_t}$ recursively. 
    \item In Section \ref{sec:QSP_withancilla}, there are $2^{\lfloor \log(m/2)\rfloor}~(\le \frac{m}{2})$ prefixes which can be enumerated in $\frac{m}{2}$ qubits in phase register exactly. In this section, $2^{r_t}-1~(\approx 2^{n/2})$ suffixes should be generated in $r_t$ qubits in target register. As we only have $r_t$ qubits, the small space is insufficient to enumerate all $2^{r_t}-1$ suffixes. Thus we need to enumerate them in many stages, and $r_t$ suffixes in each stage; in other words, we pay the price of time to compensate the shortage of space. It turns out that the transition from one stage to another can be made in a low depth if the suffixes enumerated in each stage are linearly independent as vectors in $\{0,1\}^{r_t}$. Thus we need carefully divide $2^{r_t}-1$ suffixes into $\ell$ sets $T^{(1)}, \ldots, T^{(\ell)}$ with $T^{(k)}=\{t_1^{(k)},t_2^{(k)},\ldots,t_{r_t}^{(k)}\}$ each $t_a^{(k)}\neq 0^{r_t}$ for $a\in [r_t]$ and $k \in [\ell]$, and the strings in each $T^{(k)}$ linearly independent. We allow overlap between these sets, but maintain the total number $\ell$ of sets only a constant times of $(2^{r_t}-1)/r_t$, so that the overall depth is still under the control. As the sets have overlaps, a suffix may appear more than once, so we need to note this and avoid repeatedly applying rotation  when the suffix appears multiple times. 
\end{enumerate}

We now show how to implement the above high-level ideas. We will need to find sets $T^{(1)},T^{(2)},\ldots,T^{(\ell)}$ satisfying the following two key properties.
\begin{enumerate}
    \item For each $k\in[\ell]$, the set $T^{(k)} = \big\{t_1^{(k)},t_2^{(k)},\ldots,t_{r_t}^{(k)}\big\}$ contains $r_t$ vectors from $\{0,1\}^{r_t}$ that are linearly independent over the field $\mathbb{F}_2$.
    \item The collection of these sets covers all the $r_t$-bit strings except for $0^{r_t}$, i.e.  $\bigcup_{k\in[\ell]}T^{(k)} = \{0,1\}^{r_t}-\{0^{r_t}\}$.
\end{enumerate}
{The constructions of sets $T^{(1)},\ldots,T^{(\ell)}$ are shown in Appendix \ref{sec:partition}.} 
For each $k\in[\ell]\cup\{0\}$, define an $r_t$-qubit state
\begin{equation} \label{eq:yk}
\ket {y^{(k)}} = \ket{y_1^{(k)}y_2^{(k)}\cdots y_{r_t}^{(k)}},  \text{ where } \quad
    y_j^{(k)} =\left\{\begin{array}{ll}
       x_{r_c+j}  & \text{if~} k=0, \\
       \langle {0^{r_c}t_j^{(k)}},x\rangle  &  \text{if~} k\in[\ell].
    \end{array}\right.
\end{equation}
Namely, $y^{(0)}$ is the same as $x_{target}$ (the suffix of $x$), and other $y_j^{(k)}$ are linear functions of variables in $x_{target}$ with coefficients given by $t_j^{(k)}$.
Next, let us define disjoint families $F_1,\ldots,F_\ell$ which apply the rotation when a suffix appears for the first time.
\begin{equation}\label{eq:F_k}
\begin{array}{ll}
   F_1=\big\{ct:\ t\in T^{(1)},c\in\{0,1\}^{r_c}\big\},  \\
         F_k=\big\{ct:\ t\in T^{(k)},c\in\{0,1\}^{r_c}\big\}-\bigcup\limits_{d\in[k-1]}F_{d}, ~ 2\le k\le \ell.
\end{array}
\end{equation}
These families of sets $F_1,F_2,\cdots,F_\ell$ satisfy
$F_i\cap F_j =\emptyset$ for all $i\neq j \in[\ell]$ and 
\begin{equation}\label{eq:set_eq}
    \bigcup_{k\in [\ell]} F_k 
   = \B^{r_c}\times \bigcup\limits_{k\in [\ell]} T^{(k)}
   =\B^{r_c}\times (\B^{r_t} - \{0^{r_t}\}) 
   =  \Bn-\{c0^{r_t}:\ c\in\{0,1\}^{r_c}\}.       
\end{equation}

With the above concepts, we can now show the desired effect of the operators $\mathcal{G}_k$, $\mathcal{R}$ and $\Lambda_{r_c}$. 
\begin{enumerate}
    \item For $k\in[\ell]$, 
    \begin{equation}\label{eq:Gk}
        \mathcal{G}_k\ket{x_{control}}\ket{y^{(k-1)}}=e^{i\sum\limits_{s \in F_k}\langle s,x\rangle \alpha_s } \ket{x_{control}}\ket{y^{(k)}},
    \end{equation}
    where $\alpha_s$ is determined by Eq. \eqref{eq:alpha}. In words, $\mathcal{G}_k$ has two effects: (1) It puts a phase and 
    (2) it transits from the stage $k-1$ to the stage $k$. 

    \item The transformation $\mathcal{R}$ acts on the target register and resets the suffix state as follows
    \begin{equation}\label{eq:reset}
        \mathcal{R}\ket{y^{(\ell)}}=\ket{y^{(0)}}.
    \end{equation}
    As a map on ${\{0,1\}^{r_t}}$ (instead of $\{\ket{x}: x\in  {\{0,1\}^{r_t}}\}$), $\mathcal R$ is an invertible linear transformation over $\mathbb{F}_2$.
    \item The operator $\Lambda_{r_c}$ is an $r_c$-qubit diagonal matrix satisfying that 
    \begin{equation}\label{eq:Lambda_rc}
        \Lambda_{r_c}\ket{x_{control}}= e^{i\sum\limits_{c\in \{0,1\}^{r_c}-\{0^{r_c}\}}\langle c0^{r_t},x\rangle\alpha_{c0^{r_t}}}\ket{x_{control}},
    \end{equation}
and will be implemented recursively. 
\end{enumerate}
We will define these operators and show these properties in Section \ref{sec:G_k}.
\subsection{Correctness and depth}
\label{sec:correctness}
In this section, we will prove the correctness and analyze the depth of the circuit. We will need a fact about the depth of invertible linear transformation from \cite{jiang2020optimal} (Theorem 1). The original version says that any CNOT circuit, a circuit consisting of only CNOT gates, on $n$ qubits can be compressed into  $O(n/\log n)$ depth. But note that any $n$-dimensional invertible linear transformation over $\mathbb{F}_2$ can be implemented by a CNOT circuit {\cite{patel2008optimal}}. We thus have the following result.
\begin{lemma} 
Suppose that $U\in \mathbb{F}_2^{n \times n}$ is an invertible linear transformation over $\mathbb{F}_2$. Then as a $2^n\times 2^n$ unitary matrix which permutes computational basis $\{\ket{x}: x\in \Bn\}$, the map $U$ can be realized by a CNOT circuit of depth at most $O(\frac{n}{\log n})$ { and size at most $O(\frac{n^2}{\log n})$} without ancillary qubits.
 \label{lem:SODA2020}
\end{lemma}
As mentioned in Section \ref{sec:framework_DU_withoutancilla}, $\mathcal R$ is an invertible linear transformation on the computational basis variables, thus the above lemma immediately implies the following depth upper bounds for $\mathcal{R}$.
\begin{lemma} \label{lem:R}
The operator $\mathcal{R}$ can be realized by an $O(\frac{r_t}{\log r_t})$-depth   { and $O(\frac{r_t^2}{\log r_t})$-size} CNOT circuit without ancillary qubits.
\end{lemma}

The depth of $\mathcal{G}_k$ will be easily seen from its construction in Section \ref{sec:G_k}.
\begin{lemma}\label{lem:Gk}
The operator $\mathcal{G}_k$ can be realized by an $O( 2^{r_c})$-depth {and $O(r_c 2^{r_c+1})$-size} quantum circuit using single-qubit and CNOT gates without ancillary qubits. 
\end{lemma}

Now we are ready to prove the correctness and depth of the whole circuit. The correctness of the circuit framework in Figure \ref{fig:framework_DU_withoutancilla} is shown in Appendix \ref{sec:correctness_withoutancilla}.
\begin{lemma}
Any diagonal unitary matrix $\Lambda_n$ can be realized by the quantum circuit  $(\Lambda_{r_c}\otimes\mathcal{R})\mathcal{G}_{\ell}\mathcal{G}_{\ell-1}\cdots\mathcal{G}_{1}$ as in Figure \ref{fig:framework_DU_withoutancilla}, which has depth $O(2^n/n)$ and size $2^{n+3}+O\big(\frac{n^2}{\log n}\big)$ and uses no ancillary qubits.
\end{lemma}
\begin{proof}
We prove that the circuit has depth $D(n) = O(2^n/n)$. 
Lemma \ref{lem:Gk} shows $\mathcal{G}_k$ can be realized in depth at most $\lambda_1\cdot 2^{r_c}$ for a constant $\lambda_1>0$ and Lemma \ref{lem:R} shows $\mathcal{R}$ can be implemented in depth at most $\lambda_2\cdot \frac{r_t}{\log r_t}$ without ancillary qubits for a constant $\lambda_2>0$. Therefore, $D(n)$ satisfies the following recurrence
\[
\begin{array}{ll}
    D(n) & \le \max \big\{D(r_c),~\lambda_2\cdot\frac{r_t}{\log r_t}\big\}+\lambda_1\cdot2^{r_c}\cdot \ell\\
    & \le D(\lceil n/2 \rceil)+ \frac{\lambda_2\lceil n/2 \rceil}{\log \lceil n/2 \rceil}+\lambda_1 2^{\lceil n/2
\rceil} \big( \frac{2^{\lfloor n/2 \rfloor+2}}{\lfloor n/2 \rfloor+1}-1\big)\\
     & = D(\lceil n/2 \rceil)+O(2^n/n).\\
\end{array}
\]
Solving the above recursive relation, we obtain the bound $D(n)=O(2^n/n)$ as desired. The size of this circuit $S(n)$ satisfies $S(n)\le S(n/2)+(2^{n+3}-2^{n/2+3})+O\big(\frac{n^2}{\log n}\big)\le 2^{n+3} + O\big(\frac{n^2}{\log n}\big)$.
\end{proof}

\subsection{Construction of $\mathcal{G}_k$ and $\mathcal R$}
\label{sec:G_k}

In this section, we will show how to construct operator $\mathcal{G}_k$, which consists of two stages: Generate Stage and Gray Path Stage, see Figure \ref{fig:all_com_tar}. Along the way, we will also show the construction of $\mathcal R$.


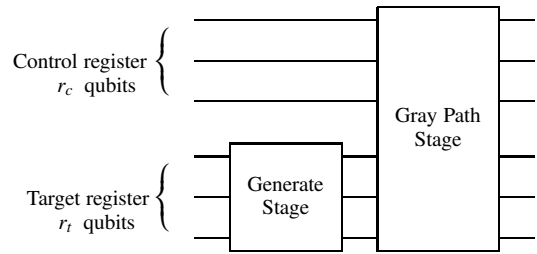
\begin{figure}[!ht]
    \centerline{
\Qcircuit @C=1.2em @R=0.5em {
& \qw & \multigate{6}{\text{Gray Path}\atop \text{Stage}} & \qw\\
& \qw & \ghost{\text{Gray Path}\atop \text{Stage}} & \qw \\
& \qw & \ghost{\text{Gray Path}\atop \text{Stage}} & \qw \inputgroupv{2}{2}{2.2em}{0.5em}{\text{Control register~~~~~~~~~~~~~~~~~~~~~~~~} \atop r_c \text{~~qubits}~~~~~~~~~~~~~}\\
& &  & \\
& \multigate{2}{ \text{Generate}\atop \text{Stage}} & \ghost{\text{Gray Path}\atop \text{Stage}} & \qw \\
& \ghost{\text{Generate}\atop \text{Stage}} & \ghost{\text{Gray Path}\atop \text{Stage}} & \qw \\
& \ghost{\text{Generate}\atop \text{Stage}} & \ghost{\text{Gray Path}\atop \text{Stage}} & \qw \inputgroupv{6}{6}{2.2em}{0.5em}{\text{Target register~~~~~~~~~~~~~~~~~~~~~} \atop r_t \text{~~qubits}~~~~~~~~~~~~~}\\
}
}
\caption{Implementation of operator 
$\mathcal{G}_k$, which consists of Generate Stage and Gray Path Stage. The depth of Generate Stage is $O\left(\frac{r_t}{\log r_t}\right)$ and the depth of Gray Path Stage is $O(2^{r_c})$. 
}
\label{fig:all_com_tar}
\end{figure}

\paragraph{Generate Stage} In this stage, we implement operator $U_{Gen}^{(k)}$, such that
\begin{equation} \label{eq:y_k-1}
\ket{y^{(k-1)}} \xrightarrow{U_{Gen}^{(k)}}\ket{y^{(k)}}, ~k\in[\ell],   
\end{equation}
\noindent where $y^{(k-1)}$ and $y^{(k)}$ are defined in Eq. \eqref{eq:yk} and determined by ${T}^{(k-1)}$ and ${T}^{(k)}$, respectively. For $k\in[\ell]$, recall that ${T}^{(k)} = \{t_1^{(k)},\cdots,t_{r_t}^{(k)}\}$. Fix this ordering, view each {$t_i^{(k)}$} as a column vector, and define a matrix $\hat{T}^{(k)}=[t_1^{(k)},\cdots,t_{r_t}^{(k)}]^T\in\{0,1\}^{r_t\times r_t}$ for $k\in[\ell]$, with special case $\hat{T}^{(0)} \defeq I_{r_t}$. Then the vectors $y^{(k)}$ can be rewritten as
\begin{equation}\label{eq:y-T}
 y^{(k)}=\hat{T}^{(k)}x_{target}, \quad\forall k\in [r_t]\cup \{0\}.
\end{equation} 
Since $t_1^{(k)},t_2^{(k)},\cdots,t_{r_t}^{(k)}$ are linearly independent over $\mathbb{F}_2$, $\hat{T}^{(k)}$ is an invertible linear transformation over $\mathbb{F}_2$. Now define a unitary $U_{Gen}^{(k)}$ by
$U_{Gen}^{(k)}\ket{ y} = \ket{ \hat{T}^{(k)}(\hat{T}^{(k-1)})^{-1} y}$, where the matrix-vector multiplication at the right hand side is over $\mathbb{F}_2$. From Eq. \eqref{eq:y-T}, we see that 
\begin{equation*}
    U_{Gen}^{(k)} \ket{y^{(k-1)}} = \ket{\hat{T}^{(k)}(\hat{T}^{(k-1)})^{-1} y^{(k-1)}} = \ket{\hat{T}^{(k)}x_{target}} = \ket{y^{(k)}}
\end{equation*}
satisfying Eq. \eqref{eq:y_k-1}.
Also note that when viewed as a linear transformation over $\mathbb{F}_2$, $U_{Gen}^{(k)}$ is invertible. Thus according to Lemma \ref{lem:SODA2020}, the following depth upper bound applies.
\begin{lemma}\label{lem:generate_withoutancilla}
The Generate Stage unitary $U_{Gen}^{(k)}$ can be realized by an $O\big(\frac{r_t}{\log r_t}\big)$-depth and $O\big(\frac{r_t^2}{\log r_t}\big)$-size CNOT circuit without ancillary qubits.
\end{lemma}

Similar to the discussion of $U^{(k)}_{Gen}$, operator $\mathcal{R}$ can be defined by 
$\mathcal R \ket{y} = \ket{ (\hat{T}^{(\ell)})^{-1} y},$
then 
$\mathcal R \ket{y^{(\ell)}} = \ket{(\hat{T}^{(\ell)})^{-1}  y^{(\ell)}} = \ket{x_{target}} = \ket{y^{(0)}}.$
Thus $\mathcal{R}$ can be also viewed an invertible linear transformation over $\mathbb{F}_2$. Applying Lemma \ref{lem:SODA2020} gives the bound in Lemma \ref{lem:R}.


\paragraph{Gray Path Stage} This stage implements the following operator
\begin{equation}\label{eq:Graycodepath}
\ket{x_{control}}\ket{y^{(k)}}\xrightarrow{U_{GrayPath}}e^{i\sum\limits_{s \in F_k}\langle s,x\rangle \alpha_s
} \ket{x_{control}}\ket{y^{(k)}},
\end{equation}
where $k\in [\ell]$ and $F_k$ is defined in Eq. \eqref{eq:F_k}.
The Gray Path Stage in this section is similar to the Gray Path Stage in  Section \ref{sec:QSP_withancilla}, though we need to use a Gray code cycle here instead of a Gray code path.
For every $i\in[r_t]$, let $c^i_1,c^i_2,\cdots,c^i_{2^{r_c}-1}, c^i_{2^{r_c}}$
denote the $i$-Gray code of $r_c$ bits starting at $c_1^i=0^{r_c}$ for $i\in[r_t]$. 
Let $h_{ij}$ denote the index of the bit that $c^i_{j-1}$ and $c^i_{j}$ differ for each $j\in\{2,3,\ldots,2^{r_c}\}$ and $h_{i1}$ the index of the bit that $c^i_1$ and $c^i_{2^{r_c}}$ differ. For the $i$-Gray code cycle of $r_c$ bits, 
\begin{equation}\label{eq:index}
    h_{ij}=\left\{\begin{array}{ll}
    (r_c+i-2\mod r_c)+1, & \text{if~} j = 1\\
    (\zeta(j-1)+i-2\mod r_c)+1, & \text{if~} j \neq 1
    \end{array}\right.
\end{equation}
The exact form of $h_{ij}$ is not crucial; the important fact to be used later is that the indices $h_{1p},h_{2p},\ldots,h_{r_tp}$ are all different.

This stage consists of $2^{r_c}+1$ phases.
\begin{enumerate}
    \item In phase 1, circuit $C_1$ applies a rotation $R(\alpha_{0^{r_c}t_i^{(k)}})$ on the $i$-th qubit in the target register for all $i\in[r_t]$ if the string $0^{r_c}t_i^{(k)} \in F_k$, where $\alpha_{0^{r_c}t_i^{(k)}}$ is defined in Eq. \eqref{eq:alpha}.
    \item In phase $p\in\{2,\ldots,2^{r_c}\}$, circuit $C_{p}$ consists of 2 steps:\begin{enumerate}
        \item Step $p.1$ is a unitary that, for all $i\in[r_t]$, applies a CNOT gate on the $i$-th qubit in target register, controlled by the $h_{ip}$-th qubit in control register.
        \item Step $p.2$ is a unitary that, for all $i\in[r_t]$, applies a rotation $R(\alpha_{c^i_pt_{i}^{(k)}})$ on the $i$-th qubit in target register if $c^i_pt_{i}^{(k)}\in F_k$, where $\alpha_{c_p^{i}t_i^{(k)}}$ is defined in Eq. \eqref{eq:alpha}.
    \end{enumerate}
    \item In phase $2^{r_c}+1$, circuit $C_{2^{r_c}+1}$ implements a unitary 
    that, for all $i\in[r_t]$, applies a CNOT gate on the $i$-th qubit in target register, controlled by the $h_{i1}$-th qubit in control register .
\end{enumerate}
The next lemma gives the correctness and depth of this constructed circuit. {The proof of Lemma \ref{lem:graypath_withoutancilla} is shown in Appendix \ref{sec:graypath_withoutancilla}.}
\begin{lemma}\label{lem:graypath_withoutancilla}
The quantum circuit defined above is of depth $O(2^{r_c})$ and size $O(r_c2^{r_c+1})$, and implements Gray Path Stage $U_{GrayPath}$ in Eq. \eqref{eq:Graycodepath}.
\end{lemma}
According to Lemma \ref{lem:generate_withoutancilla} and Lemma \ref{lem:graypath_withoutancilla}, operator $\mathcal{G}_k$ can be implemented in depth $O(2^{r_c})+O(\frac{r_t}{\log r_t})=O(2^{r_c})$. And the size of the circuit is at most $O(\frac{n^2}{\log n})+r_c2^{r_c+1}=O(r_c2^{r_c+1})$. This completes the proof of  Lemma \ref{lem:Gk}.

%% file: QSP_withmoreancilla.tex
In this section, we will introduce a different framework that can improve the upper bound in Section \ref{sec:QSP_withancilla} when the number of ancillary qubits $m=\Omega(2^n/n^2)$. 
In Section \ref{sec:new_framework}, we will present the framework, and in Section \ref{sec:new_framework_correctness}, we will give implementation details with the depth and correctness analyzed.

In the following, we will use $e_i\in \{0,1\}^{2^n}$ to denote the vector where the $i$-th bit is 1 and all other bits are 0. It is a unary encoding of $i\in \{0,1,\ldots,2^n-1\}$, and $\ket{e_i}$ is the corresponding $2^n$-qubit state. We use $n$-qubit state $\ket{i} = \ket{i_0 i_1 \cdots i_{n-1}} \in (\{\ket{0}, \ket{1}\})^{\otimes n}$ to denote the binary encoding of $i$, where $i_0,\cdots,i_{n-1}\in\{0,1\}$ and $i=\sum_{j=0}^{n-1} i_j\cdot 2^j$.

\subsection{New framework for quantum state preparation}
\label{sec:new_framework}

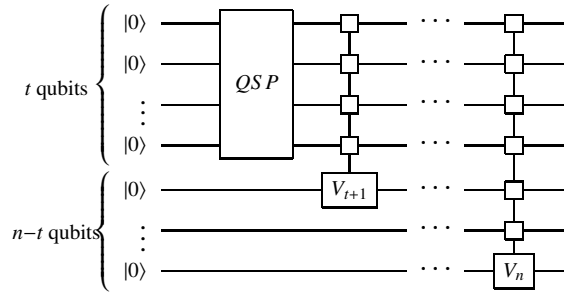
\begin{figure}[ht]

\centerline{
\Qcircuit @C=1em @R=0.5em {
\lstick{\scriptstyle\ket{0}} & \qw  & \multigate{3}{\scriptstyle QSP} &\gate{} \qwx[1] &\qw& \cdots & &\gate{} \qwx[1] &  \qw\\
\lstick{\scriptstyle\ket{0}} & \qw   & \ghost{\scriptstyle QSP} &\gate{} \qwx[1] &\qw& \cdots & &\gate{} \qwx[1]&  \qw\\
\vdots~~~~~&  \qw   & \ghost{\scriptstyle QSP}& \gate{} \qwx[1] &\qw& \cdots & &\gate{} \qwx[1]& \qw\\
\lstick{\scriptstyle\ket{0}} &  \qw   & \ghost{\scriptstyle QSP} &\gate{} \qwx[1] &\qw& \cdots & &\gate{} \qwx[1]&  \qw \\
\lstick{\scriptstyle\ket{0}} &  \qw & \qw  & \gate{\scriptstyle V_{t+1}} &\qw& \cdots & &\gate{} \qwx[1]&  \qw \\
\vdots~~~~~ &  \qw & \qw  & \qw &\qw& \cdots & &\gate{} \qwx[1]&  \qw \inputgroupv{2}{3}{4em}{0.9 em}{ \scriptstyle t \text{~qubits}~~~~~~~~~~~~~~~} \inputgroupv{6}{6}{4em}{0.1 em}{ \scriptstyle n-t \text{~qubits}~~~~~~~~~~~~~~~}\\
\lstick{\scriptstyle\ket{0}} &  \qw & \qw & \qw &\qw& \cdots& &\gate{\scriptstyle V_n}&  \qw \\
}
}
\label{fig:new_one}
\caption{A new circuit framework to prepare an $n$-qubit quantum state $\ket{\psi_{v}}$ with $m\in[\Omega(2^n/n^2),3\cdot2^n]$. Let $t=\left\lfloor\log(m/3)\right\rfloor$. The first $t$-qubit unitary $QSP$ implement the same transformation as the first $t$ UCGs in Figure \ref{fig:QSP_circuit}(a).  The last $n-t$ UCGs in are the same as the last $n-t$ UCGs in Figure \ref{fig:QSP_circuit}(a).}
\label{fig:new_framework}
\end{figure}

The quantum circuit in Section \ref{sec:preliminaries} for quantum state preparation consists of $n$ UCGs $V_1,V_2,\ldots,V_n$ (Figure \ref{fig:QSP_circuit}(a)). In Section \ref{sec:QSP_withancilla}, we showed that any $j$-qubit UCG $V_j$ can be implemented by a quantum circuit of depth $O\big(j+\frac{2^j}{m+j}\big)$ with $m$ ancillary qubits. Summing this up over $j\in [n]$ gives the $O(n^2 + 2^n/m)$ upper bound for QSP, and this quadratic term seems hard to be improved within the framework of \cite{grover2002creating}. In the new framework, we first generate the quantum state in the unary encoding $\sum_i v_i \ket{e_i}$ using the result in \cite{johri2021nearest}, and then make an encoding transform $\ket{e_i} \to \ket{i}$, from the unary encoding to the binary encoding. 

Two issues need to be handled here. The first one is the need to design an encoding transform circuit that has small depth and size, using ancillary qubits efficiently. We will give an optimal construction in Section \ref{sec:new_framework_correctness}. The second issue is
that the unary encoding itself needs $2^n$ qubits, and the encoding transform also needs $O(2^n)$ qubits, which may be beyond $m$, the number of ancillary qubits that are available in the first place.
To handle this, we will use a hybrid method. We break the generation into a prefix part and a suffix part, where the length of the prefix is whatever $m$ can support. We prepare the prefix part by unary QSP construction in \cite{johri2021nearest} and our encoding transformation, and then employ the methods in Section \ref{sec:QSP_withancilla} for the suffix part.

{ Our new circuit framework for QSP in the parameter regime $m=\Omega(2^n/n^2)$ is shown in Figure \ref{fig:new_framework}. Let $t=\lfloor \log (m/3)\rfloor$. 
In previous framework, the first $t$ UCGs are a QSP circuit to prepare a $t$-qubit quantum state 
$\ket{\psi_{v}^{(t)}}=\sum_{k=0}^{2^{t}} v'_k\ket{k}\text{, where }v'_k=\sqrt{\sum_{j=0}^{2^{n-t}-1}|v_{2^{n-t}k+j}|^2}.$
In the new framework, we introduce a new $t$-qubit QSP circuit to replace the first $t$ UCGs. The new QSP circuit consists of the following steps. 
\begin{enumerate}
    \item Generate a $2^{t}$-qubit quantum state $\ket{\psi'_{v}}=\sum_{k=0}^{2^{t}-1} v'_k\ket{e_k}$, where $e_k\in\{0,1\}^{2^{t}}$ and by the quantum circuit in \cite{johri2021nearest}. 
    \item Applying $U_{t}$ to $\ket{\psi'_{v}}$, we can obtain $\ket{\psi_{v}^{(t)}}$ with $2m/3$ ancillary qubits, where $U_{t}$ is the unitary transformation $U_{t}:\ket{e_i}\to\ket{i}\ket{0^{2^{t}-{t}}}$ for all $i\in\{0\}\cup [2^{t}-1]$. 
    \item Realize the last $n-t$ UCGs by Eq. \eqref{eq:UCG} and Lemma \ref{lem:DU_with_ancillary}. 
\end{enumerate}}

\subsection{Implementation and analysis}
\label{sec:new_framework_correctness}\label{lem:linear_depth_tof}
Now we give a more detailed implementation and analyze the correctness and cost of the algorithm. 
First, in \cite{johri2021nearest} it is shown that QSP with the unary encoding can be implemented efficiently. 
\begin{lemma}\label{lem:johr}
Given a vector $v=(v_0,v_1,\ldots, v_{2^n-1})^T\in\mathbb{C}^{2^n}$ with unit $\ell_2$-norm, any $2^n$-qubit quantum state
$
\ket{\psi'_{v}}=\sum_{k=0}^{2^n-1}v_k\ket{e_k}
$
can be prepared from the initial state $\ket{0}^{\otimes 2^n}$ by a quantum circuit using single-qubit gates and CNOT gates of depth $O(n)$
{ and size $O(2^n)$} { without} ancillary qubits. 
\label{lem:unary_qsp}
\end{lemma}

Next we consider the encoding transformation. 
\begin{lemma}\label{lem:enc-trans}
The following unitary transformation on $2^n$ qubits
\begin{equation}\label{eq:unary_transform}
\ket{e_i}\to|i\rangle\ket{0^{2^n-n}},\forall i\in\{0\}\cup [2^n-1], e_i\in\{0,1\}^{2^n},
\end{equation}
can be implemented by a quantum circuit using single-qubit gates and CNOT gate with  {$2^{n+1}$} ancillary qubits, of depth $O(n)$ {and size $O(2^n)$}.
\end{lemma}
{The proof of Lemma \ref{lem:enc-trans} is shown in Appendix \ref{sec:enc-trans}.}
Now we are ready to give the hybrid algorithm and cost analysis.
\begin{lemma}\label{lem:QSP_moreancilla}
    For any $m\in [\Omega(2^n/n^2),3\cdot2^{n}]$, any $n$-qubit quantum state $\ket{\psi_v}$ can be generated by a quantum circuit, using single-qubit gates and CNOT gates, of depth $O\big(n(n-\log (m/3) +1)+\frac{2^n}{m}\big)$ and size $O(2^{n})$  with $m$ ancillary qubits.
\end{lemma}
\begin{proof}
Let $t=\lfloor\log\frac{m}{3})\rfloor$. Define a quantum state $\ket{\psi_{v}^{(t)}} = \sum_{i=0}^{2^{t}-1} v'_i \ket{i}$, where $v'_i=\sqrt{\sum_{j=0}^{2^{n-t}-1}|v_{i\cdot 2^{n-t}+j}|^2}$. Note that $\ket{\psi_{v}^{(t)}} = V_t V_{t-1} \cdots V_1 \ket{0}^{\otimes n}$, the state after we apply the first $t$ UCGs in Figure \ref{fig:QSP_circuit}(a). 

According to Lemma \ref{lem:johr}, given the unit vector $v' = (v_0', \ldots, v_{2^t-1}')$, we can prepare a $2^t$-qubit quantum state $\ket{\psi'_v}=\sum_{i=0}^{2^t-1}v_i'\ket{e_i}$ by a quantum circuit of depth $O(t) = O(n)$  and size $O(2^t) = O(2^n)$. The resulting state is on $2^t$ qubits. Then we apply the unitary transform Eq. \eqref{eq:unary_transform} in Lemma \ref{lem:enc-trans} to transform the unary encoding to a binary encoding and obtain  $\ket{\psi_v^{(t)}} = \sum_{i=0}^{2^t-1} v_i' \ket{i}$. This transformation has depth $O(t)$ and size $O(2^t)$, and need $2^{t+1}$ ancillary qubits.
The whole process can be carried out in a work space of $2^t + 2^{t+1} \le m$ qubits.

To change $\ket{\psi_{v}^{(t)}}$ to the final target state $\ket{\psi_{v}}$, what is left is to apply $V_{t+1}, \ldots, V_n$ to $\ket{\psi_{v}^{(t)}}$. By Lemma \ref{lem:UCG_depth}, each $V_j$ can be implemented by a circuit of depth $O(j+\frac{2^j}{m})$ and size $O(2^j)$ by $m$ ancillary qubits. Hence $V_n\cdots V_{t+1}$ can be realized by a quantum circuit of depth $\sum_{j=t +1}^{n}O\big(j+\frac{2^j}{m}\big) = O\big(n(n-\lfloor\log(m/3)\rfloor)+\frac{2^n}{m}\big)$, { and size $\sum_{j=t+1}^n O(2^j)=O(2^n)$}, with $m$ ancillary qubits.

Combining the two steps, we see that the total depth and size of this quantum state preparation circuit are $O\big(n(n-\log (m/3) +1)+\frac{2^n}{m}\big)$ and $O(2^n)$, respectively.
\end{proof}
Note that when $m = 3\cdot 2^n$, the depth bound becomes $O(n)$. And if $m$ is even larger, then we can choose to only use $3\cdot 2^n$ of them. Thus we have the following result, {which is Theorem \ref{thm:QSP_anci} in the parameter regime $m = \Omega(2^n/n^2)$.}
\begin{corollary}\label{coro:QSP_moreancilla}
For a circuit preparing an $n$-qubit quantum state with $m=\Omega(2^n/n^2)$ ancillary qubits, the minimum depth $D_{\textsc{QSP}}(n,m)$ for different ranges of $m$ are characterized as follows:
\[
   \left\{ \begin{array}{ll}
         O(2^n/m), & \text{if~} m\in[\Omega(2^n/n^2),O(2^n/(n\log n))],\\
          O(n\log n), & \text{if~} m\in[\omega(2^n/(n\log n)),o(2^n)],\\
          O(n), & \text{if~} m=\Omega(2^n).
    \end{array}\right.
\]
\end{corollary}

%% file: QSP_extensions.tex


\subsection{Implications on optimality of unitary depth compression}
In this section, we will show that our results for QSP can be applied to general unitary synthesis. The proofs of Theorem \ref{thm:unitary} and Corollary \ref{coro:unitary} are shown in Appendix \ref{sec:US_depth}.

\begin{theorem}\label{thm:unitary}
Any unitary matrix $U\in\mathbb{C}^{2^n\times 2^n} $ can be implemented by a quantum circuit of depth $O\big(n2^n+\frac{4^n}{m+n}\big)$ {and size $O(4^n)$} with $m \le 2^n$ ancillary qubits.
\end{theorem}


In \cite{shende2004minimal}, it was shown that one needs at least $\Omega\left(4^n\right)$ CNOT gates to implement an arbitrary $n$-qubit unitary matrix without ancillary qubits. In the proof, the authors first put the circuit in a form that all single-qubit gates are immediately before either a CNOT gate or the output. It is known that such a CNOT gate together with its two single-qubit incoming neighbor gates can be specified by 4 free real parameters, and that each single-qubit gate right before the output has 3 free real parameters. Thus overall the circuit has $4k+3n$ parameters where $k$ is the number of CNOT gates. To generate all $n$-qubit states, the set of which is known to have dimension $4^n-1$, we need $4k+3n \ge 4^n-1$. Thus the bound follows. {This argument basically applies to quantum circuits with ancillary qubits as well, as stated in the next corollary, which shows that our circuit construction for general unitary matrices is asymptotically optimal for $m = O(2^n/n)$.}

\begin{corollary}\label{coro:unitary}
The minimum circuit depth $D_{\textsc{Unitary}}(n,m)$ for an arbitrary $n$-qubit unitary with $m$ ancillary qubits satisfies
\[
\begin{cases}
    D_{\textsc{Unitary}}(n,m) = \Theta\big(\frac{4^n}{m+n}\big), & \text{if } m=O(2^n/n), \\
    D_{\textsc{Unitary}}(n,m) \in \left[ \Omega(n),  O(n2^n)\right], & \text{if }m = \omega(2^n/n).
\end{cases}
\]
\end{corollary} 

\subsection{Decomposition with Clifford + T gate set}
The quantum gate set $\{{CNOT},H,S,T\}$, sometimes called Clifford+T gate set, is a universal gate set in that any unitary matrix can be approximately implemented using these gates only. The gates in this set all have a fault-tolerant implementation, thus the gate set is considered as one of the most promising candidates for practical quantum computing.  In this section we consider the circuits using only
the gates in this set. 
\begin{definition}[$\epsilon$-approximation]
For any $\epsilon>0$, a unitary matrix $U$ is $\epsilon$-approximated by another unitary matrix $V$ if 
\[\|U-V\|_2 \defeq \max_{\||\psi\rangle\|_2=1} \|(U-V)|\psi\rangle\|_2 < \epsilon.\]
\end{definition}


We can extend our results on the exact implementations of state preparation and unitary to their approximate versions. 

The following two corollaries are circuit implementations for quantum state preparation (Corollary \ref{coro:psi_apprx}) and unitary synthesis (Corollary \ref{coro:unitary_apprx}). The Corollary \ref{coro:psi_apprx} is a restatement of Corollary \ref{coro:approx_QSP}. {The proofs are shown in Appendix \ref{sec:clifford_decomposition}.}
\begin{corollary}\label{coro:psi_apprx}
    For any $n$-qubit target state $\ket{\psi_v}$ and $\epsilon>0$, one can prepare a state $\ket{\psi'_v}$ which is $\epsilon$-close to $\ket{\psi_v}$ in $\ell_2$-distance, by a quantum circuit consisting of $\{CNOT,H,S,T\}$ gates of depth 
    {\[\left\{\begin{array}{ll}
        O\big(\frac{2^n\log(2^n/\epsilon)}{m+n}\big) &  \text{if~}m=O(2^n/(n\log n)),\\
         O(n\log n\log(2^n/\epsilon))& \text{if~}m\in [\omega(2^n/(n\log n),o(2^n)],\\
         O(n \log(2^n/\epsilon))& \text{if~}m=\Omega(2^n),\\
    \end{array}\right.\]}
where $m$ is the number of ancillary qubits.
\end{corollary}

The following is an implementation of a unitary matrix.
\begin{corollary}\label{coro:unitary_apprx}
Any $n$-qubit general unitary matrix can be implemented by a circuit, using the $\{CNOT,H,S,T\}$ gate set, of depth $O\big(n2^n+\frac{4^n\log(4^n/\epsilon)}{m+n}\big)$ with $m$ ancillary qubits.
\end{corollary}
\noindent{\bf Remark.} Our circuits for general states can be also extended to circuits for sparse states. See details in Appendix \ref{sec:sparse_QSP}.

%% file: QSP_append.tex
\section{Circuit depth lower bound}
\label{sec:QSP_lowerbound}
\input{QSP_lowerbound}

\section{Quantum state preparation via Binary search tree}
\label{sec:app_BST}
The framework of quantum state preparation is illustrated by a vector
\[    \nu=\left(\sqrt{0.03},\sqrt{0.07},\sqrt{0.15},\sqrt{0.05},\sqrt{0.1},\right.\left.\sqrt{0.3},\sqrt{0.2},\sqrt{0.1}\right)^T \in\mathbb{C}^8.
\]
The corresponding quantum state is
a $3$-qubit quantum state 
\begin{align*}
    |\psi_\nu\rangle = & \sqrt{0.03}|000\rangle+\sqrt{0.07}|001\rangle+\sqrt{0.15}|010\rangle+\sqrt{0.05}|011\rangle\\
    & +\sqrt{0.1}|100\rangle+\sqrt{0.3}|101\rangle+\sqrt{0.2}|110\rangle+\sqrt{0.1}|111\rangle.
\end{align*}
The amplitudes of $|\psi_v\rangle$ are stored in the leaf nodes of the corresponding Binary Search Tree. Every internal node stores the square root of sum of squares of its child nodes. The root node stores the $\ell_2$-norm of the vector.

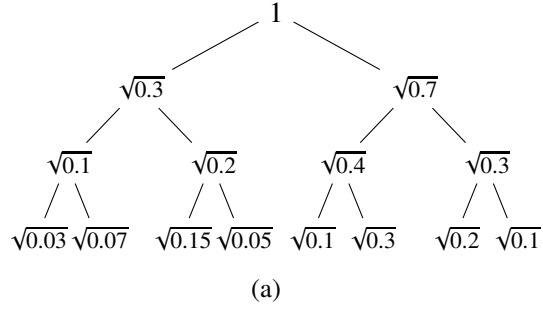
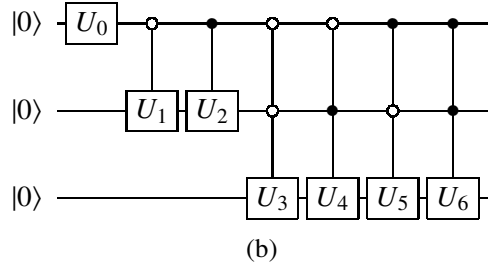
\begin{figure}[!ht]
\centering
\subfloat[]{{
    \centering
    \begin{tikzpicture} [level distance=1cm,
        level 1/.style={sibling distance=3.6cm},
        level 2/.style={sibling distance=1.9cm},
        level 3/.style={sibling distance=0.8cm},]
        \node {$1$}
            child {node[scale=0.8] {$\sqrt{0.3}$}
                child {node[scale=0.8] {$\sqrt{0.1}$} child{node[scale=0.8] {$\sqrt{0.03}$}} child{node[scale=0.8]{$\sqrt{0.07}$}}}
                child {node[scale=0.8] {$\sqrt{0.2}$} child{node[scale=0.8] {$\sqrt{0.15}$}}child{node[scale=0.8]{$\sqrt{0.05}$}}}
            }
            child {node[scale=0.8] {$\sqrt{0.7}$}
                child {node[scale=0.8] {$\sqrt{0.4}$} child{node[scale=0.8] {$\sqrt{0.1}$}}child{node[scale=0.8]{$\sqrt{0.3}$}}}
                child {node[scale=0.8] {$\sqrt{0.3}$} child{node[scale=0.8] {$\sqrt{0.2}$}}child{node[scale=0.8]{$\sqrt{0.1}$}}}
            };
    \end{tikzpicture}
}}    \label{fig:bst}
\hfill
\subfloat[]{{
\centerline{
\Qcircuit @C=.3em @R=1.6em {
&  &  &  &  &  &  &  &  \\
\lstick{\ket{0}} & \gate{U_0} & \ctrlo{1}  & \ctrl{1} & \ctrlo{1} & \ctrlo{1} & \ctrl{1} & \ctrl{1} & \qw \\
\lstick{\ket{0}} & \qw & \gate{U_1} & \gate{U_2} & \ctrlo{1} & \ctrl{1} & \ctrlo{1} & \ctrl{1} & \qw\\
\lstick{\ket{0}} & \qw & \qw & \qw & \gate{U_3} & \gate{U_4} &\gate{U_5} & \gate{U_6} &\qw \\
}
}}}
\label{fig:circuit_psi}
\caption{(a) Binanry search tree for vector $\nu=\left(\sqrt{0.03},\sqrt{0.07},\sqrt{0.15},\sqrt{0.05},\sqrt{0.1},\sqrt{0.3},\sqrt{0.2},\sqrt{0.1}\right)^T \in\mathbb{C}^8$. (b) The quantum circuit to prepare the 3-qubit state $|\psi_\nu \rangle$. The single-qubit gates used in this circuit are defined as $U_i = R_y(2\theta_i)$ for $i\in\{0,1,\ldots,6\}$. The value of $\theta_i$ is shown as follow : $\theta_0 = \arccos\left({\sqrt{0.3/1}}\right)$, $\theta_1 = \arccos\left({\sqrt{0.1/0.3}}\right)$, $\theta_2 = \arccos\left({\sqrt{0.4/0.7}}\right)$, $\theta_3 = \arccos\left({\sqrt{0.03/0.1}}\right)$, $\theta_4 = \arccos\left({\sqrt{0.15/0.2}}\right)$, $\theta_5 = \arccos\left({\sqrt{0.1/0.4}}\right)$, $\theta_6 = \arccos\left({\sqrt{0.2/0.3}}\right)$. 
}
\label{fig:framework}
\end{figure}


Based on the Binary Search Tree in Figure \ref{fig:framework}(a), the QSP circuit can be designed layer-by-layer and branch-by-branch, as in Figure \ref{fig:framework}(b).

\section{Implementations of tasks in Eq. \eqref{eq:task1} and Eq. \eqref{eq:alpha}}
\label{sec:2tasks}

 The first task (Eq. \eqref{eq:task1}) can be completed by the combination of the circuit in Figure \ref{fig:f_circuit}, a fact formalized as Lemma \ref{lem:f_circuit}, which can be easily verified.
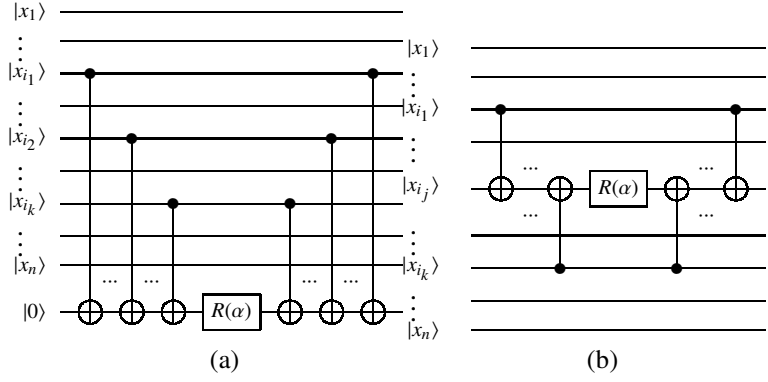
\begin{figure}[ht]
\centering
\subfloat[]{
\centering
\Qcircuit @C=.6em @R=1em {
\lstick{\scriptstyle \ket{x_1}}  & \qw & \qw & \qw & \qw & \qw & \qw & \qw& \qw\\
\lstick{\scriptstyle\vdots~~~}  & \qw  & \qw & \qw & \qw & \qw & \qw & \qw & \qw\\
\lstick{\scriptstyle\ket{x_{i_1}}}  & \ctrl{7} & \qw & \qw & \qw & \qw & \qw & \ctrl{7}& \qw\\
\lstick{\scriptstyle\vdots~~~}  & \qw & \qw & \qw & \qw & \qw & \qw & \qw& \qw\\
\lstick{\scriptstyle\ket{x_{i_2}}}  & \qw & \ctrl{5} & \qw & \qw & \qw & \ctrl{5} & \qw& \qw\\
\lstick{\scriptstyle\vdots~~~}  & \qw & \qw & \qw & \qw & \qw & \qw & \qw & \qw\\
\lstick{\scriptstyle\ket{x_{i_k}}}  & \qw & \qw & \ctrl{3} & \qw & \ctrl{3} & \qw & \qw& \qw\\
\lstick{\scriptstyle\vdots~~~}  & \qw & \qw & \qw & \qw & \qw & \qw & \qw & \qw\\
\lstick{\scriptstyle\ket{x_n}}  & \qw & \dstick{\scriptstyle\cdots~~~~~~}\qw & \dstick{\scriptstyle\cdots~~~~~~}\qw & \qw & \qw &\dstick{\scriptstyle\cdots~~~~~~} \qw & \dstick{\scriptstyle\cdots~~~~~~}\qw & \qw\\
\lstick{\scriptstyle\ket{0}}  &  \targ & \targ & \targ & \gate{ \scriptstyle R(\alpha)}& \targ &\targ & \targ & \qw\\
}
}
~~~~~
\subfloat[]{
\Qcircuit @C=.6em @R=1em {
\lstick{\scriptstyle\ket{x_1}} & & \qw & \qw & \qw & \qw & \qw & \qw & \qw& \qw\\
\lstick{\scriptstyle\vdots~~~} & & \qw  & \qw & \qw & \qw & \qw & \qw & \qw & \qw\\
\lstick{\scriptstyle\ket{x_{i_1}}} & & \ctrl{2} & \qw & \qw & \qw & \qw & \qw & \ctrl{2}& \qw\\
\lstick{\scriptstyle\vdots~~~} & & \qw & \qw & \qw & \qw & \qw & \qw & \qw& \qw\\
\lstick{\scriptstyle\ket{x_{i_j}}} & & \targ & \ustick{\scriptstyle\cdots}\qw & \targ &  \gate{\scriptstyle R(\alpha)} & \targ &\ustick{\scriptstyle\cdots} \qw & \targ& \qw\\
\lstick{\scriptstyle\vdots~~~} & & \qw & \ustick{\scriptstyle\cdots}\qw & \qw & \qw & \qw &\ustick{\scriptstyle\cdots} \qw & \qw & \qw\\
\lstick{\scriptstyle\ket{x_{i_k}}} & & \qw & \qw & \ctrl{-2} & \qw & \ctrl{-2} & \qw & \qw& \qw\\
\lstick{\scriptstyle\vdots~~~} & & \qw & \qw & \qw & \qw & \qw & \qw & \qw & \qw\\
\lstick{\scriptstyle\ket{x_n}} & & \qw & \qw & \qw & \qw & \qw & \qw & \qw & \qw\\
}
}
\caption{A quantum circuit to implement transformation $\ket{x_1x_2\cdots x_n} \to e^{i \langle s,x\rangle \alpha}\ket{x_1x_2\cdots x_n}$ with string $s=s_1s_2\cdots s_n\in\{0,1\}^n$ being the indicator vector of set $S = \{i_1, \ldots, i_k\} \subseteq [n]$, i.e.  $s_j=1$ if $j\in S$ and $s_j=0$ otherwise. (a) A quantum circuit with an ancillary qubit initialized as $\ket{0}$. The index set of controlled qubit of CNOT gates is $S$.  (b) A quantum circuit without ancillary qubits, where $i_{j}$ is an arbitrary element in $S$. The index set of the controlled qubit of CNOT gates is $S-\{i_j\}$ and the index of target qubit is $i_j$.}\label{fig:f_circuit}
\end{figure}
\begin{lemma}\label{lem:f_circuit}
Let $x = x_1 x_2 \ldots x_n$, $s = s_1 s_2 ... s_n \in\{0,1\}^n$, and $S = \{i_1, \ldots, i_k\} = \{i: s_i = 1\} \subseteq [n]$. The circuits in Figure \ref{fig:f_circuit} realize the following transformation:
\[\ket{x_1x_2\cdots x_n} \to e^{i\langle s,x\rangle \alpha}\ket{x_1 x_2 \cdots x_n}.\]
\end{lemma}

The second task (Eq. \eqref{eq:alpha}) is accomplished as follows. Based on Lemma \ref{lem:f_circuit}, one can implement transformation Eq. \eqref{eq:task1} by using $2^n-1$ circuits with parameters $\alpha_s$ for all $s\in\Bn-\{0^n\}$ in Figure \ref{fig:f_circuit}.
To determine parameters $\alpha_s$ in Eq. \eqref{eq:alpha}, the second task is essentially asking whether the $(2^n-1)\times (2^n-1)$ matrix $A$ defined by 
\begin{align}\label{eq:matrix-A}
    A(x,s) = \langle x,s\rangle, \quad x,s\in \Bn-\{0^n\}
\end{align}
is invertible. The answer is affirmative and the inverse is given by the following lemma, which can be easily verified \cite{welch2014efficient2}.
\begin{lemma}
    The matrix $A$ defined as in Eq. \eqref{eq:matrix-A} is invertible, and its inverse is $2^{1-n}(2A - \bf{J})$, where ${\bf J}\in\mathbb{R}^{(2^{n}-1)\times (2^{n}-1)}$ is the all-one matrix. 
    \label{lem:fwt}
\end{lemma}

{This gives a way to compute the parameters $\alpha_s$ efficiently on a classical computer.} 
\begin{lemma}
    For QSP problem, given a unit vector $v=(v_0,v_1,\cdots,v_{2^n-1})^T\in\mathbb{C}^{2^n}$, the values of $\{\alpha_s: s\in\{0,1\}^n-\{0^n\}\}$ in Eq. \eqref{eq:alpha} can be calculated on a classical computer using $O(n2^n)$ time and $O(n2^n)$ space.
\end{lemma}
\begin{proof}
We calculate these $\alpha_s$ in three steps. 
\begin{enumerate}
    \item Our QSP circuit consists of $n$ UCGs $V_1,V_2,\cdots,V_n$ as in Figure \ref{fig:QSP_circuit}(a). We calculate all parameters of $V_1,\ldots,V_n$ in time $O(2^n)$ and in space $O(2^n)$ using binary trees    \cite{grover2002creating,kerenidis2017quantum}. 
    \item Secondly, we decompose all UCGs into diagonal unitary matrices and some single-qubit operations according to Eq. \eqref{eq:UCG}. As in the proof of Lemma \ref{lem:lamda2circuit}, we decompose every single-qubit gate in $V_j$ into $R_z$ gates, $S$ gates and $H$ gates in time $O(1)$ and space $O(1)$, and UCG $V_j$ can be decomposed into 3 diagonal unitary matrices and two $H$ gates and $S$ gates in time $O(2^j)$ and space $O(2^j)$. Hence, the total time and space of this step are $\sum_{j=1}^nO(2^j)=O(2^n)$.
    \item Thirdly, for every diagonal unitary matrix $\Lambda_j$ with diagonal element $e^{i\theta(x)}$ for all $x\in\{0,1\}^j-\{0^j\}$, we calculate parameters $\alpha_s$ in Eq. \eqref{eq:alpha}. Let 
    \begin{align*}
        &\boldsymbol \alpha\defeq(\alpha_{0\cdots 01},\alpha_{0\cdots 10},\ldots,\alpha_{1\cdots 11})^T\in\mathbb{R}^{2^n-1},\\
        & \boldsymbol \theta\defeq(\theta(0\cdots 01),\theta(0\cdots 10),\ldots,\theta(1\cdots 11))^T\in\mathbb{R}^{2^n-1}.
    \end{align*}
Based on Eq. \eqref{eq:alpha} and lemma \ref{lem:fwt}, we have $\boldsymbol \alpha=2^{1-n}(2A - \bf{J})\boldsymbol \theta$. Notice that $(2A - \bf{J})\boldsymbol \theta$ is just the Walsh-Hadamard transform on $\boldsymbol \theta$. Thus the vector $\boldsymbol\alpha$ can be calculated by fast Walsh-Hadamard transform algorithm \cite{fino1976unified}, which costs $O(n2^n)$ time and $O(n2^n)$ space.
\end{enumerate} 
Adding these costs up, we see that the values of $\alpha_s$ can be calculated in $O(n2^n)$ time and $O(n2^n)$ space.
\end{proof}

\section{Warm-up example: Implement $\Lambda_4$ using 8 ancillary qubits}
\label{sec:warm-up}
In this section, we will show how to implement $\Lambda_4$ with $8$ ancillary qubits based on a Gray code, which can help to understand the general case.
Our construction of quantum circuit for $\Lambda_4$ is shown in Figure \ref{fig:example_Lambda_4}. 
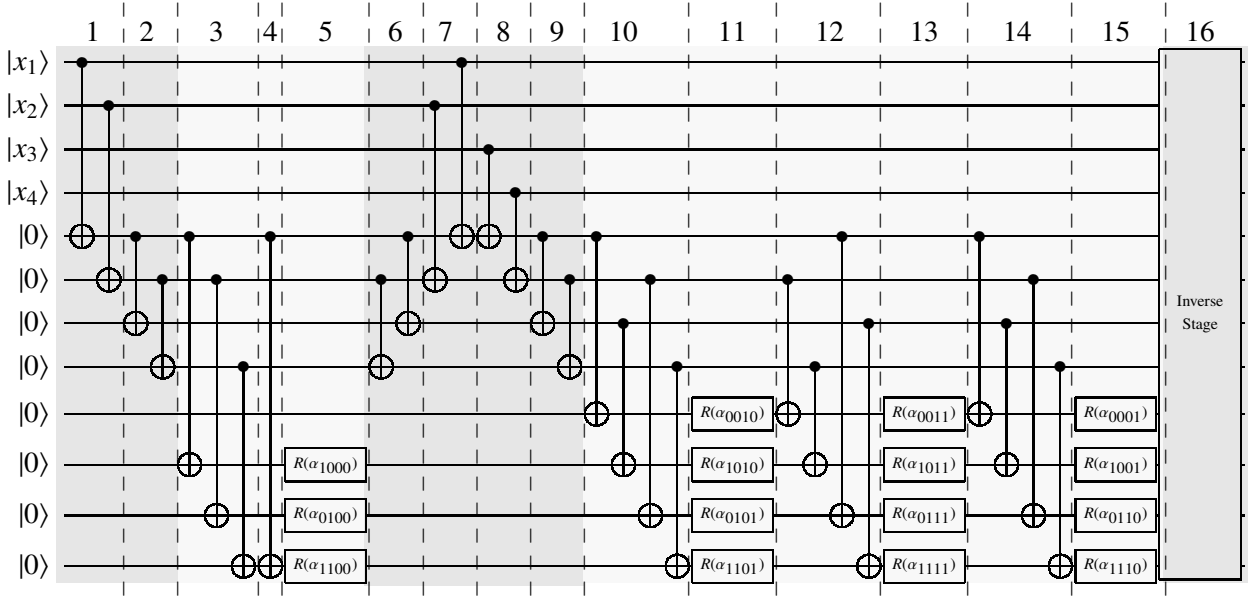
\begin{figure*}
\centerline{
\begin{tabular}{c}
\begin{pgfpicture}{0em}{0em}{0em}{0em}
\color{lgray}
\pgfrect[fill]{\pgfpoint{0.0em}{-0.5em}}{\pgfpoint{4.15em}{-18.5em}}
\color{llgray}
\pgfrect[fill]{\pgfpoint{4.15em}{-0.5em}}{\pgfpoint{8em}{-18.5em}}
\color{lgray}
\pgfrect[fill]{\pgfpoint{10.5em}{-0.5em}}{\pgfpoint{7.5em}{-18.6em}}
\color{llgray}
\pgfrect[fill]{\pgfpoint{18em}{-0.5em}}{\pgfpoint{22em}{-18.5em}}
\color{lgray}
\pgfrect[fill]{\pgfpoint{37.8em}{-0.5em}}{\pgfpoint{3em}{-18.5em}}
\end{pgfpicture}
\Qcircuit @C=0.1em @R=0.6em {
 &  & ~~~1 & \barrier[-0.4 em]{12}  & ~~~2  & \barrier[-0.4 em]{12}  &  &  3& \barrier[-0.4 em]{12}  \barrier[0.4 em]{12}& 4 & \barrier[-0.4 em]{12}  5 &  ~~~~6  & \barrier[-0.4 em]{12}  & ~~~7 & \barrier[-0.4 em]{12}  & ~~~~8 & \barrier[-0.4 em]{12}  & ~~~~9  & \barrier[-0.4 em]{12}  &  & 10  & \barrier[0.4 em]{12} &   & \barrier[-0.4em]{12} 11 &  & ~~~~12  & \barrier[0.4 em]{12}  &  & \barrier[-0.4 em]{12} 13 &   & ~~~14  & \barrier[0.4 em]{12}  &  \barrier[1.72 em]{12}  & 15  & 16 & \\
\lstick{\ket{x_1}} & \qw & \ctrl{4} & \qw &  \qw &  \qw &  \qw &  \qw &  \qw &  \qw &   \qw &  \qw &  \qw & \qw & \ctrl{4} & \qw & \qw & \qw & \qw & \qw & \qw & \qw & \qw & \qw &  \qw &  \qw &  \qw &  \qw &  \qw &  \qw &  \qw &  \qw &  \qw &  \qw & \multigate{11}{\scriptscriptstyle \text{Inverse}\atop \scriptscriptstyle\text{Stage}} & \qw\\
\lstick{\ket{x_2}} & \qw & \qw & \ctrl{4} & \qw &  \qw &  \qw &  \qw &  \qw &  \qw &  \qw &  \qw &  \qw & \ctrl{4} & \qw & \qw & \qw & \qw & \qw & \qw & \qw & \qw & \qw & \qw &  \qw &  \qw &  \qw &  \qw &  \qw &  \qw &  \qw &  \qw &  \qw &  \qw & \ghost{\scriptscriptstyle\text{Inverse}\atop \scriptscriptstyle\text{Stage}} & \qw\\
\lstick{\ket{x_3}} & \qw & \qw &  \qw &  \qw &  \qw &  \qw &  \qw &  \qw &  \qw &  \qw &  \qw &  \qw & \qw & \qw & \ctrl{2} & \qw & \qw & \qw & \qw & \qw & \qw & \qw & \qw &  \qw &  \qw &  \qw &  \qw &  \qw &  \qw &  \qw &  \qw &  \qw &  \qw & \ghost{ \scriptscriptstyle \text{Inverse}\atop \scriptscriptstyle\text{Stage}} & \qw\\
\lstick{\ket{x_4}} & \qw & \qw &  \qw &  \qw &  \qw &  \qw &  \qw &  \qw &  \qw &  \qw &  \qw &  \qw & \qw & \qw & \qw & \ctrl{2} & \qw & \qw & \qw & \qw & \qw & \qw & \qw &  \qw &  \qw &  \qw &  \qw &  \qw &  \qw &  \qw &  \qw &  \qw &  \qw & \ghost{\scriptscriptstyle\text{Inverse}\atop \scriptscriptstyle\text{Stage}} & \qw\\
\lstick{\ket{0}} & \qw & \targ & \qw & \ctrl{2} &  \qw & \ctrl{5} &  \qw &  \qw & \ctrl{7}&  \qw &  \qw & \ctrl{2} & \qw & \targ & \targ & \qw &\ctrl{2} & \qw & \ctrl{4} & \qw & \qw & \qw & \qw &  \qw &  \qw &  \ctrl{6} &  \qw &  \qw &  \ctrl{4} &  \qw &  \qw &  \qw &  \qw & \ghost{\scriptscriptstyle\text{Inverse}\atop \scriptscriptstyle\text{Stage}} & \qw\\
\lstick{\ket{0}} & \qw & \qw & \targ &  \qw & \ctrl{2} &  \qw & \ctrl{5} &  \qw &  \qw &  \qw & \ctrl{2} & \qw & \targ & \qw & \qw & \targ & \qw & \ctrl{2} & \qw & \qw & \ctrl{5} & \qw & \qw &  \ctrl{3} &  \qw &  \qw &  \qw &  \qw &  \qw &  \qw &  \ctrl{5} &  \qw &  \qw & \ghost{\scriptscriptstyle\text{Inverse}\atop \scriptscriptstyle\text{Stage}} & \qw\\
\lstick{\ket{0}} & \qw & \qw &  \qw & \targ &  \qw &  \qw &  \qw &  \qw &  \qw &  \qw &  \qw & \targ & \qw & \qw & \qw & \qw & \targ & \qw & \qw & \ctrl{3} & \qw & \qw & \qw &  \qw &  \qw &  \qw &  \ctrl{5} &  \qw &  \qw &  \ctrl{3} &  \qw &  \qw &  \qw & \ghost{\scriptscriptstyle\text{Inverse}\atop \scriptscriptstyle\text{Stage}} & \qw\\
\lstick{\ket{0}} & \qw & \qw &  \qw &  \qw &  \targ & \qw &  \qw & \ctrl{4} &  \qw &  \qw & \targ &  \qw & \qw & \qw & \qw & \qw & \qw & \targ & \qw & \qw & \qw & \ctrl{4} & \qw &  \qw &  \ctrl{2} &  \qw &  \qw &  \qw &  \qw &  \qw &  \qw &  \ctrl{4} &  \qw & \ghost{\scriptscriptstyle\text{Inverse}\atop \scriptscriptstyle\text{Stage}} & \qw\\
\lstick{\ket{0}} & \qw & \qw &  \qw &  \qw &  \qw &  \qw &  \qw &  \qw &  \qw &   \qw &  \qw & \qw & \qw & \qw & \qw & \qw & \qw & \qw & \targ & \qw & \qw & \qw & \gate{\scriptscriptstyle R(\alpha_{0010})} &  \targ &  \qw &  \qw &  \qw &  \gate{\scriptscriptstyle R(\alpha_{0011})} &  \targ &  \qw &  \qw &  \qw & \gate{\scriptscriptstyle R(\alpha_{0001})} & \ghost{\scriptscriptstyle\text{Inverse}\atop \scriptscriptstyle\text{Stage}} & \qw\\
\lstick{\ket{0}} & \qw & \qw &  \qw &  \qw &  \qw & \targ &  \qw &  \qw &  \qw & \gate{{\scriptscriptstyle R(\alpha_{1000})}}& \qw & \qw & \qw & \qw & \qw & \qw & \qw & \qw & \qw & \targ & \qw & \qw & \gate{\scriptscriptstyle R(\alpha_{1010})} &  \qw &  \targ &  \qw &  \qw & \gate{\scriptscriptstyle R(\alpha_{1011})} &  \qw &  \targ &  \qw &  \qw & \gate{\scriptscriptstyle R(\alpha_{1001})} & \ghost{\scriptscriptstyle\text{Inverse}\atop \scriptscriptstyle\text{Stage}} & \qw\\
\lstick{\ket{0}} & \qw & \qw &  \qw &  \qw &  \qw &  \qw & \targ &  \qw &   \qw & \gate{\scriptscriptstyle R(\alpha_{0100})} &  \qw & \qw & \qw & \qw & \qw & \qw & \qw & \qw & \qw & \qw &\targ & \qw & \gate{\scriptscriptstyle R(\alpha_{0101})} &  \qw &  \qw &  \targ & \qw & \gate{\scriptscriptstyle R(\alpha_{0111})} &  \qw &  \qw &  \targ &  \qw & \gate{\scriptscriptstyle R(\alpha_{0110})} & \ghost{\scriptscriptstyle\text{Inverse}\atop \scriptscriptstyle\text{Stage}} &\qw\\
\lstick{\ket{0}} & \qw & \qw &  \qw &  \qw &   \qw &  \qw &  \qw &  \targ & \targ &  \gate{\scriptscriptstyle R(\alpha_{1100})} &  \qw & \qw & \qw & \qw & \qw & \qw & \qw & \qw & \qw & \qw & \qw & \targ & \gate{\scriptscriptstyle R(\alpha_{1101})} &  \qw &  \qw &  \qw &  \targ & \gate{\scriptscriptstyle R(\alpha_{1111})} & \qw &  \qw &  \qw &  \targ & \gate{\scriptscriptstyle R(\alpha_{1110})} & \ghost{\scriptscriptstyle\text{Inverse}\atop \scriptscriptstyle\text{Stage}} & \qw\\
}
\end{tabular}
}
    \caption{Implementation of $\Lambda_4$ with 8 ancillary qubits. The first 4 qubits form the input register, the next 4 qubits form the copy register and the last 4 qubits form the phase register. Step 1-2 are Prefix Copy Stage; Step 3-5 are Gray Initial Stage; Step 6-9 are Suffix Copy Stage; Step 10-15 are Gray Path Stage; Step 16 is Inverse stage. All parameters $\alpha_s$ of phase gate $R(\alpha_s)$ for $s\in\{0,1\}^4$ are determined by Eq. \eqref{eq:alpha}. Step 16 consists of inverse quantum circuits of Step 14,12,10,9,8,7,6,4,3,2 and 1 in that order. Step 1-15 are quantum circuits of depth 1 and Step 16 is quantum circuits of depth 11.}
    \label{fig:example_Lambda_4}
\end{figure*}
All the parameters $\alpha_s$ for $s\in\{0,1\}^4$ are defined in Eq. \eqref{eq:alpha}. In Figure \ref{fig:example_Lambda_4}, the first four qubits initialized as $\ket{x}:=\ket{x_1x_2 x_3 x_4}$ constitute the input register; the next four qubits form the copy register; and the last four qubits form the phase register. All qubits in the copy and the phase register are initialized to $|0\rangle$.

For any 4-bit string $s=s_1s_2s_3s_4\in\{0,1\}^4$, $s_1s_2$ is the prefix and $s_3s_4$ is the suffix.
In Step 1-2, we make two copies of prefix $\ket{x_1 },\ket{x_2}$ into the copy register, i.e.,
\begin{align*}
    \underbrace{\ket{x_1x_2x_3x_4}}_{\text{input register}}\underbrace{\ket{0000}}_{\text{copy register}}\to\ket{x_1x_2x_3x_4}\ket{x_1x_2x_1x_2}. &\text{(Step 1-2)}
\end{align*}
In Step 3-5, we generate all $4$-bit strings whose suffixes are all $00$ in the phase register by CNOT gates.
Namely, we realize the following transformation
\begin{align*}
    &\overbrace{\ket{x_1x_2x_3x_4}}^{\text{input register}}\overbrace{\ket{x_1x_2x_1x_2}}^{\text{copy register}}\overbrace{\ket{0000}}^{\text{phase register}}\\
    \to &\ket{x_1x_2x_3x_4}\ket{x_1x_2x_1x_2}\\
    &\ket{\langle0000,x\rangle,\langle1000,x\rangle,\langle0100,x\rangle,\langle 1100,x\rangle} &  \text{(Step 3-4)}\\
    \to& e^{i\sum_{s\in\{0,1\}^2}\langle s00,x\rangle \alpha_{s00}}\ket{x_1x_2x_3x_4}\ket{x_1x_2x_1x_2}\\
    &\ket{\langle0000,x\rangle,\langle1000,x\rangle,\langle0100,x\rangle,\langle 1100,x\rangle}. &  \text{(Step 5)}
\end{align*}

In Step 6-9, we transform the copy register to initial state and make two copies of suffix $x_3x_4$:
\begin{align*}
&\overbrace{\ket{x_1x_2x_3x_4}}^{\text{input register}}\overbrace{\ket{x_1x_2x_1x_2}}^{\text{copy register}}\\
 \to&\ket{x_1x_2x_3x_4}\ket{0000} & \text{(Step 6-7)}\\
 \to& \ket{x_1x_2x_3x_4}\ket{x_3x_4x_3x_4}. & \text{(Step 8-9)}
\end{align*}
Up to now, we have generated all prefixes for suffix ``00'', and next we  need to generate all the other $4$-bit strings. 
In order to reduce the number of CNOT gates, we consider two forms of $2$-bit Gray code. The $1$-Gray code and $2$-Gray code starting from $00$ are
\[00, 10, 11, 01 \text{~~and~~} 00, 01, 11, 10.\] 
If in $1$-Gray code and $2$-Gray code, we list all the bits changed between adjacent strings, we get two lists: $1,2,1$ and $2,1,2$.
To obtain parallelism, in the first and second qubit in the phase register we generate suffixes $00,10,11,01$ ($1$-Gray code) in that order.
In the third and fourth qubit in the phase register we generate suffixes $00,01,11,10$ ($2$-Gray code) in that order.
In Step 10-15, we implement the following transformation
\begin{align*}
& e^{i\sum_{s\in\{0,1\}^2} \langle s00,x\rangle\alpha_{s00}}\ket{x_1x_2x_3x_4}\ket{x_3x_4x_3x_4}\\
&\ket{\langle0000,x\rangle,\langle1000,x\rangle,\langle0100,x\rangle,\langle1100,x\rangle} \\    
 \to& e^{i\sum_{s\in\{0,1\}^4-\{0^4\}}\langle s,x\rangle\alpha_s}\ket{x_1x_2x_3x_4}\ket{x_3x_4x_3x_4}\\
 &\ket{\langle0001,x\rangle,\langle1001,x\rangle,\langle0110,x\rangle,\langle1110,x\rangle} &\text{(Step 10-15)}\\  
 =&e^{i\theta(x)}\ket{x_1x_2x_3x_4}\ket{x_3x_4x_3x_4}\\
 &\ket{\langle0001,x\rangle,\langle1001,x\rangle,\langle0110,x\rangle,\langle1110,x\rangle}. 
\end{align*}
Every step can be realized by a quantum circuit in depth 1.
Step 16 consists of inverse quantum circuits of Step 14,12,10,9,8,7,6,4,3,2 and 1 in order. The total depth of Step 16 is 11. It transforms copy register and phase register to their initial states. Therefore, it implements the following transformation
\begin{align*}
   & e^{i\theta(x)}\ket{x_1x_2x_3x_4}\ket{x_3x_4x_3x_4}\\
   &\ket{\langle0001,x\rangle,\langle1001,x\rangle,\langle0110,x\rangle,\langle1110,x\rangle} \\
    \to & e^{i\theta(x)}\ket{x}\ket{0000}\ket{0000} &\text{(Step 16)} 
\end{align*}
As discussed above, the quantum circuit in Figure \ref{fig:example_Lambda_4} is an implementation of $\Lambda_4$ with 8 ancillary qubits.

\section{Proof of Lemma \ref{lem:2D-array}}
\label{sec:2D-array}
{\noindent\bf Lemma \ref{lem:2D-array}}
\emph{
    Let $t = {\lfloor \log \frac{m}{2} \rfloor}$ and $\ell = 2^t$. The set $\Bn$ can be partitioned into a 2-dimensional array $\{s(j,k): j\in [\ell], k\in [2^n/\ell]\}$ of $n$-bit strings, satisfying that 
    \begin{enumerate}
        \item Strings in the first column $\{s(j,1): j\in [\ell]\}$ have the last $(n-t)$ bits being all 0, and strings in each row $\{s(j,k): k\in [2^n/\ell]\}$ share the same first $t$ bits.
        \item $\forall j\in [\ell], \forall k\in [2^n/\ell-1]$, $s(j,k)$ and $s(j,k+1)$ differ by 1 bit.
        \item For any fixed $k\in [2^n/\ell-1]$, and any $t' \in \{t+1,...,n\}$, there are at most $\big(\frac{m}{2(n-t)} + 1\big)$ many $j\in [\ell]$ s.t. $s(j,k)$ and $s(j,k+1)$ differ by the $t'$-th bit.
    \end{enumerate}
}
\begin{proof}
Consider each $n$-bit string as two parts, a $t$-bit prefix followed by an $(n-t)$-bit suffix. We let $\{s(j,1): j\in [\ell]\}$ run over all $\ell$ possible prefixes, and for each fixed $j\in [\ell]$, the collection of $\{s(j,k): k\in [2^n/\ell]\}$ run over all possible suffixes. Thus $\{s(j,k): j\in [\ell], k\in [2^n/\ell]\}$ form a partition of $\Bn$, and the first condition is satisfied. 

Now for the $j$-th set of suffixes $\{s(j,k):k\in [2^n/\ell]\}$, we identify it with $\B^{n-t}$, and apply the $(j',n-t)$-Gray code and Lemma \ref{lem:GrayCode} to it, where $j' = ((j-1) \ mod\ (n-t)) + 1\in \{1,...,n-t\}$. 
For any $k\in [2^n/\ell-1]$ and any $t'\in \{t+1,...,n\}$, let us see how many $j\in [\ell]$ have that $s(j,k)$ and $s(j,k+1)$ differ by bit $t'$. When $j$ runs over $[n-t]$, $s(j,k)$ and $s(j,k+1)$ differ by bit $t'$ exactly once. When $j$ runs over $\{n-t+1,...,2(n-t)$, $s(j,k)$ and $s(j,k+1)$ differ by bit $t'$ again exactly once. Repeating this we can see that  when $j$ runs over all $[\ell]$, $s(j,k)$ and $s(j,k+1)$ differ by bit $t'$ for at most $\lceil \ell/(n-t) \rceil \le m/2(n-t)+1$ times. 
\end{proof}

\section{Proof of Lemma \ref{lem:inverse}}
\label{sec:inverse}
{\noindent\bf Lemma \ref{lem:inverse}}
    \emph{The depth and size of the Inverse Stage are at most $O(\log m + 2^n /m)$ and $\frac{m}{2}+\frac{nm}{2}+m+2^{n}=2^{n}+\frac{3m+nm}{2}$. The effect of this stage is 
    \begin{equation*}
        \ket{x} \ket{x_{suf}} \ket{f_{[\ell],2^n/\ell}} \xrightarrow{U_{Inverse}}  \ket{x} \ket{0^{m/2}} \ket{0^{m/2}}.
    \end{equation*}}

\begin{proof}
    The depth is just the summation of the CNOT depth of the first four stages, which is $O\big(\log m + 2\log m + 2\log m + 2^n/\ell\big) = O\big(\log m + 2^n /m\big).$
    { The analyze of size is similar by adding up the sizes in previous stages.}
    The effect is shown as follows, which holds by Lemma \ref{lem:GrayPath}, Eq. \eqref{eq:suf-copy}, Eq. \eqref{eq:U1} and Eq. \eqref{eq:Uc1-effect}. 
    \begin{align*}
    &\ket{x} \ket{x_{suf}} \ket{f_{[\ell],2^n/\ell}} 
    \xrightarrow{U_{2}^{\dagger} \cdots U_{2^n/\ell}^\dagger}  \ket{x} \ket{x_{suf}} \ket{f_{[\ell],1}} 
     \xrightarrow{U_{copy,1} U_{copy,2}^\dagger} \\ 
     &\ket{x} \ket{x_{pre}} \ket{f_{[\ell],1}} 
     \xrightarrow{U_1^\dagger}  \ket{x} \ket{x_{pre}} \ket{0^{\frac{m}{2}}}
    \xrightarrow{U_{copy,1}^\dagger}  \ket{x} \ket{0^{\frac{m}{2}}} \ket{0^{\frac{m}{2}}}
    \end{align*}
    \end{proof}

\section{Proof of Lemma \ref{lem:puttingtogether_ancilla}}
\label{sec:puttingtogether_ancilla}
\noindent
{\bf Lemma \ref{lem:puttingtogether_ancilla}}
\emph{The circuit implements the operation in  Eq. \eqref{matrix:lambda_n} in depth $O(\log m + 2^n /m)$ and in size $3\cdot 2^{n}+nm+{\frac{7}{2}}m$.}
\begin{proof}
    For the depth, simply adding up the depth  and the size of the five stages gives the bound. Next we analyze the operation step by step as follows. The three registers are the input, copy and phase registers, respectively. 
    \begin{align*}
        &\ket{x}\ket{0^{m/2}}\ket{0^{m/2}} \\
        & \xrightarrow{U_{copy,1}} \ket{x}\ket{x_{pre}}\ket{0^{m/2}} & \text{(Eq. \eqref{eq:Uc1-effect})} \\
        & \xrightarrow{U_{GrayInit}}  e^{i \sum\limits_{j\in [\ell]} f_{j,1}(x) \alpha_{s(j,1)} } \ket{x} \ket{x_{pre}} \ket{f_{[\ell],1}} & \text{(Eq. \eqref{eq:UGI-effect})} \\
        & \xrightarrow{U_{copy,2}U_{copy,1}^\dagger} e^{i \sum\limits_{j\in [\ell]} f_{j,1}(x) \alpha_{s(j,1)} } \ket{x} \ket{x_{suf}} \ket{f_{[\ell],1}} & \text{(Eq. \eqref{eq:suf-copy})} \\
        & \xrightarrow{R_2 U_2} e^{i \sum\limits_{j\in [\ell]\atop k\in[2]} f_{j,k}(x) \alpha_{s(j,k)}} \ket{x} \ket{x_{suf}} \ket{f_{[\ell],2}} & \text{(Eq. \eqref{eq:GrayPath})} \\
        & \quad \vdots & \\
        & \xrightarrow{R_{\frac{2^n}{\ell}} U_{\frac{2^n}{\ell}}} e^{i \sum\limits_{j\in [\ell]\atop k\in [\frac{2^n}{\ell}]} f_{j,k}(x) \alpha_{s(j,k)} } \ket{x} \ket{x_{suf}} \ket{f_{[\ell],\frac{2^n}{\ell}}} & \text{(Eq. \eqref{eq:GrayPath})} \\
        & \qquad = e^{i\sum\limits_{s\in \{0,1\}^n} \langle x,s\rangle \alpha_{s}} \ket{x} \ket{x_{suf}} \ket{f_{[\ell],2^n/\ell}} & (\text{Lem \ref{lem:2D-array}} ) \\
        & \qquad = e^{i\theta(x)} \ket{x} \ket{x_{suf}} \ket{f_{[\ell],2^n/\ell}} & \text{(Eq. \eqref{eq:alpha})} \\
        & \xrightarrow{U_{Inverse}} e^{i\theta(x)} \ket{x}\ket{0^{m/2}} \ket{0^{m/2}} & \text{(Eq. \eqref{eq:inverse})} \\
        & \qquad = \Lambda_n \ket{x}\ket{0^{m/2}}\ket{0^{m/2}} 
    \end{align*}
\end{proof}
\section{Construction of linearly independent sets}
\label{sec:partition}
What remains for completing the Generate Stage is the construction of sets $T^{(1)},T^{(2)},\ldots,T^{(\ell)}$, which we will show next.
\begin{lemma}\label{lem:partition}
There exist sets  $T^{(1)},T^{(2)},\cdots,T^{(\ell)} \subseteq \{0,1\}^n-\{0^n\}$, for some integer $\ell \le \frac{2^{n+2}}{n+1}-1$, such that:
    \begin{enumerate}
        \item For any $i\in[\ell]$, $|T^{(i)}|=n$;
        \item For any $i\in[\ell]$, the Boolean vectors in $T^{(i)}=\{{t^{(i)}_1},{t^{(i)}_2},\cdots,{t^{(i)}_n}\}$ are linearly independent over $\mathbb{F}_2$;
        \item $\bigcup_{i\in[\ell]} T^{(i)}= \{0,1\}^{n} - \{0^n\} $.
    \end{enumerate}
\end{lemma}
\begin{proof}
For any $n$-bit vector $x\in \{0,1\}^n$, let $S_x=\{x\oplus e_1,x\oplus e_2,\cdots,x\oplus e_{n}\}$.
%
Firstly, we construct a set $L\subseteq \{0,1\}^n$
which satisfies $|L|\le \frac{2^n+1}{n+1}$ and $\{0,1\}^n=(\bigcup_{x\in L} S_x)\cup L$. Let $k = \lceil\log{(n+1)}\rceil$. 
For $t\in[n]$, denote the $k$-bit binary representation of integer $t$ by $t_{k}\cdots t_2 t_1$, where $t_1,\ldots,t_{k}\in\{0,1\}$ and $t=\sum_{i=1}^{k}t_i2^{i-1}$. 
We use a bar to denote the corresponding column vector, i.e. 
\[\overline{t}=[t_1, t_2, \ldots, t_{k}]^{T}\in\{0,1\}^{k}.\]
Define a $k\times n$ Boolean matrix $H$ by concatenating vectors $\overline{1},\overline{2},\cdots,\overline{n}$ together, i.e. 
\[
H=[\overline{1},\overline{2},\cdots,\overline{n}] \in\{0,1\}^{k\times n}.
\]
Note that the $k$-dimensional identity matrix  $I_k = [\overline{2^{0}},\overline{2^1},\ldots,\overline{2^{k-1}}]$ is a submatrix of $H$, therefore $H$ is full row rank, i.e. $\text{rank}(H) = k$.
Define sets 
\begin{align} \label{equ:all-b}
   & L^{(0)} = \{x\in \Bn: Hx = 0^k\},\nonumber\\
    &L^{(1)} = \{x\in \Bn: Hx = 1^k\},
\end{align}
and 
\begin{align} \label{equ:lastbit-b}
    &A^{(0)} = \{x\in \Bn: (Hx)_k = 0\},\nonumber\\
    &A^{(1)} = \{x\in \Bn: (Hx)_k = 1\}.
\end{align}
For each $x\in \Bn$, the last bit of $Hx$ is either 0 or 1, thus $A^{(0)}\cup A^{(1)} = \Bn$. Also note that for each $b\in \B$, $L^{(b)}$ requires all bits being $b$ and $A^{(b)}$ only requires the last bit being $b$, thus $L^{(b)} \subseteq A^{(b)}$.

%
Now we will show 
\[A^{(0)}\subseteq L^{(0)}\cup \left(\cup_{x\in L^{(0)}}S_x\right) \quad \text{and} \quad  A^{(1)}\subseteq L^{(1)}\cup \left(\cup_{x\in L^{(1)}}S_x\right).\]
For any $y\in A^{(0)}-L^{(0)}$, consider $\overline{t}=Hy$: Since it satisfies $\overline{t}_{k}=0$ and $t_{k-1}\cdots t_1 \ne 0^{k-1}$, we have that 
\[1\le t \le \sum_{i=1}^{k-1} 2^{i-1} < 2^{k-1} = 2^{\lceil \log(n+1)\rceil -1} <  2^{\log(n+1)} = n+1.\]
Therefore, $1\le t\le n$, and we can thus use $He_{t}=\overline{t}$ to obtain the following equality
\[
H(y\oplus e_{t})=Hy \oplus He_{t}=\overline{t}\oplus \overline{t}=0^{k}.
\]
Therefore, $y\oplus e_{t}\in L^{(0)}$. That is, for any $y\in A^{(0)}-L^{(0)}$, there exists an $x\in L^{(0)}$ s.t. $y=x\oplus e_{t}$ for some $t\in [n]$.
Hence, 
\[A^{(0)}\subseteq L^{(0)}\cup\left(\cup_{x\in L^{(0)}, \ t\in [n]}\  \{x\oplus e_{t}\}\right) = L^{(0)}\cup\left(\cup_{x\in L^{(0)}}S_x\right).\] 

For any $y\in A^{(1)}-L^{(1)}$, $\overline{t}=Hy$ satisfies $\overline{t}_{k}=1$ and $t_{k-1}...t_1 \ne 1^{k-1}$. It looks symmetric to the $A^{(0)}-L^{(0)}$ case but there is a technicality that the corresponding integer $t$ may be outside the range $[n]$. To remedy this, define $\overline{{t}'}=\overline{t}\oplus 1^{k}$ (and let ${t}'$ be the integer corresponding to vector $\overline{{t}'}$). Now that  $\overline{t'}_{k}=0$ and $t'_{k-1}...t'_1 \ne 0^{k-1}$, and we know $t'\in [n]$. Thus we can again get  $He_{{t}'}=\overline{{t}'}$ and, in turn,  
\[
H(y\oplus e_{{t}'}) = Hy \oplus He_{{t}'} = \overline{t} \oplus \overline{t'} = \overline{t}\oplus \overline{t}\oplus 1^{k}=1^{k}.
\] 
Therefore, $y\oplus e_{{t}'}\in L^{(1)}$, and 
\[A^{(1)}\subseteq L^{(1)}\cup\left(\cup_{x\in L^{(1)}, \ t'\in [n]}\  \{x\oplus e_{t'}\}\right) = L^{(1)}\cup\left(\cup_{x\in L^{(1)}}S_x\right).\] 
Let $L=L^{(0)}\cup L^{(1)}$. We have 
\begin{align*}
    &\{0,1\}^n=A^{(0)}\cup A^{(1)}
    \subseteq L^{(0)}\cup \left(\cup_{x\in L^{(0)}}S_x\right)\cup L^{(1)}\cup \left(\cup_{x\in L^{(1)}}S_x\right)
    = L \cup \left(\cup_{x\in L}S_x\right)\subseteq \{0,1\}^n.
\end{align*}
Recall that $\text{rank}(H) = k$ over field $\mathbb{F}_2$, the size of solution set $L^{(b)}$ is $|L^{(b)}| = 2^{n-k} = 2^{n-\lceil\log(n+1)\rceil} \le \frac{2^n}{n+1}$, for each $b\in \B$. Thus $|L|= |L^{(0)}|+|L^{(1)}|\le \frac{ 2^{n+1}}{n+1}$. 

We have constructed a set $L\subseteq\{0,1\}^n$ of size at most $\frac{ 2^{n+1}}{n+1}$ satisfying $L\cup (\cup_{x\in L}S_x) = \Bn$. We will now use this set $L$ to construct $\ell \le \frac{2^{n+2}}{n+1}-1$ sets $T^{(i)}$ which satisfy the three properties in the statement of the present lemma.

Since $0^n$ is a solution of $Hx = 0^k$, it holds that $0^n\in L^{(0)}\subseteq L$. Note that the vectors in $S_{0^n}=\{e_1,e_2,\ldots,e_n\}$ are linearly independent. For any $x\in L$ and $x\neq 0^n$, let us construct two sets of linearly independent vectors $S_x^{(0)}$ and $S_x^{(1)}$. 
Since $\text{rank}[x\oplus e_1,x\oplus e_2,\cdots,x\oplus e_{n}]\ge n-1$ over field $\mathbb{F}_2$, we can select $n-1$ linearly independent vectors from $S_x$ to form a set $S_x^{(0)}\subseteq S_x$. Let $S^{(1)}_x=(S_x-S^{(0)}_x-\{0^n\})\cup \{x\}$. It is not hard to verify that if $x = e_j$ for some $j\in [n]$, then $S^{(1)}_x = \{x\} = \{e_j\}$; if $x \notin \{0^n, e_1, \ldots, e_n\}$, then  $S^{(1)}_x = \{x, x\oplus e_j\}$ for some $j\in[n]$. In any case, the vector(s) in $S^{(1)}_x$ are linearly independent (the same for $S^{(0)}_x$), and it holds that $S_x^{(0)}\cup S_x^{(1)} = S_x\cup\{x\}-\{0^n\}$. Thus for each $b\in \B$, we can always extend the set $S^{(b)}_x$ to $T^{(b)}_x$ of $n$ linearly independent vectors by adding some  vectors. 
Recalling $\{0,1\}^n=(\cup_{x\in L}S_x)\cup L$, we have
\begin{align*}
    &\{0,1\}^n-\{0^n\} \\
    & =\left(\cup_{x\in L} S_x\right)\cup L-\{0^n\} \\
     &=\cup_{x\in L}(S_x\cup\{x\})-\{0^n\}\\
     &=\left(\cup_{x\in L-\{0^n\}}(S_x\cup\{x\}-\{0^n\})\right)\cup S_{0^n}\\
     &=\left(\cup_{x\in L-\{0^n\}}(S_x^{(0)}\cup S_x^{(1)})\right)\cup S_{0^n}\\
     & =\left(\cup_{x\in L-\{0^n\}} S_x^{(0)}\right)\cup \left(\cup_{x\in L-\{0^n\}} S_x^{(1)}\right)\cup S_{0^n}\\
     & \subseteq \left(\cup_{x\in L-\{0^n\}} T_x^{(0)}\right)\cup \left(\cup_{x\in L-\{0^n\}} T_x^{(1)}\right)\cup S_{0^n}\\
     &\subseteq\{0,1\}^n-\{0^n\}.    
\end{align*}
Now collect $\{T_x^{(0)}: x\in L-\{0^n\}\}$, $\{T_x^{(1)}: x\in L-\{0^n\}\}$ and $S_{0^n}$ as our sets $T^{(1)}, \ldots, T^{(\ell)}$. Since $|L|\le \frac{2^{n+1}}{n+1}$ and $0^n\in L$, the collection contains $\ell \le 2\cdot (\frac{2^{n+1}}{n+1} - 1) + 1 = \frac{2^{n+2}}{n+1} - 1$ sets, each consisting of $n$ linearly independent vectors, and the collection of all these vectors is exactly $\{0,1\}^n-\{0^n\}$. This completes the proof.
\end{proof}

\section{Correctness of circuit framework in Figure \ref{fig:framework_DU_withoutancilla}}
\label{sec:correctness_withoutancilla}
In this section, we shown the correctness of circuit framework in Figure \ref{fig:framework_DU_withoutancilla}. For any input state $|x\rangle$, the quantum circuit $(\Lambda_{r_c}\otimes\mathcal{R})\mathcal{G}_{\ell}\mathcal{G}_{\ell-1}\cdots\mathcal{G}_{1}$ makes the following sequence of operations.
    \begin{align*}
     \ket{x}&  = \ket{x_{control}} \ket{y^{(0)}} \\
      &  \xrightarrow{\mathcal{G}_1} e^{i\sum_{s\in F_1}\langle s,x\rangle\alpha_s}\ket{x_{control}} \ket{y^{(1)}}\\
  &  \xrightarrow{\mathcal{G}_2} e^{i\sum_{s\in F_1\cup F_2 }\langle s,x\rangle\alpha_s}\ket{x_{control}} \ket{y^{(2)}} \\
    & ~~~~\vdots \\
  &  \xrightarrow{\mathcal{G}_\ell} e^{i\sum_{s\in\bigcup_{k\in[\ell]}F_k }\langle s,x\rangle\alpha_s}\ket{x_{control}} \ket{y^{(\ell)}}\\
   &  \xrightarrow{\mathbb{I}_{r_c}\otimes \mathcal{R}} e^{i\sum_{s\in\bigcup_{k\in[\ell]}F_k }\langle s,x\rangle\alpha_s}\ket{x_{control}} \ket{y^{(0)}}\\
 &  \xrightarrow{\Lambda_{r_c}\otimes \mathbb{I}_{r_t}} e^{i\sum_{s\in\left(\bigcup_{k\in[\ell]}F_k\right)\cup \left(\left\{c0^{r_t}\right\}_{c\in\{0,1\}^{r_c}-\{0^{r_c}\}}\right) }\langle s,x\rangle\alpha_s}\\
 &~~~~~~~~~~~~\ket{x_{control}} \ket{y^{(0)}}\\
    & = e^{i\sum_{s\in\{0,1\}^n-\{0^{n}\} }\langle s,x\rangle\alpha_s}\ket{x_{control}} \ket{y^{(0)}}\\
    & = e^{i\theta(x)}\ket{x_{control}}\ket{y^{(0)}} \\
    &= e^{i\theta(x)}\ket{x}
    \end{align*}
The first equation holds by Eq. \eqref{eq:yk}. For arbitrary $k\in[\ell]$, unitary transformation $\mathcal{G}_k$ holds by Eq.\eqref{eq:Gk} and $F_{j}\cap F_k = \emptyset,$ for $j\in[k-1]$. Unitary transformation $\mathcal{R}$ holds by Eq. \eqref{eq:reset}. Unitary transformation $\Lambda_{r_c}$ holds by Eq. \eqref{eq:Lambda_rc}. The last two equations hold by Eq. \eqref{eq:set_eq} and  Eq. \eqref{eq:alpha}), respectively.

\section{Proof of Lemma \ref{lem:graypath_withoutancilla}}
\label{sec:graypath_withoutancilla}
\noindent
{\bf Lemma \ref{lem:graypath_withoutancilla}}
\emph{The quantum circuit defined above is of depth $O(2^{r_c})$ and of size $r_c2^{r_c+1}$, and implements Gray Path Stage $U_{GrayPath}$ in Eq. \eqref{eq:Graycodepath}.
}
\begin{proof}
We first show the correctness. For each $p\in[2^{r_c}]$, let us define a set $F_k^{(p)}$ by
\begin{equation}\label{eq:Fkp}
 F_k^{(p)}=\left\{s:\ s\in F_k\text{~and~}s=c_{p}^it_i^{(k)}\text{ for some }i\in[r_t]\right\} .
\end{equation}
By definition of $F_k$ in Eq. \eqref{eq:F_k}, the collection of $F_k^{(p)}$'s satisfy
\begin{align}
    & F_k^{(i)}\cap F_{k}^{(j)}=\emptyset \text{~for~all~} i\neq j\in[2^{r_{c}}], \label{eq:Fkp_1}\\
    & F_k=\bigcup_{p\in[2^{r_c}]}F_k^{(p)}.\label{eq:Fkp_2}
\end{align}
Now we can see how the Gray Path Stage $U_{GrayPath}$ in Eq. \eqref{eq:Graycodepath} is realized by the above quantum circuit $C_1,C_2,\ldots,C_{2^{r_c}+1}$ step by step as follows. {For $j\in[2^{r_c}]$, define $\ket{f_j}:=\ket{\langle c_j^1t_{1}^{(k)},x\rangle,\langle c_j^2t_{2}^{(k)},x\rangle,\cdots,\langle c_j^{r_t}t_{r_t}^{(k)},x\rangle}$. Note that $\ket{f_1}=\ket{\langle 0^{r_c}t_{1}^{(k)},x\rangle,\langle 0^{r_c}t_{2}^{(k)},x\rangle,\cdots,\langle 0^{r_c}t_{r_t}^{(k)},x\rangle}$ since $c_1^{i}=0^{r_c}$ for all $ i\in[r_t]$. }
\begin{align*}
     &\ket{x_{control}}\ket{y^{(k)}} &\\
     &= \ket{x_{control}}\ket{f_1} &(\text{Eq. \eqref{eq:yk}})\\
     & \xrightarrow{C_1} e^{i\sum_{s\in F_{k}^{(1)}}\langle s,x\rangle\alpha_s}\ket{x_{control}} \ket{f_1} &(\text{Eq. \eqref{eq:Fkp}}) \\
     &\xrightarrow{C_2} e^{i\sum_{s\in F_{k}^{(1)}\cup F_{k}^{(2)}}\langle s,x\rangle\alpha_s}\ket{x_{control}} \ket{f_2} &(\text{Eq. \eqref{eq:Fkp},\eqref{eq:Fkp_1}})\\
     & ~~~~\vdots  \\
     &\xrightarrow{C_{2^{r_c}}} e^{i\sum_{s\in \bigcup_{p\in[2^{r_c}]}F_k^{(p)}}\langle s,x\rangle\alpha_s}\ket{x_{control}} \ket{f_{2^{r_c}}} &(\text{Eq. \eqref{eq:Fkp}, \eqref{eq:Fkp_1}})\\
     & = e^{i\sum_{s\in F_k}\langle s,x\rangle\alpha_s}\ket{x_{control}} \ket{f_{2^{r_c}}}&
      (\text{Eq. \eqref{eq:Fkp_2}})\\
     &\xrightarrow{C_{2^{r_c}+1}} e^{i\sum_{s\in F_k}\langle s,x\rangle\alpha_s}\ket{x_{control}}\ket{f_1}\\
     &= e^{i\sum_{s\in F_k}\langle s,x\rangle\alpha_s}\ket{x_{control}} \ket{y^{(k)}} &(\text{Eq. \eqref{eq:yk}})
\end{align*}

Next we analyze the depth. Phase 1 consists of rotations applied on different qubits in the target register, thus can be made in a single depth.
In each phase $p\in \{2,3,\ldots,2^{r_c}\}$, since $c_{p-1}^{i}$ and $c_p^i$ differ by only 1 bit, one CNOT gate suffices to implement the function $\langle c_p^{i}t_i^{(k)},x\rangle$ from $\langle c_{p-1}^{i}t_i^{(k)},x\rangle$ in the previous phase.
The control and target qubit of this CNOT gate is the $h_{ip}$-th qubit in control register and the $i$-th qubit in target register.
{According to Eq. \eqref{eq:index}, indices $h_{1p},h_{2p},\ldots,h_{r_tp}$ of control qubits are all different, and therefore, all the CNOT gates in step $p.1$ can be implemented in depth $1$.}
The rotations in step $p.2$ are on different qubits and thus fit in one depth as well. Similarly, phase $2^{r_c}+1$ can also be implemented in depth $1$.
Thus the total depth of Gray Path Stage is at most $1+2\cdot (2^{r_c}-1)+1=2\cdot 2^{r_c}$.
{The size of Gray Path Stage is at most $r_c \cdot 2\cdot 2^{r_c}= r_c2^{r_c+1}$.}
\end{proof}

\section{Proof of Lemma \ref{lem:enc-trans}}
\label{sec:enc-trans}
The Toffoli gate is a 3-qubit CCNOT gate where we flip the basis $\ket{0},\ket{1}$ of (i.e. apply $X$ gate to) the third qubit conditioned on the first two qubits are both on $\ket{1}$. This can be extended to an $n$-qubit Toffoli gate, which applies the $X$ gate to the last qubit conditioned on the first $(n-1)$ qubits all being on $\ket{1}$. An $n$-qubit Toffoli gate can be implemented by a circuit of $O(n)$ size and depth { \cite{multi-controlled-gate}}. 
\begin{lemma} \label{lem:tof}
An $n$-qubit Toffoli gate can be implemented by a quantum circuit of depth and size $O(n)$ without ancillary qubits.
\end{lemma}
The next result we need says that cascading CNOT gates with the same target qubit can be exponentially compressed \cite{moore2001parallel}.  
\begin{lemma}\label{lem:stair_circuit}
    Let $C$ be a quantum circuit consisting of $n$ CNOT gates with the same target qubit and distinct controlled qubits. Then $C$ can be compressed to $O(\log n)$ depth {and $O(n)$ size} without using ancillary qubits. 
\end{lemma}

Recall the description of Lemma \ref{lem:enc-trans}.

\noindent
{\bf Lemma \ref{lem:enc-trans}}
\emph{The following unitary transformation on $2^n$ qubits
\begin{equation*}
\ket{e_i}\to|i\rangle\ket{0^{2^n-n}},\text{~for~all~}i\in\{0\}\cup [2^n-1], e_i\in\{0,1\}^{2^n},
\end{equation*}
can be implemented by a quantum circuit using single-qubit gates and CNOT gate with  {$2^{n+1}$} ancillary qubits, of depth $O(n)$ {and size $O(2^n)$}.}

\begin{proof}
We will implement Eq. \eqref{eq:unary_transform} with $2^{n+1}$ ancillary qubits in three steps:
\begin{enumerate}
\renewcommand{\labelenumi}{Step \theenumi:}
\setlength{\itemindent}{2.2 em}
\setlength{\leftmargin}{1in}
    \item  
    
    {  $\underbrace{\ket{e_i}}_{2^n\atop\text{~qubits}}\ket{0^{2^{n+1}}} \to\ket{0^{2^n}}\underbrace{\ket{e_s}}_{2^{n/2}\atop\text{~qubits}}\underbrace{\ket{e_t}}_{2^{n/2}\atop\text{~qubits}}\ket{0^{2^{n+1}-2\cdot 2^{n/2}}}$ for all $s,t\in\{0\}\cup [2^{n/2}-1]$ and $i=s\cdot 2^{n/2}+t$.}
    
    \item  
    
    { $\ket{0^{2^n}}\ket{e_s}\ket{e_t}\ket{0^{2^{n+1}-2\cdot 2^{n/2}}}\to \ket{0^{2^n}}\ket{i}\ket{0^{2^{n+1}-n}}$ for all $s,t\in\{0\}\cup [2^{n/2}-1]$ and $i=s\cdot 2^{n/2}+t$.}
    
    \item  
    
    {$\ket{0^{2^n}}\ket{i}\ket{0^{2^{n+1}-n}}\to \ket{i}\ket{0^{3\cdot2^n-n}}$ for all $i\in\{0\}\cup [2^n-1]$.}
\end{enumerate}

{ In these three steps, the first $2^n$ qubits are  called register $A$. The last $ 2^{n+1}$ are ancillary qubits, which are initialized as $\ket{0}$ and called register $B$. Let the first $2^{n/2}$ qubits of register $B$ be register $B_1$ and the second $2^{n/2}$ qubits of register $B$ be register $B_2$.}

Firstly, we implement Step 1 by a quantum circuit of depth $O(n)$ {and size $O(2^n)$} with { $2^{n+1}$} ancillary qubits.  Step 1 consists of two phases.
\begin{itemize}
    \item Step 1a: 
    { $\ket{e_i}\ket{0^{2^{n+1}}}\to\ket{e_i}\ket{e_s}\ket{e_t}\ket{0^{2^{n+1}-2\cdot 2^{n/2}}}$} for all $s,t\in\{0\}\cup [2^{n/2}-1]$ and $i=s\cdot 2^{n/2}+t$.

Let $CNOT^i_{j,(k)}$ denote a CNOT gate whose controlled qubit is the $i$-th qubit of register $A$, and target qubit is the $j$-th qubit of register $B_k$  for $k\in[2]$.
Let $CNOT^i_{s,t}=CNOT^{i}_{s,(1)}CNOT^i_{t,(2)}$. Therefore, Step 1a can be realized by a CNOT circuit as follows:
\begin{align*}
    &\prod_{s,t=0}^{2^{n/2}-1} CNOT^{s\cdot 2^{n/2}+t}_{s,t}\\
    =&\left[\prod_{s,t=0}^{2^{n/2}-1}  CNOT^{s\cdot 2^{n/2}+t}_{s,(1)}\right]\left[\prod_{s,t=0}^{2^{n/2}-1}  CNOT^{s\cdot 2^{n/2}+t}_{t,(2)}\right]\\
    =&\left[\prod_{s=0}^{2^{n/2}-1}  \left(\prod_{t=0}^{2^{n/2}-1} CNOT^{s\cdot 2^{n/2}+t}_{s,(1)}\right) \right]\cdot \left[\prod_{t=0}^{2^{n/2}-1}\left(\prod_{s=0}^{2^{n/2}-1} CNOT^{s\cdot 2^{n/2}+t}_{t,(2)}\right)\right].
\end{align*}
For every $s\in\{0\}\cup[2^{n/2}-1]$, all CNOT gates in $\mathcal{C'}_s\defeq\prod_{t=0}^{2^{n/2}-1} CNOT^{s\cdot 2^{n/2}+t}_{s,(1)}$ have different controlled qubits and the same target qubit. According to Lemma \ref{lem:stair_circuit}, $\mathcal{C'}_s$ can be implemented by a circuit of depth $O(n)$ and size $O(2^{n/2})$ without ancillary qubits. For all $s\in\{0\}\cup[2^{n/2}-1]$, $C_s'$ act on different qubits. Therefore, they can be paralleled and $\prod_s\mathcal{C'}_s$ can be implemented by a circuit of depth $O(n)$ and size $O(2^{n})$
without ancillary qubits.
By similar discussion, $\prod_t\mathcal{C''}_t$ can be implemented by a circuit of depth $O(n)$ and size $O(2^{n})$
without ancillary qubits, where $\mathcal{C''}_t\defeq\prod_{t=0}^{2^{n/2}-1} CNOT^{s\cdot 2^{n/2}+t}_{t,(2)}$. So Step 1a can be implemented by a quantum circuit of depth $O(n)$ {and size $O(2^{n})$} without ancillary qubits.
    \item Step 1b: 
    { $\ket{e_i}\ket{e_s}\ket{e_t}\ket{0^{2^{n+1}-2\cdot 2^{n/2}}}\to \ket{0^{{2^n}}}\ket{e_s}\ket{e_t}\ket{0^{2^{n+1}-2\cdot {2^{n/2}}}}$ for all $s,t\in\{0\}\cup [2^{n/2}-1]$ and $i=s\cdot 2^{n/2}+t$.}
    Let $\textsf{T}^{s,t}_{i}$ denotes a 3-qubit Toffoli gate, whose controlled qubits are the $s$-th qubit in register $B_1$ and $t$-th qubit in register $B_2$, and the target qubit is the $i$-th qubit in register $A$. The unitary transform of Step 1b is realized by applying all Toffoli gates $\textsf{T}^{s,t}_{s\cdot 2^{n/2}+t}$. To reduce the circuit depth, we make { $2^{n/2}-1$} copies of register $B_1,B_2$ in last $2^{n+1}$ qubits of register $B$: 
    {     \[\ket{e_i}\ket{e_s}\ket{e_t}\ket{0^{2^{n+1}-2\cdot 2^{n/2}}}\to \ket{e_i}\underbrace{\ket{e_s}\ket{e_t}\cdots \ket{e_s}\ket{e_t}}_{2^{n/2}~\text{copies of~}\ket{e_s}\ket{e_t}}\]}
    for all $s,t\in\{0\}\cup [2^{n/2}-1]$ and $i=s\cdot 2^{n/2}+t$. This transformation can be parallelized to depth $O(n)$ (by Lemma \ref{lem:copy1}) {and size $O(2^{n})$}. Because there are $2^{n/2}$ copies of $\ket{e_s}\ket{e_t}$, all Toffoli gates $\textsf{T}^{s,t}_{s\cdot 2^{n/2}+t}$ are on distinct control and target qubits, thus can be executed in parallel in depth $O(1)$. Finally, restore the last  {$2^{n+1}-2\cdot 2^{n/2}$} qubits of register $B$ to all-zero state in $O(n)$ depth. Overall, Step 1b can be implemented by a quantum circuit of depth $O(n)$ {and size $O(2^{n})$} with { $2^{n+1}$} ancillary qubits.
\end{itemize}
Secondly, Step 2 can be realized by a circuit of depth $O(n)$ {and size $O(n2^{n/2})$} with $O(n2^{n/2})$ ancillary qubits. Now, we rewrite the transformation of Step 2:
{\[\ket{0^{2^n}}\ket{e_s}\ket{e_t}\ket{0^{2^{n+1}-2\cdot 2^{n/2}}}\to \ket{0^{2^n}}\underbrace{\ket{s}}_{\frac{n}{2}~\text{qubits}}\underbrace{\ket{t}}_{\frac{n}{2}~\text{qubits}}\ket{0^{2^{n+1}-n}}\]}
for all $s,t\in\{0\}\cup [2^{n/2}-1]$. If we can realize transformation 
\begin{equation}\label{eq:unary_tranformation1}
\ket{e_s}\ket{0^k}\to\ket{s}\ket{0^{2^k}},\text{for~} e_s\in\{0,1\}^{2^k},~s\in\{0\}\cup[2^k-1]
\end{equation}
by a quantum circuit of depth $O(k)$ {and size $O(k2^k)$} with {$ k2^{k}$} ancillary qubits, then we can implement Step 2 by a quantum circuit of depth $O(n)$ {and size $O(n2^{n/2})$} with {$ (n/2)2^{n/2}$} ancillary qubits.

We will implement Eq. \eqref{eq:unary_tranformation1} with {$ k2^{k}$} ancillary qubits in three steps:
\begin{itemize}
    \item [] Step 2a: $\ket{e_s}\ket{0^k}\ket{0^{k2^k}}\to\ket{e_s}\ket{s}\ket{0^{k2^k}}$ for all $s\in\{0\}\cup [2^k-1]$;
    \item [] Step 2b: $\ket{e_s}\ket{s}\ket{0^{k2^k}}\to\ket{0^{2^k}}\ket{s}\ket{0^{k2^k}}$ for all $s\in\{0\}\cup [2^k-1]$;
    \item [] Step 2c: $\ket{0^{2^k}}\ket{s}\ket{0^{k2^k}}\to \ket{s}\ket{0^{2^k}}\ket{0^{k2^k}}$ for all $s\in\{0\}\cup [2^k-1]$.
\end{itemize}
The first $2^k$ qubits are called register $A$ and the second $n$ qubits are called register $B$. The last $k2^k$ qubits are ancillary qubits, which are called register $C$. Let $CNOT^s_j$ denote a CNOT gate, whose controlled qubit is the $s$-th qubit in register $A$ and target qubit is the $j$-th qubit in register $B$ for all $i\in\{0\}\cup[2^k-1]$ and $j\in\{0\}\cup [n-1]$. Define $CNOT^s_{S_s}\defeq \prod_{j\in S_s}CNOT^s_j$, where $S_s\defeq \left\{j|s_j=1\text{~for~}j\in \{0\}\cup [k-1]\right\}$. 
\begin{itemize}
    \item Step 2a: Firstly, we implement Step 2a by a quantum circuit with $k2^k$ ancillary qubits of depth $O(k)$. It can be easily verified that Step 2a can be implemented by a CNOT circuit
\begin{align*}
&\prod_{s\in\{0\}\cup [2^k-1]}CNOT^s_{S_s}=\prod_{s\in\{0\}\cup [2^k-1]}\prod_{j\in S_s}CNOT^s_j= \prod_{s\in \{0\}\cup [k-1]}\prod_{s:s\in\{0\}\cup [2^k-1],s_j=1} CNOT^s_j.
\end{align*}
For every $j\in \{0\}\cup [k-1]$, CNOT circuit $C_j\defeq\prod_{s:s\in\{0\}\cup [2^k-1],s_j=1} CNOT^s_j$ consists of $2^{k-1}$ CNOT gates. All these CNOT gates have distinct control qubits and the same target qubit. According to Lemma \ref{lem:stair_circuit}, $C_j$ can be parallelized to depth $O(k)$ {and size $O(2^k)$} without ancillary qubits. Step 2a consists of $C_0,C_1,\ldots,C_{k-1}$ and the target qubits of CNOT gates in $C_t$ and $C_\ell$ are different if $t\neq \ell$. In order to reduce the depth of Step 2a, we make $k$ copies of register $A$ in ancillary qubits (register $C$). Then $C_0,C_1,\ldots,C_{k-1}$ can be implemented simultaneously using the $k$ copies of register $A$. 
Thus $\prod_{j=0}^{k-1} C_j$ can be implemented simultaneously in depth $O(k)$ {and size $O(k2^k)$}. Finally, we reset register $C$ back to 0 in {depth $O(k)$ and size $O(2^k)$}. Step 1 is summarized as follows:
\begin{align*}
    &\ket{e_s}\ket{0^k}\ket{0^{k2^k}} \\
    \to & \ket{e_s}\ket{0^k}\underbrace{\ket{e_s}\cdots\ket{e_s}}_{k ~\text{copies~of}~\ket{e_s}}
    & (\text{Lemma \ref{lem:copy1}, depth  $O(\log k)$, {size $O(k2^k)$}})\\
    \to & \ket{e_s}\ket{s} \ket{e_s}\cdots\ket{e_s}
    &(\text{Lemma \ref{lem:stair_circuit}, depth  $O(k)$, {size $O(k2^k)$}})\\
    \to & \ket{e_s}\ket{s} \ket{0^{k2^k}} 
    &  (\text{Lemma \ref{lem:copy1}, depth  $O(\log k)$, {size $O(k2^k)$}})
\end{align*}
The total depth and size of Step 2a are $O(\log k)+O(k)+O(\log k)=O(k)$ { and $O(k2^k)$}, respectively.

\item Step 2b: Secondly, we implement Step 2b by a quantum circuit with $k2^k$ ancillary qubits of depth $O(k)$ {and size $O(k2^k)$}. Define an $(k+1)$-qubit quantum gate $\text{Tof}_s$ acting on register $B$ and the $s$-th qubit in register $A$:
\begin{align*}
    \text{Tof}_s\ket{x}|j\rangle=\ket{x\oplus \delta_{sj}}\ket{j}, \text{~~for~}s, j\in\{0\}\cup [2^k-1],
\end{align*}
where $|x\rangle$ is the $s$-th qubit in register $A$ and $\ket{j}$ is in register $B$. That is, conditioned on the state in register $B$ is $\ket{s}$, $\text{Tof}_s$ flips the $s$-th qubit in register A. Step 2b is just  $\prod_{s=0}^{2^k-1} \text{Tof}_s$.

Any $\text{Tof}_s$ can be implemented by $[I\otimes(\otimes_{j=0}^{k-1} X^{s_j})]\text{Tof}_{2^k-1}[I\otimes(\otimes_{j=0}^{k-1} X^{s_j})]$, where Toffoli gate $\text{Tof}_{2^k-1}$ can be implemented in depth $O(k)$ and size $O(k)$ without ancillary qubits (Lemma \ref{lem:tof}). Therefore $\text{Tof}_s$ can be realized by an $O(k)$-depth { and $O(k)$-size} quantum circuit without ancillary qubits. To realize $\text{Tof}_0,\ldots,\text{Tof}_{2^k-1}$ simultaneously, we make $2^k$ copies of register $B$ in register $C$ depth $O(k)$ { and $O(k2^k)$}. Then  $\text{Tof}_0,\ldots,\text{Tof}_{2^k-1}$ can be implemented simultaneously by using these copies. Finally, we reset register $C$ back to 0 in depth $O(k)$ { and size $O(k2^k)$}. Step 2b can be summarized as follows:
\begin{align*}
    &\ket{e_s}\ket{s}\ket{0^{k2^k}} \nonumber\\
    \to & \ket{e_s}\ket{s}\underbrace{\ket{s}\cdots\ket{s}}_{2^k \text{~copies~of~}\ket{s}} 
    & (\text{Lemma \ref{lem:copy1}, depth  $O(k)$, size~} {O(k2^k)}
    )\\
    \to & \ket{0^{2^k}}\ket{s} \ket{s}\cdots\ket{s}
    &(\text{depth  $O(k)$, size }{ O(k2^k)})\\
    \to & \ket{0^{2^k}}\ket{s} \ket{0^{k2^k}} 
    &  (\text{Lemma \ref{lem:copy1}, depth  $O(k)$, size } { O(k2^k)})
\end{align*}
The total depth and size of Step 2b are $O(k)+O(k)+O(k)=O(k)$ { and $O(k2^k)$}, respectively.

\item Step 2c: Thirdly, for Step 2c, we swap the first $k$ qubits in register $A$ and register $B$ by $k$ swap gates. Hence, Step 2c can be implemented in depth $O(1)$ { and size $O(k)$}.
\end{itemize}
Thirdly, for Step 3, we swap the first $n$ qubit in register $A$ and register $B$ in depth $O(1)$ and size $O(n)$ without ancillary qubits by swap gates.
\end{proof}

\section{Circuit depth for general unitary synthesis}
\label{sec:US_depth}

In Eq. \eqref{matrix:Vn}, we called a $(2\times 2)$-block diagonal matrix $V_j$ a $j$-qubit UCG. In the view of a circuit, this is a multiple controlled gate where the target qubit is the last one and the conditions are on the first $j-1$ qubits. But this target qubit can actually be any one, and all the implementations in Section \ref{sec:QSP_withancilla} and Section \ref{sec:QSP_withoutancilla} still apply.
Let $V^n_k$ denote an $n$-qubit UCG whose index of target qubit is $k$. By repeatedly applying cosine-sine decomposition, one can factor an arbitrary unitary matrix $U$ into a sequence of UCGs as follows \cite{mottonen2005decompositions}. Recall that the Ruler function $\zeta(n)$ is defined as $\zeta(n)=\max\{k:2^{k-1}|n\}$. 
\begin{lemma}\label{lem:CSD}
Any $n$-qubit unitary matrix $U\in\mathbb{C}^{2^n\times 2^n}$ can be decomposed as $U=V^n_{n}(0) \cdot \prod_{i=1}^{2^{n-1}-1} V^n_{n-\zeta(i)}(i)  \cdot V^n_{n}(2^{n-1}),$
where different $i$ in $V_k^n(i)$ denote different forms of $n$-qubit UCGs despite the same target qubit $k$.
\end{lemma}
The proofs of Theorem \ref{thm:unitary} and Corollary \ref{coro:unitary} in Section \ref{sec:extensions} are shown as follows.

{\noindent\bf Theorem \ref{thm:unitary}}
\emph{Any unitary matrix $U\in\mathbb{C}^{2^n\times 2^n} $ can be implemented by a quantum circuit of depth $O\big(n2^n+\frac{4^n}{m+n}\big)$ {and size $O(4^n)$} with $m \le 2^n$ ancillary qubits.
}
\begin{proof}
Based on Eq. \eqref{eq:UCG}, Lemma \ref{lem:DU_with_ancillary} and Lemma \ref{lem:DU_without_ancillary}, given $m \le 2^n$ ancillary qubits, any $n$-qubit UCG $V_k^n(i)$  can be implemented by a circuit of {size $O(2^n)$ and }depth $O\big(n+\frac{2^n}{n+m}\big)$.
Since Lemma \ref{lem:CSD} shows that any $U$ can be decomposed into $O(2^n)$ many $n$-qubit UCGs, a circuit can simply implement them sequentially 
to realize $U$, yielding a circuit of {size $O(2^n)\cdot O(2^n)=O(4^n)$ and} depth $O(2^n)\cdot  O\big(n+\frac{2^n}{m+n}\big)=O\big(n2^n+\frac{4^n}{m+n}\big)$. 
\end{proof}

{\noindent \bf Corollary \ref{coro:unitary}}
\emph{The minimum circuit depth $D_{\textsc{Unitary}}(n,m)$ for an arbitrary $n$-qubit unitary with $m$ ancillary qubits satisfies
\[
\begin{cases}
    D_{\textsc{Unitary}}(n,m) = \Theta\big(\frac{4^n}{m+n}\big), & \text{if } m=O(2^n/n), \\
    D_{\textsc{Unitary}}(n,m) \in \left[ \Omega(n),  O(n2^n)\right], & \text{if }m = \omega(2^n/n).
\end{cases}
\]}

\begin{proof}
The lower bound for the number of CNOT gates is $\Omega(4^n)$ by a similar argument. The only difference is that with $m$ ancillary qubits, the number of single-qubit gates right before the output is at most $n+m$ instead of $n$. We thus have $4k+3(n+m) \ge 4^n-1$. When $m = O\big(2^n/n\big)$, this still gives $k = \Omega(4^n)$. Since each layer can have at most $(m+n)/2$ CNOT gates, we know that it needs at least $\Omega\big(\frac{4^n}{m+n}\big)$ depth for any circuit of $n$ input qubits and $m$ ancillary qubits.
Theorem \ref{thm:lowerbound_QSP} shows a lower bound $\Omega(n)$ for an $n$-qubit QSP circuit, which is a special case of a circuit for an $n$-qubit unitary matrix.
Therefore, we get a depth lower bound of $\Omega\big(\max\big\{n,\frac{4^n}{n+m}\big\}\big)$ for an $n$-qubit unitary matrix. Putting the depth upper bound $O\big(n2^n+\frac{4^n}{n+m}\big)$ and lower bound $\Omega\big(\max\big\{n,\frac{4^n}{n+m}\big\}\big)$ together, we complete the proof.
\end{proof}

\section{Decomposition with Clifford + T gate set}
\label{sec:clifford_decomposition}
\begin{lemma}[\cite{ross2015optimal}]\label{lem:sk}
For $\epsilon>0$, any rotation $R_z(\theta)\in\mathbb{C}^{2\times 2}$ can be $\epsilon$-approximated by a quantum circuit consisting of $O(\log (1/\epsilon))$ many $H$ and $T$ gates, without ancillary qubits.
\end{lemma}

Based on this lemma, it is not hard to extend our results on the exact implementation of diagonal unitary matrix $\Lambda_n$ to its approximate version (Lemma \ref{lem:diagonal_approx}). This in turn gives approximate realization of UCGs $V_n$ (Lemma \ref{lem:UCG_apprx}), state preparation (Corollary \ref{coro:psi_apprx}), and unitary operation (Corollary \ref{coro:unitary_apprx}).

\begin{lemma}\label{lem:diagonal_approx}
Any $n$-qubit diagonal unitary matrix $\Lambda_n$ can be $\epsilon$-approximated by a quantum circuit of depth $O\left(n+\frac{2^n\log(2^n/\epsilon)}{m+n}\right)$, using the Clifford+T gate set with $m$ ancillary qubits.
\end{lemma}
\begin{proof}
In Section \ref{sec:QSP_withancilla} and Section \ref{sec:QSP_withoutancilla}, our quantum circuits for $\Lambda_n$ consist of only CNOT gates and $2^n-1$ rotation gates $R(\alpha)$ for $\alpha\in\mathbb{R}$. 
By Lemma \ref{lem:sk}, every $R(\alpha)$ can be $(\epsilon/2^n)$-approximated by $O(\log(2^n/\epsilon))$ $H$ and $T$ gates up to a global phase. The overall accuracy of the circuit can then be seen from a union bound.

If $m\in[2n,2^n]$, the total circuit depth of $\Lambda_n$ is $O(\log m) +O(\log m + \log(2^n/\epsilon))+O(\log m)+O(2^n\log(2^n/\epsilon)/m)=O(\log m + 2^n\log(2^n/\epsilon)/m)=O(n+2^n\log(2^n/\epsilon)/m).$
If $m\le 2n$, the depth of the circuit implementing the diagonal unitary matrix is $O(2^n\log(2^n/\epsilon)/n)$.
Putting these two results together, we complete the proof.
\end{proof}

\begin{corollary}\label{lem:UCG_apprx}
Any $n$-qubit UCG $V_n$ can be $\epsilon$-approximated  by a quantum circuit using the Clifford+T gate set, of depth $O\left(n+\frac{2^n\log(2^n/\epsilon)}{m+n}\right)$ with $m$ ancillary qubits.
\end{corollary}
\begin{proof}
From Eq. \eqref{eq:UCG}, we can see that any $n$-qubit UCG can be decomposed into three $n$-qubit diagonal unitary matrices, two $S$ gates and two $H$ gates. By Lemma \ref{lem:diagonal_approx}, every $\Lambda_n$ can be $(\epsilon/3)$-approximated by a quantum circuit of depth $O\left(n+\frac{2^n\log(3\cdot2^n/\epsilon)}{m+n}\right)$ with $m$ ancillary qubits. Hence, $V_n$ can be $\epsilon$-approximated by a circuit of depth $3\times O\left(n+\frac{2^n\log(3\cdot2^n/\epsilon)}{m+n}\right)+2+2=O\left(n+\frac{2^n\log(2^n/\epsilon)}{m+n}\right)$.
\end{proof}

The approximate implementations of state preparation and general unitary matrix are shown in Corollary \ref{coro:psi_apprx} and Corollary \ref{coro:unitary_apprx}.

\noindent
{\bf Corollary \ref{coro:psi_apprx}}
\emph{
    For any $n$-qubit target state $\ket{\psi_v}$ and $\epsilon>0$, one can prepare a state $\ket{\psi'_v}$ which is $\epsilon$-close to $\ket{\psi_v}$ in $\ell_2$-distance, by a quantum circuit consisting of $\{CNOT,H,S,T\}$ gates of depth 
    {\[\left\{\begin{array}{ll}
        O\left(\frac{2^n\log(2^n/\epsilon)}{m+n}\right) &  \text{if~}m=O(2^n/(n\log n)),\\
         O(n\log n\log(2^n/\epsilon))& \text{if~}m\in [\omega(2^n/(n\log n),o(2^n)],\\
         O(n \log(2^n/\epsilon))& \text{if~}m=\Omega(2^n),\\
    \end{array}\right.\]}
where $m$ is the number of ancillary qubits.
}
\begin{proof}
{According to Theorem \ref{thm:QSP_anci}, any $n$-qubit quantum state $\ket{\psi_v}$ can be prepared by a quantum circuit $QSP$, using single-gates and CNOT gates, of size $c\cdot 2^n$ for some constant $c>0$ and depth
\[d=\left\{\begin{array}{ll}
        O\left(\frac{2^n}{m+n}\right) &  \text{if~}m=O(2^n/(n\log n)),\\
         O(n\log n)& \text{if~}m\in [\omega(2^n/(n\log n),o(2^n)],\\
         O(n)& \text{if~}m=\Omega(2^n),\\
    \end{array}\right.\]
Based on Eq. \eqref{eq:single_qubit_gate} and Lemma \ref{lem:sk}, every single-qubit gate can be $(\epsilon/c2^n)$-approximated by a quantum circuit consisting of $O(\log((2^n)/\epsilon))$ Clifford+T gates. Approximate all single-qubit gates in this way. The depth of the new quantum circuit $QSP'$ consisting of Clifford+T gates is $d\times O(\log(2^n/\epsilon))=O(d\log(2^n/\epsilon))$. And circuit $QSP'$ prepare a quantum state $\ket{\psi'_v}$ satisfying
\[\|\ket{\psi_v}-\ket{\psi_v'}\|_2=\|(QSP-QSP')\ket{0}^{\otimes n}\|_2\le \frac{\epsilon}{c2^n} \times c2^n= \epsilon.\]}
\end{proof}

\noindent{\bf Corollary \ref{coro:unitary_apprx}}
\emph{Any $n$-qubit general unitary matrix can be implemented by a quantum circuit, using the $\{CNOT,H,S,T\}$ gate set, of depth $O\left(n2^n+\frac{4^n\log(4^n/\epsilon)}{m+n}\right)$ with $m$ ancillary qubits.}
\begin{proof}
Lemma \ref{lem:CSD} shows that any $U\in\mathbb{C}^{2^n\times 2^n}$ can be decomposed into $2^n-1$ $n$-qubit UCGs. According to Lemma \ref{lem:UCG_apprx}, an $n$-qubit UCG $V_n$ can be $\epsilon/(2^{n}-1)$-approximated by a quantum circuit $V_n'$ consisting of $\{CNOT,H,S,T\}$ gates in depth $O\left(n+\frac{2^n\log(4^n/\epsilon)}{m+n}\right)$ with $m$ ancillary qubits. Hence, $U$ can be $\epsilon$-approximated by a quantum circuit in depth $O\left(n+\frac{2^n\log(4^n/\epsilon)}{m+n}\right)\times (2^n+1)=O\left(n2^n+\frac{4^n\log(4^n/\epsilon)}{m+n}\right)$.
\end{proof}

\section{Sparse quantum state preparation}
\label{sec:sparse_QSP}
A vector $v=(v_0,v_1,\ldots,v_{2^n-1})\in \mathbb{C}^{2^n}$ is said to be $s$-sparse if there are at most $s$ nonzero elements in $v$. In this section, we consider how to efficiently prepare $s$-sparse states.


\begin{lemma}\label{lem:CNOT_sum}
The unitary transformation defined by
\begin{equation}\label{eq:CNOT_sum}
    \ket{x_1x_2\cdots x_n}\ket{t}\to \ket{x_1x_2\cdots x_n}\ket{\bigoplus_{i=1}^n x_i \oplus t}
\end{equation}
$\forall x_1,\ldots,x_n,t\in\B$, can be implemented in depth $O(\log(n))$.
\end{lemma}
\begin{proof}
The circuit implementation of Eq. \eqref{eq:CNOT_sum} is shown in Figure \ref{fig:CNOT_sum}.
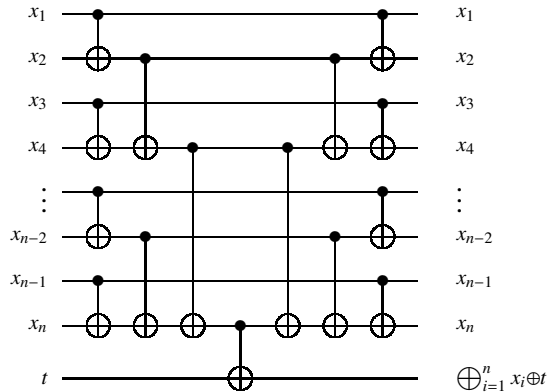
\begin{figure}[!hbt]
    \centerline{
    \Qcircuit @C=0.8em @R=1em {
    \lstick{\scriptstyle x_1}&\ctrl{1} & \qw & \qw & \qw &\qw &\qw & \ctrl{1} & \qw & \rstick{\scriptstyle x_1}\\
   \lstick{\scriptstyle x_2} &\targ & \ctrl{2} &\qw & \qw &\qw &\ctrl{2} & \targ &\qw &  \rstick{\scriptstyle x_2}\\
    \lstick{\scriptstyle x_3}&\ctrl{1} & \qw &\qw & \qw &\qw &\qw &\ctrl{1} &\qw &  \rstick{\scriptstyle x_3}\\
    \lstick{\scriptstyle x_4}&\targ & \targ & \ctrl{4} &\qw &\ctrl{4}& \targ &\targ &\qw &  \rstick{\scriptstyle x_4}\\
    \lstick{\scriptstyle \vdots}&\ctrl{1} & \qw & \qw & \qw &\qw &\qw &\ctrl{1} &\qw &  \rstick{\scriptstyle \vdots}\\
    \lstick{\scriptstyle x_{n-2}}&\targ & \ctrl{2} &\qw & \qw &\qw &\ctrl{2} &\targ &\qw &  \rstick{\scriptstyle x_{n-2}}\\
    \lstick{\scriptstyle x_{n-1}}&\ctrl{1} & \qw &\qw & \qw &\qw &\qw &\ctrl{1} &\qw & \rstick{\scriptstyle x_{n-1}}\\
    \lstick{\scriptstyle x_{n}}&\targ & \targ &\targ &\ctrl{1} &\targ & \targ &\targ &\qw &\rstick{\scriptstyle x_{n}}\\
    \lstick{\scriptstyle t}&\qw & \qw &\qw &\targ &\qw & \qw &\qw &\qw & \rstick{\scriptstyle \bigoplus_{i=1}^n x_i \oplus t}\\
    }
    }
    \caption{An $O(\log(n))$-depth circuit implementation of Eq. \eqref{eq:CNOT_sum}. }
    \label{fig:CNOT_sum}
\end{figure}
\end{proof}

\begin{lemma}\label{lem:perm_sparse}
Suppose that we are given two sets $S_1\subseteq\B^{n_1}$ and $S_2\subseteq\{0,1\}^{n_2}$, both of size $s$, and also given a bijection $P:S_1\to S_2$. Then a unitary transformation satisfying
\begin{equation}\label{eq:perm_sparse}
\ket{x}\ket{0^{n_2}}\to \ket{x}\ket{P(x)}, \forall x\in S_1
\end{equation}
can be implemented by a quantum circuit of depth $O\big(n_2\log(m)+\frac{(n_1+\log(m))sn_1n_2}{m}\big)$, using $m~(\ge 2n_1)$ ancillary qubits.
\end{lemma}
\begin{proof}
Define an $(n_1+1)$-qubit unitary ${\sf Tof}^x_a\ket{z}\ket{b}:=\ket{z}\ket{a\cdot \delta_{xz}\oplus b }$ for any $x,z\in\B^{n_1}$ and $a,b\in \B$, where $\delta_{xz}=1$ if $x=z$ and $\delta_{xz}=0$ if $x\neq z$. Unitary ${\sf Tof}^x_a$ can be implemented by an $n_1$-qubit Toffoli gate and at most $2n_1$ Pauli $X$ gates. According to Lemma \ref{lem:tof}, ${\sf Tof}^x_a$ can be implemented in depth $O(n_1)$.

Let $\{x(1),x(2),\ldots, x(s)\}$ be the elements of $S_1$, and $P(x(i))=y(i)$.  Denote $y(i)=y_1(i)y_2(i)\cdots y_{n_2}(i)$, and we will compute $y(i)$ bit by bit.

We start by computing $y_1(i)$, the first bit of $y(i)$, and then all other bits can be computed similarly. More precisely, we aim to implement the following unitary transformation $U_1$:
\[\ket{x(i)}\ket{0}\xrightarrow{U_1} \ket{x(i)}\ket{y_1(i)},\forall i\in[s].\]
Unitary $U_1$ can be implemented in 3 steps.
\begin{itemize}
    \item Step 1: make $\frac{m}{2n_1}$ copies of $\ket{x(i)}$ using $\frac{m}{2}$ ancillary qubits, i.e., implement the following unitary transformation
    \[\ket{x(i)}\ket{0}\ket{0^{m/2}}\to \ket{x(i)}\ket{0}\ket{x(i)x(i)\cdots x(i)}, \forall i \in[s]. \]
    This step can be implemented in depth $O(\log(m))$ by Lemma \ref{lem:copy1}.
    
    \item Step 2: implement the following unitary transformation using $\frac{m}{2n_1}$ ancillary qubits
    \begin{align*}
        & \ket{x(i)}\ket{0}\ket{x(i)x(i)\cdots x(i)}\ket{0^{\frac{m}{2n_1}}}
    \to \ket{x(i)}\ket{y_1(i)}\ket{x(i)x(i)\cdots x(i)}\ket{0^{\frac{m}{2n_1}}},\forall i \in[s].
    \end{align*}
    We divide set $S_1$ into $\frac{s}{m/2n_1}=\frac{2sn_1}{m}$ parts $S_1^{(1)},\ldots,S_1^{(\frac{2sn_1}{m})}$. The size of each part is $\frac{m}{2n_1}$. For the first part $S_1^{(1)}:=\{x(i): i\in[\frac{m}{2n_1}]\}$, we implement the following unitary transformation:
    \begin{equation}\label{eq:part1}
       \ket{x(i)}\ket{0}\ket{x(i)x(i)\cdots x(i)}\ket{0^{\frac{m}{2n_1}}}
    \to \left\{\begin{array}{ll}
         \ket{x(i)}\ket{y_1(i)}\ket{x(i)\cdots x(i)}\ket{0^{\frac{m}{2n_1}}},    & \text{if~}x(i) \in S_1^{(1)} \\
          \ket{x(i)}\ket{0}\ket{x(i)\cdots x(i)}\ket{0^{\frac{m}{2n_1}}},   & \text{otherwise.}
        \end{array}\right.
    \end{equation}
    Namely, we compute $y_1(i)$ for those $x(i)\in S_1^{(1)}$ (and keep other $x(i)$ untouched). This can be achieved by the following three sub-steps.
    In step 2.1, we apply unitaries ${\sf Tof}^{x(1)}_{y_1(1)}, { \sf Tof}^{x(2)}_{y_1(2)},\ldots,  { \sf Tof}^{x(\frac{m}{2n_1})}_{y_1(\frac{m}{2n_1})}$. For each ${ \sf Tof}^{x(i)}_{y_1(i)}$,  the control qubits are the $i$-th copy of $x(i)$ and the target qubit is the $i$-th qubit of ancillary qubits (those in $\ket{0^{\frac{m}{2n_1}}}$). Therefore, they can be realized in parallel by a circuit of depth $O(n_1)$. The effect of this step 2.1 is
    \begin{align*}
        & \ket{x(i)}\ket{0}\ket{x(i)\cdots x(i)}\ket{0^{\frac{m}{2n_1}}} \to \ket{x(i)}\ket{0}\ket{x(i)\cdots x(i)}\ket{0^{i-1}y_1(i)0^{\frac{m}{2n_1}-i}},\quad\forall x(i) \in S_1^{(1)}.
    \end{align*}
    In step 2.2, we implement the following unitary transformation
    \begin{align*}
        & \ket{x(i)}\ket{0}\ket{x(i)x(i)\cdots x(i)}\ket{0^{i-1}y_1(i)0^{\frac{m}{2n_1}-i}} \to \ket{x(i)}\ket{y_1(i)}\ket{x(i)x(i)\cdots x(i)}\ket{0^{i-1}y_1(i)0^{\frac{m}{2n_1}-i}},
    &\forall x(i) \in S_1^{(1)}.
    \end{align*}
    in depth $O(\log(m/(2n_1)))$ according to Lemma \ref{lem:CNOT_sum}.
    In step 2.3, we restore the ancillary qubits by an inverse circuit of step 2.1 of depth $O(n_1)$, i.e., we realize the following unitary transformation:
    \begin{align*}
        & \ket{x(i)}\ket{y_1(i)}\ket{x(i)x(i)\cdots x(i)}\ket{0^{i-1}y_1(i)0^{\frac{m}{2n_1}-i}}\to \ket{x(i)}\ket{y_1(i)}\ket{x(i)x(i)\cdots x(i)}\ket{0^{\frac{m}{2n_1}}}, \forall x(i) \in S_1^{(1)}.
    \end{align*}
    We can verify that step 2.1, 2.2 and 2.3 realize unitary in Eq. \eqref{eq:part1} by circuit of depth $O(n_1+\log(m/n_1))=O(n_1+\log(m))$. By the same discussion, for every $j\in[\frac{2sn_1}{m}]$,
        \begin{equation}\label{eq:partj}
         \ket{x(i)}\ket{0}\ket{x(i)x(i)\cdots x(i)}\ket{0^{\frac{m}{2n_1}}}
    \to \ket{x(i)}\ket{y_1(i)}\ket{x(i)x(i)\cdots x(i)}\ket{0^{\frac{m}{2n_1}}},\forall x(i) \in S_1^{(j)}.
    \end{equation}
    can be implemented in depth $O(n_1+\log(m))$.
    In summary, step 2 can be implemented in depth $O\big(n_1+\log(m)\big)\cdot\frac{2sn_1}{m}=O\big(\frac{(n_1+\log(m))sn_1}{m}\big)$.
    
    \item Step 3: restore the ancillary qubits by an inverse circuit of step 1, of depth $O(\log(m))$.
\end{itemize}
In summary, unitary $U_1$ can be implement in depth $2\cdot O(\log(m))+O\big(\frac{(n_1+\log(m))sn_1}{m}\big)=O\big(\log(m)+\frac{(n_1+\log(m))sn_1}{m}\big)$.
By similar discussion of $U_1$, for all $j\in[n_2]$, the following unitary $U_j$
\[\ket{x(i)}\ket{0}\xrightarrow{U_j} \ket{x(i)}\ket{y_j(i)},\forall i\in[s].\]
can be also implemented in depth $O\big(\log(m)+\frac{(n_1+\log(m))sn_1}{m}\big)$. 
By applying $U_1,U_2,\ldots, U_{n_2}$, we compute all $n_2$ bits of $y(i)$. 
This implements unitary in Eq. \eqref{eq:perm_sparse} and the total depth is $n_2\cdot O\big(\log(m)+\frac{(n_1+\log(m))sn_1}{m}\big)=O\big(n_2\log(m)+\frac{(n_1+\log(m))sn_1n_2}{m}\big)$.
\end{proof}

\begin{lemma}[\cite{malvetti2021quantum}]\label{lem:sparse_QSP_without_ancilla}
Any $n$-qubit $s$-sparse quantum state can be prepared by a quantum circuit of size $O(ns)$, using no ancillary qubits.
\end{lemma}

\begin{theorem}
For any $m\ge 0$, any $n$-qubit $s$-sparse quantum state can be prepared by a quantum circuit of depth $O(n\log(sn)+\frac{s\log(s)n^2}{n+m})$, using $m$ ancillary qubits.
\end{theorem}
\begin{proof}
For simplicity, we assume $\log (s)$ is an integer and $n'=\log(s)$. Define $S_1 = \B^{n'}$, and $S_2$ to be all $s$-sparse strings in $\B^n$. Let $P$ be any bijection from $S_1$ to $S_2$. Any $n$-qubit $s$-sparse quantum state can be represented as $\ket{\psi}=\sum\limits_{x\in\{0,1\}^{n'}}v_x\ket{P(x)}$.
\begin{itemize}
\item {\bf Case 1: $m\ge 3n$.}
First, we prepare an $n'$-qubit quantum state \[\ket{\psi'}=\sum\limits_{x\in\{0,1\}^{n'}}v_x\ket{x},\]
which can be implemented in depth $O((n')^2+\frac{2^{n'}}{n'+m})=O(\log^2(s)+\frac{s}{\log(s)+m})$ using $m$ ancillary qubtis by Lemma \ref{lem:partial_result}.

Second, we implement the following unitary transformations and then we complete preparing the state $\ket{\psi}$.
\begin{align}
    \ket{x0^{n-n'}}\ket{0^n}\to &\ket{x0^{n-n'}} \ket{P(x)}  \label{eq:perm_base1}\\
\to & \ket{0^n} \ket{P(x)} \label{eq:perm_base2}\\
\to & \ket{P(x)}\ket{0^n},\forall x\in\B^{n'}. \label{eq:perm_base3}
\end{align}
Based on Lemma \ref{lem:perm_sparse}, using $m$ ancillary qubits, Eq. \eqref{eq:perm_base1} can be implemented in depth $O(n\log(m)+\frac{(\log(s)+\log(m))s\log(s) n}{m})$. Eq. \eqref{eq:perm_base2}  can be viewed as a similar process by Lemma \ref{lem:perm_sparse}, though we need to swap $S_1$ and $S_2$ and reverse the direction of $P$. This can be implemented in depth $O(\log(s)\log(m)+\frac{(n+\log(m))s\log(s) n}{m})$, respectively. Eq. \eqref{eq:perm_base3} can be implemented by $n$ swap gates in depth 1. Therefore, if $m\le \frac{s\log(s)n}{\log(s)+\log(n)}$, the total depth is 
$O(n\log(m)+\frac{(\log(s)+\log(m))s\log(s) n}{m})+O(\log(s)\log(m)+\frac{(n+\log(m))s\log(s) n}{m})+1=O(n\log(sn)+\frac{s\log(s)n^2}{m})$. If $m> \frac{s\log(s)n}{\log(s)+\log(n)}$, we use only $\frac{s\log(s)n}{\log(s)+\log(n)}$ ancillary qubits and the total depth is $O(n\log(sn))$. Combining these two cases, the total depth is $O(n\log(sn)+\frac{s\log(s)n^2}{m})$.


    \item {\bf Case 2: $m < 3n$.} If $m<3n$, we do not use ancillary qubits. According to Lemma \ref{lem:sparse_QSP_without_ancilla}, $\ket{\psi}$ can be prepared in depth $O(ns)$.
\end{itemize}

Combining the above two cases, the circuit depth for $\ket{\psi}$ is  $O(n\log(sn)+\frac{s\log(s)n^2}{n+m})$.
\end{proof}

%% file: QSP_lowerbound.tex
In this section we prove Theorem \ref{thm:lowerbound_QSP}. In \cite{plesch2011quantum}, the authors presented a depth lower bound of $\Omega\left(\frac{2^n}{n}\right)$ for quantum circuits without ancillary  qubits. This can be extended to a lower bound of $\Omega\left(\frac{2^n}{n+m}\right)$ for circuits with $m$ ancillary qubits. Next we prove the linear lower bound.

\begin{lemma}
\label{lem:lower_bound_QSP}
Almost all $n$-qubit quantum states need a quantum circuit of depth at least $n-\log n - O(1)$ to prepare, even if the circuit uses arbitrary single- and double-qubit gates, regardless of the number of ancillary qubits.
\end{lemma}
\begin{proof}
As two adjacent single-qubit gates on the same qubit can be compressed into one, we can assume that a circuit with minimum depth has double-qubit gates and single-qubit gates appearing in alternative layers. Without loss of generality, assume that the odd layers contain only double-qubit gates and the even layers contain only single-qubit gates. Suppose the circuit depth is $D$, i.e. it has $D$ layers of gates. Note that a single- or double-qubit gate can be represented by $O(1)$ real parameters. 

Consider a time-space directed graph $G=(V,E)$ with $D+1$ layers $L_1, \ldots, L_{D+1}$ of nodes, corresponding to the $D+1$ time steps separated by the $D$ layers $U_1, \ldots, U_D$ of gates. There are $n+m$ nodes in each layer of $G$,  corresponding to the $n+m$ qubits in the circuit with $n$ input qubits and $m$ ancillary qubits. Edges appear only between nodes in adjacent layers $L_i, L_{i+1}$, and two nodes $(v_i,v_{i+1})\in E$ if $v_i\in L_i$, $v_{i+1}\in L_{i+1}$, and the two corresponding qubits of $v_i$ and $v_{i+1}$ are among the input and output qubits for some gate in $U_i$, the $i$-th layer of gates in the circuit. (Thus each single-qubit gate induces one edge, and each double-qubit gate induces 4 edges.) All edges are in the direction from input to output of the gate.  

Since the circuit generates the state, we have $U\ket{0^{n+m}} = \ket{\psi}\otimes \ket{\phi}$, where $\ket{\psi}$ is the target $n$-qubit state. Define the \emph{light cone} of $\ket{\psi}$ to be the nodes in $G$ that can reach, by walking along the directed edges, the nodes in $L_{D+1}$ corresponding to the qubits in $\ket{\psi}$. Intuitively, only gates within this region contributes to the generation of $\ket{\psi}$. Indeed, we can remove gates outside the light cone, from the last layer to the first, one by one. Each removal of a gate $U$ is equivalent to applying $U^\dagger$ at the end, which only affects $\ket{\phi}$ (and $\ket{\psi}$ remains unchanged.) Thus we can remove all gates outside the light cone yet the remaining circuit still generates $\ket{\psi}$.

As each node $v_{i+1}$ in the graph $G$ connects to at most 2 nodes in the previous layer $L_i$, the region contains at most $O(n \cdot 2^D)$ nodes, thereby also at most $O(n \cdot 2^D)$ gates in the circuit. As each gate can be fully specified by $O(1)$ real parameters, the circuit has at most $O(n \cdot 2^D)$ parameters. If $D \le n -\log n - \omega(1)$, the number of parameters is strictly smaller than $2^n-1$, the dimension of unit sphere $S$ for all possible $\ket{\psi}$. Then the circuit as a map from the parameters to the generated state, has its image a measure-zero subset of $S$. Since there are only finitely many layouts in a $D$-layer circuit, the union of these images still has measure 0, thus almost all $n$-qubit quantum states cannot be generated by circuits of depth $n-\log n - \omega(1)$.
\end{proof}